\documentclass[11pt]{article}
\usepackage[a4paper,margin=2cm]{geometry}

\newcommand{\modifiedA}[1]{#1\xspace}
\newcommand{\markModificationA}{}
\newcommand{\modifiedB}[1]{#1\xspace}

\newcommand{\modified}[1]{#1\xspace}

\usepackage{hyperref}
\usepackage{latexsym}
\usepackage{xspace}
\usepackage{amssymb}
\usepackage{stmaryrd}
\usepackage{url}
\input{xy}
\xyoption{all}
\usepackage{tikz}
\usepackage{color}
\usepackage{authblk}

\newtheorem{theorem}{Theorem}[section]
\newtheorem{lemma}[theorem]{Lemma}
\newtheorem{proposition}[theorem]{Proposition}
\newtheorem{corollary}[theorem]{Corollary}

\newtheorem{Definition}[theorem]{Definition}
\newtheorem{Example}[theorem]{Example}
\newtheorem{Remark}[theorem]{Remark}

\newenvironment{definition}{\begin{Definition}\begin{em}}{\end{em}\end{Definition}}
\newenvironment{example}{\begin{Example}\begin{em}}{\end{em}\end{Example}}
\newenvironment{remark}{\begin{Remark}\begin{em}}{\end{em}\end{Remark}}
\newenvironment{proof}{
	
	\smallskip
	
	\noindent
	{\em Proof.}}{

	\smallskip

	}

\def\eqref#1{(\ref{#1})}
\def\defeq{\stackrel{\mathrm{def}}{=}}
\def\tuple#1{\langle#1\rangle}

\newcommand{\mand}{\sqcap}
\newcommand{\mor}{\sqcup}
\newcommand{\V}{\forall}
\newcommand{\E}{\exists}

\def\ineg{\stackrel{.}{\neg}}

\newcommand{\PhiF}{\Phi_{\mathrm{full}}}
	
\newcommand{\mL}{\mathcal{L}}
\newcommand{\mLP}{$\mathcal{L}_\Phi$\xspace}
\newcommand{\mLPp}{$\mathcal{L}_\Phi^0$\xspace}

\newcommand{\mLPn}{$\mathcal{L}_{(\Phi,\ineg)}$\xspace}

\newcommand{\mLPUn}{$\mathcal{L}_{(\Phi \cup \{U\},\ineg)}$\xspace}

\newcommand{\DLP}{$\mathcal{L}_{(\Phi,\triangle)}$\xspace}
\newcommand{\DLPp}{$\mathcal{L}_{(\Phi,\triangle)}^0$\xspace}

\newcommand{\mT}{\mathcal{T}}

\newcommand{\mA}{\mathcal{A}}
\newcommand{\mI}{\mathcal{I}}
\newcommand{\mIp}{{\mathcal{I}'\!}}

\newcommand{\mZ}{\mathcal{Z}}

\newcommand{\ALCreg}{$\mathcal{ALC}_{reg}$\xspace}

\newcommand{\CN}{\mathbf{C}}
\newcommand{\RN}{\mathbf{R}}
\newcommand{\IN}{\mathbf{I}}

\newcommand{\NN}{\mathbb{N}}

\newcommand{\Self}{\mathtt{Self}}

\newcommand{\simP}{\sim_\Phi}
\newcommand{\simPI}{\sim_{\Phi,\mI}}
\newcommand{\simPn}{\stackrel{.}{\sim}_\Phi}
\newcommand{\simPIn}{\stackrel{.}{\sim}_{\Phi,\mI}}

\newcommand{\mIsimPn}{\mI/_{\simPn}}
\newcommand{\mIpsimPn}{\mIp/_{\simPn}}

\newcommand{\equivP}{\equiv_\Phi}
\newcommand{\equivPp}{\equiv_\Phi^0}

\newcommand{\equivPn}{\equiv_{(\Phi,\ineg)}}
\newcommand{\equivPdp}{\equiv_{(\Phi,\triangle)}^0}

\newcommand{\myend}{\mbox{}\hfill{\scriptsize$\blacksquare$}}

\newcommand{\comment}[1]{}

\newcommand{\fand}{\varotimes}
\newcommand{\fOr}{\varoplus}
\newcommand{\fneg}{\varominus}
\newcommand{\fto}{\Rightarrow}
\newcommand{\fequiv}{\Leftrightarrow}

\newcommand{\email}[1]{\mbox{Email: \url{#1}}}

\begin{document}
\sloppy
	
\title{Bisimulation and Bisimilarity for Fuzzy Description Logics\\ under the G{\"o}del Semantics\thanks{This is a revised and corrected version of the publication ``Bisimulation and bisimilarity for fuzzy description logics under the G\"odel semantics'', {\em Fuzzy Sets and Systems} 388: 146-178 (2020)}}

\author[1]{Linh Anh Nguyen}
\author[2]{Quang-Thuy Ha}
\author[3,4]{Ngoc-Thanh Nguyen}
\author[5,2]{Thi Hong Khanh Nguyen}
\author[6]{Thanh-Luong Tran}

\affil[1]{\small
	Institute of Informatics, University of Warsaw, 
	Banacha 2, 02-097 Warsaw, Poland, \email{nguyen@mimuw.edu.pl}
}

\affil[2]{\small
	Faculty of Information Technology, VNU University of Engineering and Technology, 
	144 Xuan Thuy, Hanoi, Vietnam, \email{thuyhq@vnu.edu.vn}
}

\affil[3]{\small
	Department of Information Systems,  
	Faculty of Computer Science and Management, 
	Wroclaw University of Science and Technology, Poland,
	\email{Ngoc-Thanh.Nguyen@pwr.edu.pl}
}

\affil[4]{\small
	Faculty of Information Technology, Nguyen Tat Thanh University, Ho Chi Minh City, Vietnam
}

\affil[5]{\small
	Faculty of Information Technology, Electricity Power University, 
	235 Hoang Quoc Viet, Hanoi, Vietnam, 
	\email{khanhnth@epu.edu.vn}
}

\affil[6]{\small
	Department of Information Technology, University of Sciences, Hue University, 
	77 Nguyen Hue, Hue city, Vietnam, \email{ttluong@hueuni.edu.vn} 
}

\date{}

\maketitle

\begin{abstract}
Description logics (DLs) are a suitable formalism for representing knowledge about domains in which objects are described not only by attributes but also by binary relations between objects. Fuzzy extensions of DLs can be used for such domains when data and knowledge about them are vague and imprecise. One of the possible ways to specify classes of objects in such domains is to use concepts in fuzzy DLs. As DLs are variants of modal logics, indiscernibility in DLs is characterized by bisimilarity. The bisimilarity relation of an interpretation is the largest auto-bisimulation of that interpretation. In DLs and their fuzzy extensions, such equivalence relations can be used for concept learning.  
In this paper, we define and study fuzzy bisimulation and bisimilarity for fuzzy DLs under the G\"odel semantics, as well as crisp bisimulation and strong bisimilarity for such logics extended with involutive negation. 
The considered logics are fuzzy extensions of the DL \ALCreg (a variant of PDL) with additional features among inverse roles, nominals, (qualified or unqualified) number restrictions, the universal role, local reflexivity of a role and involutive negation. 
We formulate and prove results on invariance of concepts under fuzzy (resp.\ crisp) bisimulation, conditional invariance of fuzzy TBoxex/ABoxes under bisimilarity (resp.\ strong bisimilarity), and the Hennessy-Milner property of fuzzy (resp.\ crisp) bisimulation for fuzzy DLs without (resp.\ with) involutive negation under the G\"odel semantics. Apart from these fundamental results, we also provide results on using fuzzy bisimulation to separate the expressive powers of fuzzy DLs, as well as results on using strong bisimilarity to minimize fuzzy interpretations. 
\end{abstract}



\section{Introduction}
\label{section:intro}

In traditional machine learning, objects are usually described by attributes, and a class of objects can be specified, among others, by a logical formula using attributes. Decision trees and rule-based classifiers are variants of classifiers based on logical formulas. To construct a classifier, one can restrict to using a sublanguage that allows only essential attributes and certain forms of formulas. If two objects are indiscernible w.r.t.\ that sublanguage, then they belong to the same decision class. Indiscernibility is an equivalence relation that partitions the domain into equivalence classes, and each decision class is the union of some of those equivalence classes. 

There are domains in which objects are described not only by attributes but also by binary relations between objects. Examples include social networks and linked data. For such domains, description logics (DLs) are a suitable formalism for representing knowledge about objects. Basic elements of DLs are concepts, roles and individuals (objects). A concept name is a unary predicate, a role name is binary predicate. A~concept is interpreted as a set of objects. It can be built from atomic concepts, atomic roles and named individuals (nominals) by using constructors. As DLs are variants of modal logics, indiscernibility in DLs is characterized by bisimilarity. The bisimilarity relation of an interpretation $\mI$ w.r.t.\ a logic language is the largest auto-bisimulation of $\mI$ w.r.t.\ that language. It has been exploited for concept learning in DLs (see, e.g., \cite{LbRoughification,TranNH15,DivroodiHNN18}). 

In practical applications, data and knowledge may be vague and imprecise, and fuzzy logics can be used to deal with them. There are different families of fuzzy operators. The G\"odel, {\L}ukasiewicz, Product and Zadeh families are the most popular ones. The first three of them use t-norm-based residua to define implication. The G\"odel and Zadeh families define conjunction and disjunction of truth values as their infimum and supremum, respectively. Extending DLs to fuzzy DLs, each family of fuzzy operators can be used to specify a semantics appropriately (see, e.g., \cite{BobilloDGS09}). 
Fuzzy DLs have attracted researchers for two decades (see \cite{BobilloCEGPS2015,BorgwardtP17b} for overviews and surveys). If objects are described by attributes and binary relations, and data and knowledge about them are vague and imprecise, then one of the possible ways to specify classes of objects is to use concepts in fuzzy DLs. Bisimilarity in fuzzy DLs can be used for learning such concepts. Thus, bisimilarity and bisimulation in fuzzy DLs are worth studying.  

The objective of this paper is to introduce and study bisimulation and bisimilarity for the classes of fuzzy DLs \mLP and \mLPn under the G\"odel semantics, where $\mL$ extends the DL \ALCreg (a variant of PDL, i.e., propositional dynamic logic) with fuzzy truth values, \mLP extends $\mL$ with the features from $\Phi \subseteq \{I$, $O$, $U$, $\Self$, $Q_n$, $N_n \mid$ $n \in \NN \setminus \{0\}\}$, which stand for inverse roles ($I$), nominals ($O$), the universal role ($U$), local reflexivity of a role ($\Self$), qualified number restrictions ($Q_n$) and unqualified number restrictions ($N_n$), respectively, with $n$ being the bound, and \mLPn extends \mLP with involutive negation. The DL \ALCreg allows PDL-like role constructors, which are union, sequential composition, reflexive-transitive closure and the test operator. 


\subsection{Research Problems and Our Results}

In this subsection, we present the research problems and introduce our most important results. 

\subsubsection{Fuzzy Bisimulation and Bisimilarity for Fuzzy DLs under the G{\"o}del Semantics}
\label{sssection: fuz-bis}

Given a fuzzy DL \mLP under the G{\"o}del semantics and fuzzy interpretations $\mI$ and $\mIp$, a function \mbox{$Z : \Delta^\mI \times \Delta^\mIp \to [0,1]$}, where $\Delta^\mI$ and $\Delta^\mIp$ are the domains of $\mI$ and $\mIp$, is called a fuzzy $\Phi$-bisimulation between $\mI$ and $\mIp$ if it satisfies certain conditions, which are appropriately designed so that the following properties hold when some light restrictions on $\mI$ and $\mIp$ are assumed: 
\begin{itemize}
\item {\em Invariance of concepts:} For every $x \in \Delta^\mI$, $x' \in \Delta^\mIp$ and every concept $C$ of \mLP, 
\begin{equation}\label{eq: HDHWP}
Z(x,x') \leq (C^\mI(x) \fequiv C^\mIp(x')),
\end{equation}
where $\fequiv$ is the {\em G{\"o}del equivalence operator}.\footnote{The G{\"o}del equivalence operator is specified as follows: $(p \fequiv q) = 1$ if $p = q$, and $(p \fequiv q) = \min\{p,q\}$ otherwise.}

\item {\em The Hennessy-Milner property:} The function 
\begin{equation} \label{eq: HGDKW}
\lambda \tuple{x,x'} \in \Delta^\mI \times \Delta^\mIp.\inf\{C^\mI(x) \fequiv C^\mIp(x') \mid C \textrm{ is a concept of \mLP} \}
\end{equation}
is the {\em greatest} fuzzy $\Phi$-bisimulation between $\mI$ and $\mIp$.$\,$\modifiedB{\footnote{\modifiedB{The order is specified as follows. Given $Z_1,Z_2: \Delta^\mI \times \Delta^\mIp \to [0,1]$, $Z_2$ is {\em greater than or equal to} $Z_1$ if $Z_1(x,y) \leq Z_2(x,y)$ for all $\tuple{x,y} \in \Delta^\mI \times \Delta^\mIp$.}}}
\end{itemize}
The Hennessy-Milner property can be strengthened by replacing \mLP in~\eqref{eq: HGDKW} with a certain sublanguage \mLPp of~\mLP. As a restriction on $\mI$ and $\mIp$ that guarantees these two properties, one may require that both $\mI$ and $\mIp$ are {\em image-finite}.\footnote{All highlighted terms in this introduction will be specified later in the paper.} It can be weakened to make the assertions more general. 

In~\cite{Fan15} Fan introduced fuzzy bisimulations for some G\"odel modal logics, which are fuzzy modal logics using the G\"odel semantics.\modifiedA{\footnote{\modifiedA{Her notion of fuzzy bisimulation is discussed in Remark~\ref{remark: JHFKS}.}}} The logics considered in~\cite{Fan15} include the fuzzy monomodal logic $K$ and its extension with converse. She proved that fuzzy bisimulations for these logics have the two above mentioned properties for the case when $\mI$ and $\mIp$ are image-finite\footnote{called degree-finite in~\cite{Fan15} for the case when the considered logic allows converse}. 

Inspired by the results of~\cite{Fan15}, in this paper we define fuzzy bisimulations in a uniform way for the whole class of fuzzy DLs \mLP. We prove the invariance of concepts under such bisimulations and the Hennessy-Milner property of such bisimulations, using restrictions weaker than image-finiteness for $\mI$ and $\mIp$. Namely, we prove the former (resp.\ latter) property for the class of {\em witnessed} (resp.\ witnessed and {\em modally saturated}) interpretations. In addition, we also identify a (tight) sublanguage \mLPp of \mLP that can replace \mLP in \eqref{eq: HGDKW} to make the Hennessy-Milner property stronger. 

If there exists a fuzzy $\Phi$-bisimulation $Z$ between $\mI$ and $\mI'$ such that $Z(a^\mI,a^\mIp)=1$ for all named individuals $a$, then we say that $\mI$ and $\mI'$ are $\Phi$-bisimilar. 
A fuzzy TBox $\mT$ is said to be invariant under $\Phi$-bisimilarity if, for every fuzzy interpretations $\mI$ and $\mI'$ that are witnessed and $\Phi$-bisimilar to each other, $\mI \models \mT$ iff $\mI' \models \mT$. The notion of invariance of fuzzy ABoxes under $\Phi$-bisimilarity is defined analogously. 
In this paper, we also provide results on invariance of fuzzy TBoxes and fuzzy ABoxes under $\Phi$-bisimilarity.  

\subsubsection{Crisp Bisimulation and Strong Bisimilarity for Fuzzy DLs with Involutive Negation}

Given a fuzzy DL \mLP under the G{\"o}del semantics, let's define the notion of crisp $\Phi$-bisimulation by using the same conditions of fuzzy $\Phi$-bisimulation except that the range of $Z$ is $\{0,1\}$ (instead of $[0,1]$). Then, what properties do crisp $\Phi$-bisimulations have? The answer is related to the fuzzy DL \mLPn that extends \mLP with involutive negation, which we denote by $\ineg$. The G{\"o}del negation is denoted by $\lnot$ and is non-involutive, as \mbox{\(\lnot p \defeq \textrm{(if $p = 0$ then 1 else $0$)}\)} and $\lnot\lnot p \neq p$ for $0 < p < 1$, while $\ineg$ is involutive in the sense that \mbox{$\ineg\!p \defeq 1 - p$} and $\ineg\ineg\!p \equiv p$.  

In~\cite{Fan15} Fan also studied crisp bisimulations for the logic that extends the fuzzy monomodal logic $K$ with involutive negation and its further extension with converse. She formulated and proved a theorem on the Hennessy-Milner property of crisp bisimulations for these logics. It uses the Baaz projection operator $\triangle$ defined as \mbox{\(\triangle p \,=\, \lnot\!\ineg p \,=\, \textrm{(if $p = 1$ then 1 else 0)}\)}. 

Inspired by those results of Fan~\cite{Fan15}, extending our results mentioned in Section~\ref{sssection: fuz-bis} we formulate and prove theorems on crisp $\Phi$-bisimulations for the whole class of fuzzy DLs \mLPn. 
They are as follows:
\begin{itemize}
\item {\em Invariance of concepts:} If fuzzy interpretations $\mI$ and $\mI'$ are witnessed w.r.t.\ \mLPn 
and $Z$ is a crisp $\Phi$-bisimulation between $\mI$ and $\mIp$, then for every $x \in \Delta^\mI$, $x' \in \Delta^\mIp$ and every concept $C$ of \mLPn, 
$Z(x,x') = 1$ implies $C^\mI(x) = C^\mIp(x')$.

\item {\em The Hennessy-Milner property:} If fuzzy interpretations $\mI$ and $\mI'$ are witnessed and modally saturated w.r.t.~\DLPp and $Z : \Delta^\mI \times \Delta^\mIp \to \{0,1\}$ is the function specified by ($Z(x,x') = 1$ iff $C^\mI(x) = C^\mIp(x)$ for all concepts $C$ of \DLPp), then $Z$ is the greatest crisp $\Phi$-bisimulation between $\mI$ and~$\mI'$.
\end{itemize}
Here, \DLPp is a sublanguage of \mLPn that uses $\triangle$ instead of $\lnot$ and $\ineg$ and excludes certain constructors. Modal saturatedness w.r.t.\ \DLPp is also defined appropriately by us in order to make the Hennessy-Milner property as strong as possible.  

If there exists a crisp $\Phi$-bisimulation $Z$ between $\mI$ and $\mI'$ such that $Z(a^\mI,a^\mIp)=1$ for all named individuals~$a$, then we say that $\mI$ and $\mI'$ are strongly $\Phi$-bisimilar. 
In this paper, we also provide results on invariance of fuzzy TBoxes and fuzzy ABoxes under strong $\Phi$-bisimilarity. 

\subsubsection{Separating the Expressive Powers of Fuzzy DLs}

Bisimulations have been widely used to analyze the expressive powers of modal and description logics (see, e.g., \cite{BRV2001,BSDL-INS}). In this paper, as a special point that relies on fuzzy bisimulations, we prove that involutive negation and the Baaz projection operator cannot be expressed in fuzzy DLs under the G\"odel semantics by using the other constructors. In particular, for any $\Phi$ considered in this paper, the fuzzy DLs \mLPn and \DLP are strictly more expressive than the fuzzy DL \mLP in defining  concepts. We also provide similar results concerning the expressive power in defining fuzzy TBoxes or fuzzy ABoxes for some cases of~$\Phi$. 

\subsubsection{Minimizing Fuzzy Interpretations}

Given an equivalence relation on the domain of a structure, one can try to minimize the structure by grouping the elements that are in the same equivalence class. The resultant is usually called the quotient structure w.r.t.\ that equivalence relation. The problems are: how to specify the contents of the quotient structure and whether this latter is \modifiedB{a minimal structure} equivalent to the original one w.r.t.\ some basic aspects. Minimization is useful not only for saving memory but also for speeding up computations on the structure. 

Let $\simPIn$ denote the binary relation on $\Delta^\mI$ such that $x \simPIn x'$ iff $Z(x,x') = 1$, where $Z$ is the greatest crisp $\Phi$-bisimulation between $\mI$ and itself. The relation $\simPIn$ is called the strong $\Phi$-bisimilarity relation of $\mI$. It is an equivalence relation. 
In this paper, we introduce the quotient fuzzy interpretation of $\mI$ w.r.t.~$\simPIn$ for the case when $\Phi \subseteq \{I,O,U\}$, and prove that under some light assumptions it is a \modifiedB{minimal} fuzzy interpretation equivalent to~$\mI$ w.r.t.\ some aspects like validity of fuzzy axioms/assertions of \mLPn. 


\subsection{Related Work}

Bisimulation and bisimilarity arose from research on modal logic~\cite{vBenthem76,vBenthem83,vBenthem84} and state transition systems~\cite{Park81,HennessyM85}. Since then, they have been widely studied for variants and extensions of modal logic, including dynamic logic, temporal logic, hybrid logic and description logic (see, e.g., \cite{KurtoninaR99,Rijke00,BRV2001,ArecesBM01,GorankoOtto06,Sangiorgi09,GratieFM12,Piro2012,BSDL-INS,BSpDL}). They have been used for analyzing the expressive powers of the concerned logics, minimizing state transition systems, and concept learning in DLs~\cite{LbRoughification,TranNH15,DivroodiHNN18}.
Regarding bisimulation and bisimilarity formulated for fuzzy structures (including fuzzy transition systems, fuzzy automata, fuzzy Kripke models and fuzzy interpretations in DLs), apart from the already mentioned work~\cite{Fan15} of Fan, other most notable related works are~\cite{CaoCK11,CiricIDB12,EleftheriouKN12,Nguyen-TFS2019}. 

In~\cite{CaoCK11} Cao et al.\ studied (crisp) bisimulations for fuzzy transition systems (FTSs), which may be infinite. They gave three kinds of (crisp) bisimulation for FTSs, leading to two kinds of bisimilarity, which coincide when restricted to image-finite FTSs. The notion of image-finiteness can be generalized to being witnessed~\cite{Hajek05}. The first kind of bisimilarity introduced in~\cite{CaoCK11} is defined so that some properties can be proved without requiring the considered FTSs to be witnessed. The second kind of bisimilarity introduced in~\cite{CaoCK11} for FTSs is called strong bisimilarity. Cao et al.~\cite{CaoCK11} provided some results on composition operations, subsystems, quotients and homomorphisms of FTSs, which are 
related to bisimulation.

In~\cite{CiricIDB12} {\'C}iri{\'c} et al.\ introduced (fuzzy) bisimulations for fuzzy automata. Such a bisimulation is a fuzzy relation between the sets of states of the two considered automata. There are four kinds of bisimulation defined in \cite{CiricIDB12}: forward, backward, forward-backward and backward-forward. The first kind is the usual one that researchers would have in mind as the default. Backward bisimulations are a kind of forward bisimulations between reversed automata. The two remaining kinds are mixtures of forward simulation and backward simulation. Apart from a result on invariance of languages under bisimulations, other main results of~\cite{CiricIDB12} concern characterizations of bisimulations via factor fuzzy automata, which are similar to quotient structures but defined by using a fuzzy equivalence relation (instead of a crisp equivalence relation). 

In~\cite{EleftheriouKN12} Eleftheriou et al.\ presented (weak) bisimulation and bisimilarity for Heyting-valued modal logics and proved the Hennessy-Milner property of those notions. A Heyting-valued modal logic uses a Heyting algebra as the space of truth values. There is a close relationship between Heyting-valued modal logics and fuzzy modal logics under the G\"odel semantics~\cite{Fan15}, as every linear Heyting algebra is a G\"odel algebra~\cite{EleftheriouKN12} and every G\"odel algebra is a Heyting algebra with the Dummett condition~\cite{CattaneoCGK04}. As discussed by Fan in~\cite{Fan15}, there is a relationship between fuzzy bisimulations for G\"odel modal logics and weak bisimulations for Heyting-valued modal logics~\cite{EleftheriouKN12}, especially for the case when the underlying Heyting algebra is linear. 

In~\cite{Nguyen-TFS2019} Nguyen studied bisimilarity for fuzzy DLs under the Zadeh semantics. The logics studied in~\cite{Nguyen-TFS2019} are similar to the fuzzy DLs \mLP studied in the current paper except that:
\begin{itemize}
	\item among the fuzzy values from [0,1] only 0 ($\bot$) and 1 ($\top$) can be used to construct concepts,
	\item the feature $Q$ (qualified number restrictions using any bound) is considered instead of~$Q_n$ and~$N_n$, 
	\item and most importantly, the Zadeh semantics is used instead of the G\"odel semantics.  
\end{itemize}
Nguyen~\cite{Nguyen-TFS2019} defined bisimilarity for fuzzy DLs under the Zadeh semantics by using cut-based simulations, which are related to directed simulations~\cite{KurtoninaR97,BSDL-P-LOGCOM}. He provided results on preservation of information by such simulations, the Hennessy-Milner property of such simulations, and conditional invariance of fuzzy TBoxes/ABoxes under bisimilarity between witnessed interpretations, all for fuzzy DLs under the Zadeh semantics.

\modifiedA{It is also worth mentioning the work~\cite{LutzPW11} by Lutz et al.\ on characterizations of concepts and TBoxes w.r.t.\ first-order logic, like van Benthem's characterization of modal formulas as the bisimulation invariant fragment of first-order logic~\cite{GorankoOtto06}. The work~\cite{LutzPW11} is based on notions such as bisimulation, equisimulation, disjoint union and direct product. It also studies TBox rewritability. Extending our results on invariance of concepts and fuzzy TBoxes/ABoxes by relating them to fuzzy first-order logic in the style of~\cite{LutzPW11} would be interesting but is beyond the scope of the current paper.}


\subsection{Motivations}

Fuzzy transition systems, fuzzy automata and fuzzy Kripke models are structures not oriented towards modeling domains that use terminological knowledge to describe individuals. Although bisimulations have been formulated and studied for them~\cite{CaoCK11,CiricIDB12,EleftheriouKN12,Fan15}, it is desirable to extend, generalize or modify the notions of bisimulation to deal with fuzzy interpretations in DLs. The reason is that DLs have their own area of applications, with natural features, including number restrictions, TBoxes and ABoxes, which are not common in other formalisms. 

Fuzzy/crisp bisimulations have not been formulated and studied for fuzzy DLs under the G{\"o}dle semantics. Extending the notions of bisimulation formulated for G{\"o}dle modal logics~\cite{Fan15} to fuzzy DLs is not a trivial task, especially in coping with (qualified or unqualified) number restrictions. To deal with number restrictions, the approach of using relational composition as in~\cite{CiricIDB12,Fan15} for defining conditions of bisimulation is not suitable, and in the current paper we have to use ``elementary'' conditions to define bisimulations. Consequently, as demonstrated by Example~\ref{example: HFKSB}, our notion of fuzzy bisimulation is different in nature from the ones in \cite{CiricIDB12,Fan15} for non-witnessed structures. Restricting to simple logics like the ones studied in~\cite{Fan15}, this difference does not matter much, as invariance results and the Hennessy-Milner property are usually formulated and proved only for witnessed (or image-finite) structures. 
Our notion of fuzzy bisimulation is also different from the notion of (crisp) bisimulation defined in~\cite{CaoCK11} for fuzzy transition systems and the notion of cut-based (weak) bisimulation defined in~\cite{EleftheriouKN12} for Heyting-valued modal logics. The approach of~\cite{EleftheriouKN12} uses \modifiedA{a family} of crisp relations, where each of the \modifiedA{relations} is specified by a cut-value (see~\cite[Section~V]{Fan15} for a discussion).

Although bisimilarity for fuzzy DLs under the Zadeh semantics has been studied in~\cite{Nguyen-TFS2019}, the Zadeh semantics for fuzzy DLs is essentially different from the G{\"o}del semantics (see \cite[Section~VII]{Nguyen-TFS2019} for a discussion). Nguyen~\cite{Nguyen-TFS2019} justified that both fuzzy bisimulation and crisp bisimulation for fuzzy DLs under the Zadeh semantics seem undefinable. To define bisimilarity for fuzzy DLs under that semantics, he had to use cut-based simulations. 

The aim of this paper is twofold. First, we formulate fuzzy/crisp bisimulations and study their fundamental properties like invariance and the Hennessy-Milner property. Second, we study some of their possible applications. Regarding the first problem, in comparison with~\cite{Fan15}, the current paper makes a significant extension with several dimensions: 
\begin{itemize}
\item DLs are variants of multimodal logics (while the logics considered in~\cite{Fan15} are monomodal).

\item We study bisimulations in a uniform way for a large class of DLs, which allow PDL-like role constructors and features among inverse roles, (qualified/unqualified) number restrictions, nominals, the universal role and the concept constructor representing local reflexivity of a role.

\item Apart from invariance of concepts, we also study invariance of fuzzy TBoxes and fuzzy ABoxes.

\item We formulate and prove the Hennessy-Milner property for the class of witnessed and modally saturated interpretations, which is more general than the class of image-finite interpretations.
\end{itemize}
Regarding applications of bisimulation and bisimilarity, our study leads to new results on:
\begin{itemize}
\item separating the expressive powers of fuzzy DLs, 
\item minimizing fuzzy interpretations while preserving validity of fuzzy axioms/assertions.
\end{itemize}

Another potential application of bisimilarity in fuzzy DLs is concept learning for the domains in which individuals are described not only by fuzzy attributes but also by fuzzy relations between individuals. The point is that bisimilarity is a natural notion of indiscernibility for such domains. Our study provides theoretical results and forms a starting point for concept learning in fuzzy DLs under the G{\"o}del semantics. 

\subsection{The Structure of the Paper}

The remainder of this paper is structured as follows.  
In Section~\ref{section: prel}, we formally specify the considered fuzzy DLs and their G\"odel semantics. In Section~\ref{sec: fus-bis}, we define and study fuzzy bisimulations for those fuzzy DLs and bisimilarity relations based on such bisimulations. In Section~\ref{section: crisp bis}, we provide notions and results on crisp bisimulation and strong bisimilarity for fuzzy DLs extended with involutive negation or the Baaz projection operator. In Section~\ref{section: exp-powers}, we give results on separating the expressive powers of fuzzy DLs by using fuzzy bisimulations. In Section~\ref{section: minimization}, we provide results on using strong bisimilarity to minimize fuzzy interpretations while preserving validity of fuzzy axioms/assertions. Conclusions are given in Section~\ref{sec: conc}. 

This work revises and extends the conference papers~\cite{HaNNT18,cBSfDL2}, which do not contain proofs. Apart from proofs, the current paper extends~\cite{HaNNT18,cBSfDL2} with the results of Section~\ref{section: exp-powers} and Corollary~\ref{cor: HFJHW}. 


\section{Preliminaries}
\label{section: prel}

In this section, we recall the G\"odel fuzzy operators, fuzzy DLs under the G\"odel semantics and define related notions that are needed for this paper. 

\subsection{The G\"odel Fuzzy Operators}

The family of G\"odel fuzzy operators are defined as follows, where $p, q \in [0,1]$:
\begin{eqnarray*}
	p \fand q & = & \min\{p,q\} \\
	p \fOr q & = & \max\{p,q\} \\
	\fneg p & = & (\textrm{if $p = 0$ then 1 else 0}) \\
	(p \fto q) & = & (\textrm{if $p \leq q$ then 1 else $q$}) \\
	(p \fequiv q) & = & (p \fto q) \fand (q \fto p). 
\end{eqnarray*}

Note that $\fand$ and $\fOr$ are associative and commutative. 
Also note that $(p \fequiv q) = 1$ if $p = q$, and $(p \fequiv q) = \min\{p,q\}$ otherwise. Clearly, $\fequiv$ is commutative. 
Assume that the decreasing order of the fuzzy operators w.r.t.\ the binding strength is $\fneg$, $\fand$, $\fOr$, $\fto$, $\fequiv$. 
The following lemma can easily be checked. 

\begin{lemma}
The following assertions hold for all $x,x',y,y',z \in [0,1]$.
\begin{eqnarray}
x \leq x' \textrm{ and } y \leq y' & \!\!\textrm{implies}\!\! & x \fand y \leq x' \fand y' \label{fop: HSDJW 1}\\
x' \leq x \textrm{ and } y \leq y' & \!\!\textrm{implies}\!\! & x \fto y \leq x' \fto y' \label{fop: HSDJW 2}\\
x \fand y \leq z & \!\textrm{iff}\! & x \leq y \fto z \label{fop: HSDJW 3} \\
x \fand (y \fequiv z) & \leq & y \fequiv x \fand z \label{fop: HSDJW 4} \\
x \fand (y \fequiv z) & \leq & (x \fto y) \fto z \label{fop: HSDJW 4b} \\
x \fto (y \fequiv z) & \leq & x \fand y \fto z \label{fop: HSDJW 5} \\
x \fto (y \fto z) & \leq & y \fto (x \fto z) \label{fop: HSDJW 5b} \\
(x \fequiv x') \fand (y \fequiv y') & \leq & x \fand y \fequiv x' \fand y'. \label{fop: HSDJW 6}
\end{eqnarray}
\end{lemma}

For a finite set $\Gamma = \{p_1,\ldots,p_n\} \subset [0,1]$ with $n \geq 0$, we define: 
\begin{eqnarray*}
\modifiedA{\textstyle\bigotimes\Gamma} & = & p_1 \fand \cdots \fand p_n \fand 1 \\
\modifiedA{\textstyle\bigoplus\Gamma} & = & p_1 \fOr \cdots \fOr p_n \fOr 0.
\end{eqnarray*}

Given $R: \Delta \times \Delta' \to [0,1]$, the {\em inverse} $R^-$ of $R$ is the function of type $\Delta' \times \Delta \to [0,1]$ such that $R^-(x,y) = R(y,x)$ for all $x \in \Delta'$ and $y \in \Delta$. 
Given $R: \Delta \times \Delta' \to [0,1]$ and $S: \Delta' \times \Delta'' \to [0,1]$, the {\em composition} $R \circ S$ is the function of type $\Delta \times \Delta'' \to [0,1]$ defined as follows:
\[ (R \circ S)(x,y) = \sup \{R(x,z) \fand S(z,y) \mid z \in \Delta' \}. \]

Given $R,S: \Delta \times \Delta' \to [0,1]$, if $R(x,y) \leq S(x,y)$ for all $\tuple{x,y} \in \Delta \times \Delta'$, then we write $R \leq S$ and say that $S$ is {\em greater than or equal to} $R$. 
%
%
If $\mZ$ is a set of functions of type $\Delta \times \Delta' \to [0,1]$, then by $\sup\mZ$ we denote the function of the same type defined as follows:
\[ (\sup\mZ)(x,y) = \sup\{Z(x,y) \mid Z \in \mZ\}. \]


\subsection{Fuzzy Description Logics under the G\"odel Semantics}

By $\Phi$ we denote a set of symbols among $I$, $O$, $U$, $\Self$, $Q_n$ and $N_n$ (with $n \in \NN \setminus \{0\}$), which stand for inverse roles, nominals, the universal role, local reflexivity of a role, qualified number restrictions and unqualified number restrictions, respectively, with $n$ being the bound used in the number restriction. In this subsection, we first define the syntax of roles and concepts in the fuzzy DL \mLP, where $\mL$ extends the DL \ALCreg with fuzzy truth values and \mLP extends $\mL$ with the features from $\Phi$. We then define fuzzy interpretations and the G\"odel semantics of \mLP.

Our logic language uses a set $\CN$ of {\em concept names}, a set $\RN$ of role names, and a set $\IN$ of individual names. 
%
A {\em basic role} w.r.t.~$\Phi$ is either a role name or the inverse $r^-$ of a role name~$r$ (when $I \in \Phi$).

{\em Roles} and {\em concepts} of \mLP are defined as follows:
\begin{itemize}
	\item if $r \in \RN$, then $r$ is a role of \mLP, 
	\item if $R$, $S$ are roles of \mLP and $C$ is a concept of \mLP, then $R \circ S$, $R \mor S$, $R^*$ and $C?$ are roles of \mLP,
	\item if $I \in \Phi$ and $R$ is a role of \mLP, then $R^-$ is a role of \mLP,
	\item if $U \in \Phi$, then $U$ is a role of \mLP, called the {\em universal role} (we assume that $U \notin \RN$), 
	
	
	\item if $p \in [0,1]$, then $p$ is a concept of \mLP,
	\item if $A \in \CN$, then $A$ is a concept of \mLP,
	\item if $C$, $D$ are concepts of \mLP and $R$ is a role of \mLP, then:
	\begin{itemize}
		\item $C \mand D$, $C \to D$, $\lnot C$, $C \mor D$, $\V R.C$, $\E R.C$ are concepts of \mLP, 
		\item if $O \in \Phi$ and $a \in \IN$, then $\{a\}$ is a concept of \mLP,
		\item if $\Self \in \Phi$ and $r \in \RN$, then $\E r.\Self$ is a~concept of \mLP, 
		\item if $Q_n \in \Phi$ and $R$ is a basic role w.r.t.~$\Phi$, then $\geq\! n\,R.C$ and $<\!n\,R.C$ are concepts of \mLP, 
		\item if $N_n \in \Phi$ and $R$ is a basic role w.r.t.~$\Phi$, then $\geq\! n\,R$ and $<\!n\,R$ are concepts of \mLP. 
	\end{itemize}
\end{itemize}

\modifiedA{We also use $\bot$ to denote the concept $0$ and $\top$ to denote the concept $1$.}

By \mLPp we denote the largest sublanguage of \mLP that disallows the role constructors $R \circ S$, $R \mor S$, $R^*$, $C?$ and the concept constructors $\lnot C$, $C \mor D$, $\V R.C$, $<\!n\,R.C$, $<\!n\,R$ 
\modifiedA{and, in the case when $\Phi \cap \{Q_n \mid n \in \NN \setminus \{0\}\} = \emptyset$, uses $\to$ only in the form $C \to p$ or $p \to C$, where $p \in [0,1]$.}

\newcommand{\hasParent}{\mathit{hasParent}}
\newcommand{\Male}{\mathit{Male}}
\newcommand{\Female}{\mathit{Female}}
\newcommand{\confucius}{\mathit{confucius}}

\modifiedA{The role constructor $C?$ is called the test operator. 
We give below an example of a concept with the test operator, where $R^+$ denotes $R \circ R^*$:
\[ 
	\E (\modifiedB{\Male? \circ \hasParent})^+.\{\confucius\}
\]
This concept represents the set of \modifiedB{descendants of Confucius in the male line}. 
The formal semantics of concepts is specified in Definition~\ref{def: f-i}. 
}

\newcommand{\Person}{\mathit{Person}}
\newcommand{\Group}{\mathit{Group}}
\newcommand{\Post}{\mathit{Post}}
\newcommand{\Hobby}{\mathit{Hobby}}
\newcommand{\Topic}{\mathit{Topic}}

\newcommand{\hasCloseFriend}{\mathit{hasCloseFriend}}
\newcommand{\posts}{\mathit{posts}}
\newcommand{\postedBy}{\mathit{postedBy}}
\newcommand{\likes}{\mathit{likes}}
\newcommand{\likedBy}{\mathit{likedBy}}
\newcommand{\shares}{\mathit{shares}}
\newcommand{\sharedBy}{\mathit{sharedBy}}
\newcommand{\relatedTo}{\mathit{relatedTo}}
\newcommand{\interestedIn}{\mathit{interestedIn}}
\newcommand{\isMemberOf}{\mathit{isMemberOf}}
\newcommand{\hasMember}{\mathit{hasMember}}

\newcommand{\travelling}{\mathit{traveling}}
\newcommand{\shopping}{\mathit{shopping}}
\newcommand{\camping}{\mathit{camping}}
\newcommand{\fashion}{\mathit{fashion}}
\newcommand{\politics}{\mathit{politics}}
\newcommand{\arts}{\mathit{arts}}
\newcommand{\movies}{\mathit{movies}}

\begin{example}\label{example: JHDKS}
We can represent analytical data and knowledge about a social network by using:
\begin{itemize}
\item concept names: $\Person$, $\Male$, $\Female$, $\Group$, $\Post$, $\Hobby$, $\Topic$, \ldots
\item role names: 
	\begin{itemize}
	\item $\hasCloseFriend$ (a person has another person as a close friend), 
	\item $\posts$ (a person created a post), 
	\item $\postedBy$ (a post was created by a person), 
	\item $\likes$ (a person likes a post), 
	\item $\likedBy$ (a post is liked by a person), 
	\item $\shares$ (a person shares a post), 
	\item $\sharedBy$ (a post is shared by a person), 
	\item $\relatedTo$ (a hobby or a topic is related to another one), 
	\item $\interestedIn$ (a person has a hobby or is interested in a topic), 
	\item $\isMemberOf$ (a person is a member of a group), 
	\item $\hasMember$ (a group has a person as a member), \ldots
	\end{itemize}
	
\item individual names: 
	\begin{itemize}
	\item $\travelling$, $\shopping$, $\camping$, \ldots (as hobbies), 
	\item $\fashion$, $\politics$, $\arts$, $\movies$, \ldots (as topics), \ldots	
	\end{itemize}
\end{itemize}
The role names $\hasCloseFriend$, $\relatedTo$ and $\interestedIn$ are fuzzy predicates (i.e., may be graded). 
Assume that the constructors like $\E R.C$ or $\geq\!n\,R.C$ have a greater binding strength than the constructors $\mand$, $\mor$ and $\to$. 
Here are examples of complex concepts: 
\begin{itemize}
\item $\geq\!3\,\shares.(\Post \mand \E\relatedTo.\{\fashion\})$: this represents the fuzzy set of people who share at least 3 posts related to fashion,
\item $\E\interestedIn.\{\fashion\}\, \mand \geq\!5\,\hasCloseFriend.\E\interestedIn.\{\shopping\}$: this represents the fuzzy set of people interested in fashion and having at least 5 close friends who are interested in shopping, 
\item $0.5 \to \E\interestedIn.\{\camping\}$: roughly speaking, this concept stands for the fuzzy set of people interested in camping to a degree greater than or equal to~0.5, but we should understand it by taking into account the meaning of the G\"odel implication.
\myend  
\end{itemize}
\end{example}

We use letters $A$ and $B$ to denote {\em atomic concepts} (which are concept names), 
$C$ and $D$ to denote arbitrary concepts, 
$r$ and $s$ to denote {\em atomic roles} (which are role names), 
$R$ and $S$ to denote arbitrary roles, 
$a$ and $b$ to denote individual names. 

Given a finite set $\Gamma = \{C_1,\ldots,C_n\}$ of concepts, \modifiedA{we define:}
\begin{eqnarray*}
\modifiedA{\textstyle\bigsqcap\Gamma} & \modifiedA{=} & \modifiedA{C_1 \mand \ldots \mand C_n \mand 1,} \\
\modifiedA{\textstyle\bigsqcup\Gamma} & \modifiedA{=} & \modifiedA{C_1 \mor \ldots \mor C_n \mor 0.}
\end{eqnarray*}

\begin{definition}\label{def: f-i}
A {\em (fuzzy) interpretation} is a pair $\mI = \langle \Delta^\mI, \cdot^\mI \rangle$, where $\Delta^\mI$ is a~non-empty set, called the {\em domain}, and $\cdot^\mI$ is the {\em interpretation function}, which maps every individual name $a$ to an element $a^\mI \in \Delta^\mI$, every concept name $A$ to a function $A^\mI : \Delta^\mI \to [0,1]$, and every role name $r$ to a function \mbox{$r^\mI : \Delta^\mI \times \Delta^\mI \to [0,1]$}. 
The function $\cdot^\mI$ is extended to complex roles and concepts as follows~(cf.~\cite{BobilloDGS09}), where the extrema are taken in the complete lattice $[0, 1]$:
\begin{eqnarray*}
		U^\mI(x,y) & = & 1 \\
		(R^-)^\mI(x,y) & = & R^\mI(y,x) \\
		(C?)^\mI(x,y) & = & \textrm{(if $x = y$ then $C^\mI(x)$ else 0)} \\
		(R \circ S)^\mI(x,y) & = & \sup\{R^\mI(x,z) \fand S^\mI(z,y) \mid z \in \Delta^\mI \} \\
		(R \mor S)^\mI(x,y) & = & R^\mI(x,y) \fOr S^\mI(x,y) \\
		(R^*)^\mI(x,y) & = & \sup \{\modifiedA{\textstyle\bigotimes}\{R^\mI(x_i,x_{i+1}) \mid 0 \leq i < n\} \mid 
		n \geq 0,\ x_0,\ldots,x_n \in \Delta^\mI,\ x_0 = x,\ x_n = y\} \\[1ex]
		p^\mI(x) & = & p \\
		\{a\}^\mI(x) & = & \textrm{(if $x = a^\mI$ then 1 else 0)}\\
		(\lnot C)^\mI(x) & = & \fneg C^\mI(x) \\
		(C \mand D)^\mI(x) & = & C^\mI(x) \fand D^\mI(x) \\
		(C \mor D)^\mI(x) & = & C^\mI(x) \fOr D^\mI(x) \\
		(C \to D)^\mI(x) & = & (C^\mI(x) \fto D^\mI(x)) \\
		(\E r.\Self)^\mI(x) & = & r^\mI(x,x) \\
		(\E R.C)^\mI(x) & = & \sup \{R^\mI(x,y) \fand C^\mI(y) \mid y \in \Delta^\mI\} \\
		(\V R.C)^\mI(x) & = & \inf \{R^\mI(x,y) \fto C^\mI(y) \mid y \in \Delta^\mI\} \\
		(\geq n\,R.C)^\mI(x) & = & \sup \{\modifiedA{\textstyle\bigotimes}\{R^\mI(x,y_i) \fand C^\mI(y_i) \mid 1 \leq i \leq n\} \mid 
		y_1,\ldots,y_n \in \Delta^\mI,\ y_i \neq y_j \textrm{ if } i \neq j\} \\
		(< n\,R.C)^\mI(x) & = & \inf \{(\modifiedA{\textstyle\bigotimes}\{R^\mI(x,y_i) \fand C^\mI(y_i) \mid 1 \leq i \leq n\} \fto \\
		& & \quad\;\;\;\; \modifiedA{\textstyle\bigoplus}\{y_j = y_k \mid 1 \leq j < k \leq n\}) \mid y_1,\ldots,y_n \in \Delta^\mI\} \\ 
		(\geq n\,R)^\mI(x) & = & (\geq n\,R.1)^\mI(x) \\
		(< n\,R)^\mI(x) & = & (< n\,R.1)^\mI(x).
\end{eqnarray*}

\vspace{-1.8em}
		
\myend
\end{definition}

For definitions of the Zadeh, {\L}ukasiewicz and Product semantics for fuzzy DLs, we refer the reader to~\cite{BobilloDGS09}. 

\modifiedA{An interpretation $\mI$ is {\em crisp} if $\{0,1\}$ is the range of the functions $A^\mI$ and $r^\mI$ for all $A \in \CN$ and $r \in \RN$.} 

\begin{remark}\label{remark: JFLWB}
Observe that \mbox{$(<\! n\,R.C)^\mI(x)$} is either 1 or~0. Namely, \mbox{$(<\!n\,R.C)^\mI(x) = 1$} if, for every set $\{y_1$, \ldots, $y_n\}$ of $n$ pairwise distinct elements of $\Delta^\mI$, there exists \mbox{$1 \leq i \leq n$} such that \mbox{$R^\mI(x,y_i) \fand C^\mI(y_i) = 0$}. Otherwise, $(<\!n\,R.C)^\mI(x) = 0$. 
A similar observation holds for \mbox{$(<\! n\,R)^\mI(x)$}.\myend
\end{remark}

\begin{example}\label{example: HDKSL}
	Let $\RN = \{r\}$, $\CN = \{A\}$ and $\IN = \emptyset$. Consider the fuzzy interpretation $\mI$ illustrated and specified below:
	
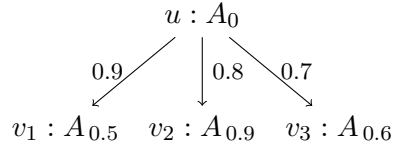
\begin{center}		
		\begin{tikzpicture}
		\node (I) {};
		\node (u) [node distance=0.2cm, below of=I] {$u:A_0$};
		\node (ub) [node distance=1.5cm, below of=u] {$v_2:A_{\,0.9}$};
		\node (v) [node distance=1.8cm, left of=ub] {$v_1:A_{\,0.5}$};
		\node (w) [node distance=1.8cm, right of=ub] {$v_3:A_{\,0.6}$};
		\draw[->] (u) to node [left]{\footnotesize{0.9}} (v);
		\draw[->] (u) to node [right]{\footnotesize{0.8}} (ub);
		\draw[->] (u) to node [right]{\footnotesize{0.7}} (w);
		\end{tikzpicture}
	\end{center}
	\begin{itemize}
		\item $\Delta^\mI = \{u,v_1,v_2,v_3\}$,
		\item $A^\mI(u) = 0$, $A^\mI(v_1) = 0.5$, $A^\mI(v_2) = 0.9$, $A^\mI(v_3) = 0.6$,
		\item $r^\mI(u,v_1) = 0.9$, $r^\mI(u,v_2) = 0.8$, $r^\mI(u,v_3) = 0.7$, 
		and $r^\mI(x,y) = 0$ for the other pairs $\tuple{x,y}$.
	\end{itemize}
	
	\noindent
	We have that:
	\begin{itemize}
		\item $(\V r.A)^\mI(u) = 0.5$, $(\E r.A)^\mI(u) = 0.8$,
		$(<\!2\,r.A)^\mI(u) = 0$, $(\geq\!2\,r.A)^\mI(u) = 0.6$,
		\item for $C = \V (r \mor r^-)^*.A$ and $1 \leq i \leq 3$: $C^\mI(v_i) = 0$, 
		\item for $C = \E (r \mor r^-)^*.A\,$:  
		$C^\mI(v_1) = 0.8$, $C^\mI(v_2) = 0.9$ and $C^\mI(v_3) = 0.7$.
		\myend
	\end{itemize}
\end{example}

A fuzzy interpretation $\mI$ is {\em witnessed w.r.t.\ \mLP}~(cf.~\cite{Hajek05}) if any infinite set under the infimum (resp.~supremum) operator in Definition~\ref{def: f-i} has \modifiedA{a} smallest (resp.\ biggest) element. 
The notion of being {\em witnessed w.r.t.\ \mLPp} is defined similarly under the assumption that only roles and concepts of \mLPp are allowed. 
A fuzzy interpretation $\mI$ is {\em finite} if $\Delta^\mI$, $\CN$, $\RN$ and $\IN$ are finite, and is {\em image-finite} w.r.t.\ $\Phi$ if, for every $x \in \Delta^\mI$ and every basic role $R$ w.r.t.~$\Phi$, $\{y \in \Delta^\mI \mid R^\mI(x,y) > 0\}$ is finite. 
Observe that every finite fuzzy interpretation is witnessed w.r.t.\ \mLP and, if $U \notin \Phi$, then every image-finite fuzzy interpretation w.r.t.\ $\Phi$ is witnessed w.r.t.~\mLPp.  

A {\em fuzzy assertion} in \mLP is an expression of the form $a \doteq b$, $a \not\doteq b$, $C(a) \bowtie p$ or $R(a,b) \bowtie p$, where $C$ is a concept of \mLP, $R$ is a role of \mLP, $\bowtie\ \in \{\geq, >, \leq, <\}$ and $p \in [0,1]$. A~{\em fuzzy ABox} in \mLP is a finite set of fuzzy assertions in \mLP. 

A {\em fuzzy GCI} (general concept inclusion) in \mLP is an expression of the form $(C \sqsubseteq D) \rhd p$, where $C$ and $D$ are concepts of \mLP, $\rhd \in \{\geq, > \}$ and $p \in (0,1]$. A {\em fuzzy TBox} in \mLP is a finite set of fuzzy GCIs in~\mLP. 

Given a fuzzy interpretation $\mI$ and a fuzzy assertion or GCI $\varphi$, we define the relation $\mI \models \varphi$ ($\mI$ {\em validates}~$\varphi$) as follows:
\[
\begin{array}{lcl}
\mI \models a \doteq b & \textrm{iff} & a^\mI = b^\mI, \\[0.5ex] 
\mI \models a \not\doteq b & \textrm{iff} & a^\mI \neq b^\mI, \\[0.5ex] 
\mI \models C(a) \bowtie p & \textrm{iff} & C^\mI(a^\mI) \bowtie p, \\[0.5ex] 
\mI \models R(a,b) \bowtie p & \textrm{iff} & R^\mI(a^\mI,b^\mI) \bowtie p, \\[0.5ex] 
\mI \models (C \sqsubseteq D) \rhd p & \textrm{iff} & (C \to D)^\mI(x) \rhd p \textrm{ for all } x \in \Delta^\mI.
\end{array}
\]

A fuzzy interpretation $\mI$ is a {\em model} of a fuzzy ABox $\mA$, denoted by $\mI \models \mA$, if $\mI \models \varphi$ for all $\varphi \in \mA$. Similarly, $\mI$ is a model of a fuzzy TBox $\mT$, denoted by $\mI \models \mT$, if $\mI \models \varphi$ for all $\varphi \in \mT$.

\begin{example}
Let's continue Example~\ref{example: JHDKS}. Suppose that $\mI$ is a fuzzy interpretation that validates the following fuzzy \modifiedA{GCIs}:
\begin{eqnarray}
\geq\!3\,\shares.(\Post \mand \E\relatedTo.\{\fashion\}) \sqsubseteq \E\interestedIn.\{\fashion\} & \geq & 0.5 \label{eq: HFJSK 1} \\
\E\interestedIn.\{\fashion\} \sqsubseteq \E\interestedIn.\{\shopping\} & \geq & 0.4 \label{eq: HFJSK 2} \\
(0.5 \to \E\interestedIn.\{\camping\}) \sqsubseteq \E\interestedIn.\{\travelling\} & \geq & 0.6 \label{eq: HFJSK 3} \\
\E\interestedIn.\{\camping\} \sqsubseteq \E\interestedIn.\{\travelling\} & \geq & 0.6 \label{eq: HFJSK 4} 
\end{eqnarray}
The fuzzy \modifiedA{GCI}~\eqref{eq: HFJSK 1} states that, for every $x \in \Delta^\mI$, 
\begin{itemize}
\item either $(\geq\!3\,\shares.(\Post \mand \E\relatedTo.\{\fashion\}))^\mI(x) \leq (\E\interestedIn.\{\fashion\})^\mI(x)$,
\item or $(\geq\!3\,\shares.(\Post \mand \E\relatedTo.\{\fashion\}))^\mI(x) > (\E\interestedIn.\{\fashion\})^\mI(x) \geq 0.5$. 
\end{itemize}
Similarly, the fuzzy \modifiedA{GCI}~\eqref{eq: HFJSK 2} states that, for every $x \in \Delta^\mI$, 
\begin{itemize}
\item either $(\E\interestedIn.\{\fashion\})^\mI(x) \leq (\E\interestedIn.\{\shopping\})^\mI(x)$,
\item or $(\E\interestedIn.\{\fashion\})^\mI(x) > (\E\interestedIn.\{\shopping\})^\mI(x) \geq 0.4$. 
\end{itemize}
Together they imply that, for every $x \in \Delta^\mI$,  
\begin{itemize}
\item either $(\E\interestedIn.\{\shopping\})^\mI(x) \geq (\geq\!3\,\shares.(\Post \mand \E\relatedTo.\{\fashion\}))^\mI(x)$,	
\item or $(\E\interestedIn.\{\shopping\})^\mI(x) \geq 0.4$. 
\end{itemize}
The fuzzy \modifiedA{GCI}~\eqref{eq: HFJSK 3} states that, for every $x \in \Delta^\mI$, 
\begin{itemize}
\item either $(\E\interestedIn.\{\camping\})^\mI(x) \geq 0.5$ and $(\E\interestedIn.\{\travelling\})^\mI(x) \geq 0.6$, 
\item or $(\E\interestedIn.\{\camping\})^\mI(x) \leq (\E\interestedIn.\{\travelling\})^\mI(x)$. 
\end{itemize}
The fuzzy \modifiedA{GCI}~\eqref{eq: HFJSK 4} states that, for every $x \in \Delta^\mI$, 
\begin{itemize}
\item either $(\E\interestedIn.\{\camping\})^\mI(x) \leq (\E\interestedIn.\{\travelling\})^\mI(x)$,  
\item or $(\E\interestedIn.\{\camping\})^\mI(x) > (\E\interestedIn.\{\travelling\})^\mI(x) \geq 0.6$. 
\end{itemize}
Thus, the fuzzy \modifiedA{GCI}~\eqref{eq: HFJSK 4} subsumes the fuzzy \modifiedA{GCI}~\eqref{eq: HFJSK 3}.
\myend
\end{example} 

Two concepts $C$ and $D$ are {\em equivalent}, denoted by $C \equiv D$, if $C^\mI = D^\mI$ for every fuzzy interpretation~$\mI$. 
Two roles $R$ and $S$ are {\em equivalent}, denoted by $R \equiv S$, if $R^\mI = S^\mI$ for every fuzzy interpretation~$\mI$. 

We say that a role $R$ is in {\em inverse normal form} if the inverse constructor is applied in $R$ only to role names. 
In this paper, we assume that roles are presented in inverse normal form because every role can be translated to an equivalent role in inverse normal form using the following rules:
\[
\begin{array}{rclrcl}
U^- & \equiv & U &
(R \circ S)^- & \equiv & S^- \circ R^-\\

(R^-)^- & \equiv & R &
(R \sqcup S)^- & \equiv & R^- \sqcup S^-\\

(C?)^- & \equiv & C?\qquad\qquad &
(R^*)^- & \equiv & (R^-)^*.
\end{array}
\]

\begin{remark}\label{remark: OFHSJ}
	The concept constructors $\lnot C$ and $C \mor D$ can be excluded from \mLP because	
	\begin{eqnarray*}
		\lnot C & \equiv & (C \to 0) \\
		C \mor D & \equiv & ((C \to D) \to D) \mand ((D \to C) \to C).
	\end{eqnarray*}
	
\vspace{-1.8em}

\myend
\end{remark}


\section{Fuzzy Bisimulations}
\label{sec: fus-bis}

In this section, we define and study fuzzy bisimulations for fuzzy DLs under the G{\"o}del semantics, as well as bisimilarity relations based on such bisimulations. 
\modifiedA{As mentioned in the Introduction, the notion of {\em fuzzy $\Phi$-bisimulation} (specified by Definition~\ref{def: DHGAK} given below) is designed to satisfy the invariance of concepts and the Hennessy-Milner property. The former property states that, if $Z : \Delta^\mI \times \Delta^\mIp \to [0,1]$ is a fuzzy $\Phi$-bisimulation between fuzzy interpretations $\mI$ and $\mIp$, then for every $x \in \Delta^\mI$, $x' \in \Delta^\mIp$ and every concept $C$ of \mLP, 
\mbox{$Z(x,x') \leq (C^\mI(x) \fequiv C^\mIp(x'))$}.}

\newcommand{\eqFBlast}{\eqref{eq: FB 7n}\xspace}

\begin{definition}\label{def: DHGAK}
	Let $\Phi \subseteq \{I,O,U,\Self,Q_n,N_n \mid n \in \NN \setminus \{0\}\}$ be a set of features and $\mI$, $\mI'$ fuzzy interpretations.
	A function $Z : \Delta^\mI \times \Delta^\mIp \to [0,1]$ is called a {\em fuzzy $\Phi$-bisimulation} (under the G\"odel semantics) between $\mI$ and $\mI'$ if the following conditions hold for every $x \in \Delta^\mI$, $x' \in \Delta^\mIp$, $A \in \CN$, $a \in \IN$, $r \in \RN$ and every basic role $R$ w.r.t.~$\Phi$:
	\begin{eqnarray}
		&& Z(x,x') \leq (A^\mI(x) \fequiv A^\mIp(x')) \label{eq: FB 2} \\[0.5ex]
		&& \V y \in \Delta^\mI\, \E y' \in \Delta^\mIp\ Z(x,x') \fand R^\mI(x,y) \leq Z(y,y') \fand R^\mIp(x',y') \label{eq: FB 3} \\[0.5ex]
		&& \V y' \in \Delta^\mIp\, \E y \in \Delta^\mI\ Z(x,x') \fand R^\mIp(x',y') \leq Z(y,y') \fand R^\mI(x,y); \label{eq: FB 4}
	\end{eqnarray}
	if $O \in \Phi$, then
	\begin{eqnarray}
	Z(x,x') \leq (x = a^\mI \Leftrightarrow x' = a^\mIp); \label{eq: FB 5}
	\end{eqnarray}
	if $U \in \Phi$, then
	\begin{eqnarray}
		&&\V y \in \Delta^\mI\ \E y' \in \Delta^\mIp\ Z(x,x') \leq Z(y,y') \label{eq: FB 8} \\
		&&\V y' \in \Delta^\mIp\ \E y \in \Delta^\mI\ Z(x,x') \leq Z(y,y'); \label{eq: FB 9}
	\end{eqnarray}
	if $\Self \in \Phi$, then
	\begin{eqnarray}
		Z(x,x') \leq (r^\mI(x,x) \fequiv r^\mIp(x',x')); \label{eq: FB 10}
	\end{eqnarray}
if $Q_n \in \Phi$, then 
\begin{eqnarray}
\parbox{14.5cm}{if $Z(x,x') > 0$ and $y_1,\ldots,y_n$ are pairwise distinct elements of $\Delta^\mI$ such that \mbox{$R^\mI(x,y_j) > 0$} for all $1 \leq j \leq n$, then there exist pairwise distinct elements $y'_1,\ldots,y'_n$ of $\Delta^\mIp$ such that, for every $1 \leq i \leq n$, there exists $1 \leq j \leq n$ such that $Z(x,x') \fand \modified{R^\mI(x,y_1) \fand\cdots\fand R^\mI(x,y_n)} \leq Z(y_j,y'_i) \fand R^\mIp(x',y'_i)$,} \label{eq: FB 6} \\[0.5em]
\parbox{14.5cm}{if $Z(x,x') > 0$ and $y'_1,\ldots,y'_n$ are pairwise distinct elements of $\Delta^\mIp$ such that \mbox{$R^\mIp(x',y'_j) > 0$} for all $1 \leq j \leq n$, then there exist pairwise distinct elements $y_1,\ldots,y_n$ of $\Delta^\mI$ such that, for every $1 \leq i \leq n$, there exists $1 \leq j \leq n$ such that $Z(x,x') \fand \modified{R^\mIp(x',y'_1) \fand\cdots\fand R^\mIp(x',y'_n)} \leq Z(y_i,y'_j) \fand R^\mI(x,y_i)$;} \label{eq: FB 7}
\end{eqnarray}
if $N_n \in \Phi$, then 
\begin{eqnarray}
\parbox{14.5cm}{if $Z(x,x') > 0$ and $y_1,\ldots,y_n$ are pairwise distinct elements of $\Delta^\mI$ such that \mbox{$R^\mI(x,y_j) > 0$} for all $1 \leq j \leq n$, then there exist pairwise distinct elements $y'_1,\ldots,y'_n$ of $\Delta^\mIp$ such that, for every $1 \leq i \leq n$, \modified{$Z(x,x') \fand R^\mI(x,y_1) \fand\cdots\fand R^\mI(x,y_n) \leq R^\mIp(x',y'_i)$},} \label{eq: FB 6n} \\[0.5em]
\parbox{14.5cm}{if $Z(x,x') > 0$ and $y'_1,\ldots,y'_n$ are pairwise distinct elements of $\Delta^\mIp$ such that \mbox{$R^\mIp(x',y'_j) > 0$} for all $1 \leq j \leq n$, then there exist pairwise distinct elements $y_1,\ldots,y_n$ of $\Delta^\mI$ such that, for every $1 \leq i \leq n$, \modified{$Z(x,x') \fand R^\mIp(x',y'_1) \fand\cdots\fand R^\mIp(x',y'_n) \leq R^\mI(x,y_i)$}.} \label{eq: FB 7n}
\end{eqnarray}
For example, if $\Phi = \{I,Q_2\}$, then only Conditions \eqref{eq: FB 2}-\eqref{eq: FB 4}, \eqref{eq: FB 6} and \eqref{eq: FB 7} with $n = 2$ are essential.
\myend
\end{definition}

Conditions~\eqref{eq: FB 6}, \eqref{eq: FB 7}, \eqref{eq: FB 6n} and \eqref{eq: FB 7n} (for the cases with $Q_n$ and $N_n$) contain corrections w.r.t.~\cite{fss/NguyenHNNT20}. 

\begin{example}\label{example: HDJAA 2}
	Let $\RN = \{r\}$, $\CN = \{A\}$, $\IN = \emptyset$ and $\Phi = \emptyset$. Consider the fuzzy interpretations $\mI$ and $\mI'$ illustrated below (and specified similarly as in Example~\ref{example: HDKSL}).
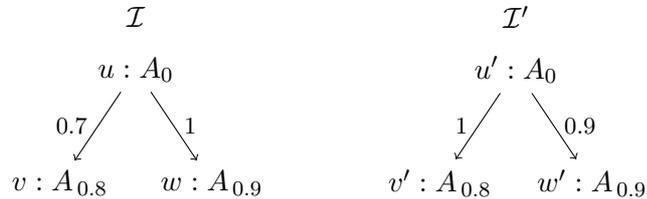
	\begin{center}
	\begin{tikzpicture}
	\node (x0) {};
	\node (x) [node distance=3.0cm, right of=x0] {};
	\node (I) [node distance=0.0cm, below of=x] {$\mI$};
	\node (I1) [node distance=5.0cm, right of=I] {$\mI'$};
	\node (u) [node distance=0.7cm, below of=I] {$u:A_0$};
	\node (ub) [node distance=1.5cm, below of=u] {};
	\node (v) [node distance=1.0cm, left of=ub] {$v:A_{\,0.8}$};
	\node (w) [node distance=1.0cm, right of=ub] {$w:A_{\,0.9}$};
	\node (up) [node distance=0.7cm, below of=I1] {$u':A_0$};
	\node (ubp) [node distance=1.5cm, below of=up] {};
	\node (vp) [node distance=1.0cm, left of=ubp] {$v':A_{\,0.8}$};
	\node (wp) [node distance=1.0cm, right of=ubp] {$w':A_{\,0.9}$};
	\draw[->] (u) to node [left]{\footnotesize{0.7}} (v);
	\draw[->] (u) to node [right]{\footnotesize{1}} (w);
	\draw[->] (up) to node [left]{\footnotesize{1}} (vp);
	\draw[->] (up) to node [right]{\footnotesize{0.9}} (wp);
	\end{tikzpicture}
	\end{center}	
	
	
	\noindent
	If $Z$ is a fuzzy $\Phi$-bisimulation between $\mI$ and $\mI'$, then: 
	\begin{itemize}
		\item $Z(v,w') \leq 0.8$ and $Z(w,v') \leq 0.8$ due to~\eqref{eq: FB 2}, 
		\item $Z(u,u') \leq 0.8$ due to~\eqref{eq: FB 4} for $x = u$, $x' = u'$ and $y' = v'$, 
		\item $Z(u,v') = Z(u,w') = Z(v,u') = Z(w,u') = 0$ due to~\eqref{eq: FB 2}. 
	\end{itemize}
	It can be \modifiedA{checked} that the function $Z : \Delta^\mI \times \Delta^\mIp \to [0,1]$ specified by
	\begin{itemize}
		\item $Z(v,v') = Z(w,w') = 1$, 
		\item $Z(v,w') = Z(w,v') = Z(u,u') = 0.8$,
		\item $Z(u,v') = Z(u,w') = Z(v,u') = Z(w,u') = 0$
	\end{itemize}
	is a fuzzy $\Phi$-bisimulation between $\mI$ and $\mI'$, and hence is the greatest fuzzy $\Phi$-bisimulation between $\mI$ and~$\mI'$.  
\myend
\end{example}

\begin{proposition}\label{prop: HFHSJ}
	Let $\mI$, $\mIp$ and $\mI''$ be fuzzy interpretations.
	\begin{enumerate}
		\item The function $Z : \Delta^\mI \times \Delta^\mI \to [0,1]$ specified by \[ Z(x,x') = (\textrm{if $x = x'$ then 1 else 0}) \] is a fuzzy $\Phi$-bisimulation between $\mI$ and itself.
		\item If $Z$ is a fuzzy $\Phi$-bisimulation between $\mI$ and $\mIp$, then $Z^-$ is a fuzzy $\Phi$-bisimulation between $\mIp$ and~$\mI$.
		\item If $Z_1$ is a fuzzy $\Phi$-bisimulation between $\mI$ and $\mIp$, and $Z_2$ is a fuzzy $\Phi$-bisimulation between $\mIp$ and $\mI''$, then $Z_1 \circ Z_2$ is a fuzzy $\Phi$-bisimulation between $\mI$ and $\mI''$.
		\item If $\mZ$ is a finite set of fuzzy $\Phi$-bisimulations between $\mI$ and $\mIp$, then $\sup\mZ$ is also a fuzzy $\Phi$-bisimulation between $\mI$ and $\mIp$.
	\end{enumerate}   
\end{proposition}

The proof of this proposition is straightforward.

\begin{remark}
	It seems that the assertion~4 of Proposition~\ref{prop: HFHSJ} cannot be strengthened by allowing $\mZ$ to be infinite. So, the greatest fuzzy $\Phi$-bisimulation between $\mI$ and $\mIp$ may not exist. 
	As stated later by Theorem~\ref{theorem: fG H-M}, if $\mI$ and $\mI'$ are witnessed and modally saturated w.r.t.~\mLPp (see Definition~\ref{def: modally saturated}), then the greatest fuzzy $\Phi$-bisimulation between $\mI$ and $\mIp$ exists.
\myend
\end{remark} 

Let $\mI$ and $\mIp$ be fuzzy interpretations. For $x \in \Delta^\mI$ and $x' \in \Delta^\mIp$, we write $x \simP x'$ to denote that there exists a fuzzy  $\Phi$-bisimulation $Z$ between $\mI$ and $\mI'$ such that $Z(x,x') = 1$. If $x \simP x'$, then we say that $x$ and $x'$ are {\em $\Phi$-bisimilar}. 
Let $\simPI$ be the binary relation on $\Delta^\mI$ such that, for $x,x' \in \Delta^\mI$, \mbox{$x \simPI x'$} iff $x \simP x'$. By Proposition~\ref{prop: HFHSJ}, $\simPI$ is an equivalence relation. We call it the {\em $\Phi$-bisimilarity} relation of~$\mI$. 
If $\IN \neq \emptyset$ and there exists a fuzzy $\Phi$-bisimulation $Z$ between $\mI$ and $\mI'$ such that $Z(a^\mI,a^\mIp)=1$ for all $a \in \IN$, then we say that $\mI$ and $\mI'$ are {\em $\Phi$-bisimilar} and write $\mI \simP \mIp$. 

{\markModificationA
\begin{remark}\label{remark: JHFKS}
As mentioned earlier, in~\cite{Fan15} Fan introduced fuzzy bisimulations for the fuzzy monomodal logic $K$ and its extension with converse under the G\"odel semantics. Her definition of bisimulations when reformulated for \mLP with $\Phi \subseteq \{I\}$ can be stated as follows. 
Given $\Phi = \emptyset$ or $\Phi = \{I\}$, a function \mbox{$Z : \Delta^\mI \times \Delta^\mIp \to [0,1]$} is a {\em fuzzy $\Phi$-bisimulation} between fuzzy interpretations $\mI$ and $\mI'$ if the following conditions hold for every $x \in \Delta^\mI$, every $x' \in \Delta^\mIp$ and every basic role $R$ w.r.t.~$\Phi$:
\begin{eqnarray}
Z(x,x') & \leq & \inf \{A^\mI(x) \fequiv A^\mIp(x') \mid A \in \CN\} \label{eq: Fan FB 2} \\[0.5ex]
Z^- \circ R^\mI & \leq & R^\mIp \circ Z^- \label{eq: Fan FB 3} \\[0.5ex]
Z \circ R^\mIp & \leq & R^\mI \circ Z. \label{eq: Fan FB 4}
\end{eqnarray}
 
Conditions~\eqref{eq: Fan FB 2}--\eqref{eq: Fan FB 4} correspond to Conditions~\eqref{eq: FB 2}--\eqref{eq: FB 4}, respectively. In particular, if $Z$ is a fuzzy $\Phi$-bisimulation between $\mI$ and $\mI'$ according to Definition~\ref{def: DHGAK}, then it is also a fuzzy $\Phi$-bisimulation between $\mI$ and $\mI'$ according to the modified definition. Conversely, when $\mI$ and $\mI'$ are image-finite w.r.t.~$\Phi$, every fuzzy $\Phi$-bisimulation between $\mI$ and $\mI'$ according to the modified definition is also a fuzzy $\Phi$-bisimulation between $\mI$ and $\mI'$ according to Definition~\ref{def: DHGAK}. Note that $\Phi$ is assumed here to be either $\emptyset$ or $\{I\}$.

For the case when $\Phi = \{I\}$, the above modified definition of fuzzy $\Phi$-bisimulations is equivalent to the one obtained from it by replacing \eqref{eq: Fan FB 3} and \eqref{eq: Fan FB 4} with the following, where $r$ is universally quantified over $\RN$:
\begin{eqnarray}
Z^- \circ r^\mI & = & r^\mIp \circ Z^- \label{eq: FanI FB 3} \\[0.5ex]
Z \circ r^\mIp & = & r^\mI \circ Z. \label{eq: FanI FB 4}
\end{eqnarray} 
This latter form reflects Fan's definition given in~\cite{Fan15}. 

Observe that Conditions~\eqref{eq: Fan FB 3}--\eqref{eq: FanI FB 4} are more compact than Conditions~\eqref{eq: FB 3} and \eqref{eq: FB 4}. However, it is hard to follow this style when extending the notion of fuzzy $\Phi$-bisimulation for dealing \modifiedB{with number restrictions} (i.e., when $\Phi$ contains $Q_n$ or $N_n$ for some $n$). 

\modifiedB{Under restrictions to image-finite structures and the logics considered in~\cite{Fan15}, our Theorems~\ref{theorem: UFNSJ}, \ref{theorem: fG H-M}, \ref{theorem: UFNSJ2} and~\ref{theorem: fG H-M 2} coincide with the results of~\cite{Fan15}.}
\myend
\end{remark}
}

\begin{example}\label{example: HFKSB}
	Let $\RN = \{r\}$, $\CN = \{A\}$, $\IN = \{a\}$ and $\Phi = \emptyset$. Consider the fuzzy interpretations $\mI$ and $\mI'$ illustrated below and specified similarly as in Example~\ref{example: HDKSL}, with $a^\mI = u$, $a^\mIp = u'$ and $\Delta^\mIp = \{u', v'_i \mid i \in \NN \setminus \{0\}\}$:
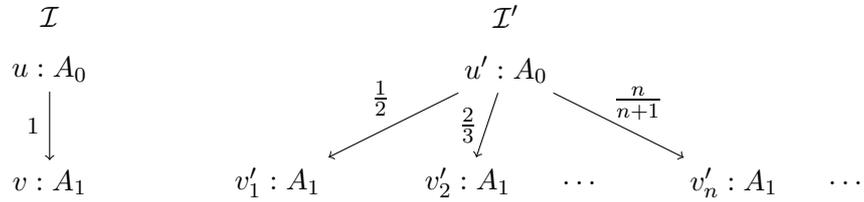
	\begin{center}
		\begin{tikzpicture}
		\node (x0) {};
		\node (x) [node distance=3.0cm, right of=x0] {};
		\node (I) [node distance=0.0cm, below of=x] {$\mI$};
		\node (I1) [node distance=6.0cm, right of=I] {$\mI'$};
		\node (u) [node distance=0.7cm, below of=I] {$u:A_0$};
		\node (v) [node distance=1.5cm, below of=u] {$v:A_1$};
		\node (up) [node distance=0.7cm, below of=I1] {$u':A_0$};
		\node (ubp) [node distance=1.5cm, below of=up] {};
		\node (v1) [node distance=3.0cm, left of=ubp] {$v'_1:A_1$};
		\node (v2) [node distance=0.5cm, left of=ubp] {$v'_2:A_1$};
		\node (v3) [node distance=1.0cm, right of=ubp] {\ldots};
		\node (vn) [node distance=3.0cm, right of=ubp] {$v'_n:A_1$};
		\node (vnp) [node distance=4.5cm, right of=ubp] {\ldots};
		\draw[->] (u) to node [left]{\footnotesize{1}} (v);
		\draw[->] (up) to node [above=2pt, pos=0.60]{$\frac{1}{2}$} (v1);
		\draw[->] (up) to node [left]{$\frac{2}{3}$} (v2);
		\draw[->] (up) to node [above=2pt, pos=0.65]{$\frac{n}{n+1}$} (vn);
		\end{tikzpicture}
	\end{center}
	The fuzzy interpretation $\mIp$ is similar to a fuzzy transition system given in~\cite[Fig.~2]{CaoCK11}. It is not witnessed w.r.t.\ \mLP. 
	We have that $v \simP v'_i$ for every $i$, but $u \not\simP u'$. Hence, $\mI$ and $\mIp$ are not $\Phi$-bisimilar. Treating $\mI$ and $\mIp$ as Kripke models in the fuzzy monomodal logic $K$, it can be checked that $Z = \{\tuple{u,u'}, \tuple{v,v'_i} \mid i \in \NN\setminus \{0\}\}$ is the greatest fuzzy bisimulation between $\mI$ and $\mIp$ according to Fan's definition of fuzzy bisimulation~\cite{Fan15}. This shows that our notion of fuzzy bisimulation is different in nature from the ones in \cite{CiricIDB12,Fan15} for non-witnessed structures. 
	\myend
\end{example}


\subsection{Invariance Results}
\label{section: invariance}

In this subsection, we prove an important property of fuzzy bisimulations under the G{\"odel} semantics. It states that, if fuzzy interpretations $\mI$ and $\mI'$ are witnessed w.r.t.\ \mLP and $Z$ is a fuzzy $\Phi$-bisimulation between them, then \mbox{$Z(x,x') \leq (C^\mI(x) \fequiv C^\mIp(x'))$} for every $x \in \Delta^\mI$, $x' \in \Delta^\mIp$ and every concept $C$ of \mLP. This is a part of Lemma~\ref{lemma: GDHAW}. We also present results on invariance of concepts (Theorem~\ref{theorem: UFNSJ}) as well as conditional invariance of fuzzy TBoxes/ABoxes (Theorems \ref{theorem: UDKMS} and \ref{theorem: IFDMS}) under bisimilarity in fuzzy DLs that use the G\"odel semantics. 

\newcommand{\TextLemmaGDHAW}{
	Let $\mI$ and $\mI'$ be fuzzy interpretations that are witnessed w.r.t.\ \mLP and $Z$ a fuzzy $\Phi$-bisimulation between $\mI$ and $ \mI'$. Then, the following properties hold for every concept $C$ of \mLP, every role $R$ of \mLP, every $x \in \Delta^\mI$ and every $x' \in \Delta^\mIp$:
	\begin{eqnarray}
	&& Z(x,x') \leq (C^\mI(x) \fequiv C^\mIp(x')) \label{eq: GDHAW 1} \\[0.5ex]
	&& \V y \in \Delta^\mI\ \E y' \in \Delta^\mIp\ Z(x,x') \fand R^\mI(x,y) \leq Z(y,y') \fand R^\mIp(x',y') \label{eq: GDHAW 2} \\[0.5ex]
	&& \V y' \in \Delta^\mIp\ \E y \in \Delta^\mI\ Z(x,x') \fand R^\mIp(x',y') \leq Z(y,y') \fand R^\mI(x,y). \label{eq: GDHAW 3} 
	\end{eqnarray}
}

\newcommand{\TextLemmaGDHAWp}{
	Let $\mI$ and $\mI'$ be fuzzy interpretations that are witnessed w.r.t.\ \mLP and $Z$ a fuzzy $\Phi$-bisimulation between $\mI$ and $ \mI'$. Then, the following properties hold for every concept $C$ of \mLP, every role $R$ of \mLP, every $x \in \Delta^\mI$ and every $x' \in \Delta^\mIp$:
	\begin{eqnarray*}
	\eqref{eq: GDHAW 1} && Z(x,x') \leq (C^\mI(x) \fequiv C^\mIp(x')) \\[0.5ex]
	\eqref{eq: GDHAW 2} && \V y \in \Delta^\mI\ \E y' \in \Delta^\mIp\ Z(x,x') \fand R^\mI(x,y) \leq Z(y,y') \fand R^\mIp(x',y') \\[0.5ex]
	\eqref{eq: GDHAW 3} && \V y' \in \Delta^\mIp\ \E y \in \Delta^\mI\ Z(x,x') \fand R^\mIp(x',y') \leq Z(y,y') \fand R^\mI(x,y). 
	\end{eqnarray*}
}

\begin{lemma} \label{lemma: GDHAW}
\TextLemmaGDHAW
\end{lemma}

\modifiedA{See the Appendix for the proof of this lemma.}

The following lemma differs from Lemma~\ref{lemma: GDHAW} in that \mLPp is used instead of \mLP. 
Its proof is a shortened version of the one of Lemma~\ref{lemma: GDHAW}, as~\eqref{eq: GDHAW 2} (resp.\ \eqref{eq: GDHAW 3}) is the same as~\eqref{eq: FB 3} and~\eqref{eq: FB 8} (resp.\ \eqref{eq: FB 4} and \eqref{eq: FB 9}) when $R$ is a role of \mLPp, and we can ignore the cases when $C$ is $\V R.D$, \mbox{$<\!n\,R.D$} or \mbox{$<\!n\,R$}.

\begin{lemma} \label{lemma: GDHAW2}
	Let $\mI$ and $\mI'$ be fuzzy interpretations that are witnessed w.r.t.\ \mLPp and $Z$ a fuzzy $\Phi$-bisimulation between $\mI$ and $ \mI'$. Then, for every concept $C$ of \mLPp, every $x \in \Delta^\mI$ and every $x' \in \Delta^\mIp$, 
	\[ Z(x,x') \leq (C^\mI(x) \fequiv C^\mIp(x')). \]
\end{lemma}

A concept $C$ of \mLP is said to be {\em invariant under $\Phi$-bisimilarity} (between witnessed interpretations) if, for any fuzzy interpretations $\mI$ and $\mI'$ that are witnessed  w.r.t.\ \mLP and any $x \in \Delta^\mI$ and $x' \in \Delta^\mIp$, if $x \simP x'$, then $C^\mI(x) = C^\mIp(x')$.
The following theorem immediately follows from the assertion~\eqref{eq: GDHAW 1} of Lemma~\ref{lemma: GDHAW}. 

\begin{theorem}\label{theorem: UFNSJ}
All concepts of \mLP are invariant under $\Phi$-bisimilarity. 
\end{theorem}

A fuzzy TBox $\mT$ is said to be {\em invariant under $\Phi$-bisimilarity} (between witnessed interpretations) if, for every fuzzy interpretations $\mI$ and $\mI'$ that are witnessed w.r.t.\ \mLP and $\Phi$-bisimilar to each other, $\mI \models \mT$ iff $\mI' \models \mT$. The notion of invariance of fuzzy ABoxes under $\Phi$-bisimilarity (between witnessed interpretations) is defined analogously. 

\begin{theorem}\label{theorem: UDKMS}
	If $U \in \Phi$, then all fuzzy TBoxes in \mLP are invariant under $\Phi$-bisimilarity. 
\end{theorem}

\begin{proof}
	Suppose $U \in \Phi$ and let $\mI$ and $\mI'$ be fuzzy interpretations that are witnessed w.r.t.\ \mLP and $\Phi$-bisimilar to each other. Thus, it is implicitly assumed that $\IN \neq \emptyset$. Let $\mT$ be a fuzzy TBox in \mLP and suppose $\mI \models \mT$. Let $(C \sqsubseteq D) \rhd p$ be a fuzzy GCI from $\mT$. We need to show that $\mI' \models (C \sqsubseteq D) \rhd p$. Let $y' \in \Delta^\mIp$. We show that $(C \to D)^\mIp(y') \rhd p$. Let $Z$ be a fuzzy $\Phi$-bisimulation between $\mI$ and $\mIp$ such that $Z(a^\mI,a^\mIp) = 1$ for all $a \in \IN$. Choose any $a$ from $\IN$. Since $Z(a^\mI,a^\mIp) = 1$, by~\eqref{eq: FB 9}, there exists $y \in \Delta^\mI$ such that $Z(a^\mI,a^\mIp) \leq Z(y,y')$. Thus, $Z(y,y') = 1$. Since $\mI \models ((C \sqsubseteq D) \rhd p)$, we have $(C \to D)^\mI(y) \rhd p$. Since $Z(y,y') = 1$, by Theorem~\ref{theorem: UFNSJ}, it follows that \mbox{$(C \to D)^\mIp(y') \rhd p$}.
	\myend
\end{proof}

\begin{theorem}\label{theorem: IFDMS}
	Let $\mA$ be a fuzzy ABox in \mLP. If $O \in \Phi$ or $\mA$ consists of only fuzzy assertions of the form $C(a) \bowtie p$, then $\mA$ is invariant under $\Phi$-bisimilarity. 
\end{theorem}

\begin{proof}
	Suppose that $O \in \Phi$ or $\mA$ consists of only fuzzy assertions of the form $C(a) \bowtie p$. 
	Let $\mI$ and $\mI'$ be fuzzy interpretations that are witnessed w.r.t.\ \mLP. Suppose that $\mI \simP \mI'$ and $\mI \models \mA$. Let $\varphi \in \mA$. It is sufficient to show that $\mI' \models \varphi$. Let $Z$ be a fuzzy $\Phi$-bisimulation between $\mI$ and $\mIp$ such that $Z(a^\mI,a^\mIp) = 1$ for all $a \in \IN$. 
	\begin{itemize}
		\item Case $\varphi = (a \doteq b)\,$: Since $\mI \models \mA$, $a^\mI = b^\mI$. Since $\mI \simP \mI'$, $a^\mI \simP a^{\mI'}$. Since $a^\mI = b^\mI$, by Condition~\eqref{eq: FB 5}, it follows that $a^{\mI'} = b^{\mI'}$. Therefore, $\mI' \models \varphi$. 
		
		\item Case $\varphi = (a \not\doteq b)$ is similar to the above one.
		
		\item Case $\varphi = (C(a) \bowtie  p)\,$: Since $\mI \models \mA$, $C^\mI(a^\mI) \bowtie p$. Since $\mI \simP \mI'$, $a^\mI \simP a^{\mI'}$. By Theorem~\ref{theorem: UFNSJ}, it follows that $C^{\mI'}(a^{\mI'}) = C^\mI(a^\mI) \bowtie p$. Hence, $\mI' \models \varphi$. 
		
		\item Case $\varphi = (R(a,b) \rhd p)$, with $\rhd \in \{\geq,>\}$: Since $\mI \models \mA$, $R^\mI(a^\mI,b^\mI) \rhd p$. By~\eqref{eq: GDHAW 2}, there exists $y' \in \Delta^\mIp$ such that $Z(b^\mI,y') \rhd p$ and $R^\mIp(a^\mIp,y') \rhd p$. Consider $C = \{b\}$. Since $Z(b^\mI,y') \rhd p$ and $C^\mI(b^\mI) = 1 \rhd p$, by Lemma~\ref{lemma: GDHAW}, $C^{\mI'}(y') \rhd p$. If $\rhd$ is $>$ or $p > 0$, then we must have $y' = b^\mIp$. Hence, $\mI' \models \varphi$. 
		
		\item Case $\varphi = (R(a,b) \lhd p)$, with $\lhd \in \{\leq,<\}$: For a contradiction, suppose $\mI' \not\models \varphi$. Thus, $R^\mIp(a^\mIp,b^\mIp) \rhd p$, where $\rhd$ is \mbox{$\not\!\!\lhd$}. Similarly to the above case, we can derive that $\mI \models (R(a,b) \rhd p)$, which contradicts $\mI \models \varphi$.
		\myend
	\end{itemize}
\end{proof}


\subsection{The Hennessy-Milner Property}
\label{sec: HM-prop}

In this subsection, we present and prove the Hennessy-Milner property of fuzzy bisimulations for fuzzy DLs under the G{\"o}del semantics (Theorem~\ref{theorem: fG H-M}). It uses the notion of modal saturatedness defined below, which is a technical notion related to compactness that can replace image-finiteness to make the Hennessy-Milner property stronger (see~\cite[Section~2.5]{BRV2001} for a further discussion).

\begin{definition}\label{def: modally saturated}
	A fuzzy interpretation $\mI$ is said to be {\em modally saturated} w.r.t.~\mLPp (and the G\"odel semantics) if the following conditions hold:
	\begin{itemize}
		\item for every $p \in (0,1]$, every $x \in \Delta^\mI$, every basic role $R$ w.r.t.~$\Phi$ and every infinite set $\Gamma$ of concepts in~\mLPp, if for every finite subset $\Lambda$ of $\Gamma$ there exists $y \in \Delta^\mI$ such that $R^\mI(x,y) \fand C^\mI(y) \geq p$ for all $C \in \Lambda$, then there exists $y \in \Delta^\mI$ such that $R^\mI(x,y) \fand C^\mI(y) \geq p$ for all $C \in \Gamma$; 
		
		\item if $Q_n \in \Phi$, then for every $p \in (0,1]$, every $x \in \Delta^\mI$, every basic role $R$ w.r.t.~$\Phi$ and every infinite set $\Gamma$ of concepts in~\mLPp, if for every finite subset $\Lambda$ of $\Gamma$ there exist $n$ pairwise distinct $y_1,\ldots,y_n \in \Delta^\mI$ such that $R^\mI(x,y_i) \fand C^\mI(y_i) \geq p$ for all $1 \leq i \leq n$ and $C \in \Lambda$, then there exist $n$ pairwise distinct $y_1,\ldots,y_n \in \Delta^\mI$ such that $R^\mI(x,y_i) \fand C^\mI(y_i) \geq p$ for all $1 \leq i \leq n$ and $C \in \Gamma$; 
		
		\item if $U \in \Phi$, then for every $p \in (0,1]$ and every infinite set $\Gamma$ of concepts in~\mLPp, if for every finite subset $\Lambda$ of $\Gamma$ there exists $y \in \Delta^\mI$ such that $C^\mI(y) \geq p$ for all $C \in \Lambda$, then there exists $y \in \Delta^\mI$ such that $C^\mI(y) \geq p$ for all $C \in \Gamma$.
		\myend
	\end{itemize}
\end{definition}

Our notion of modal saturatedness is an adaptation of the ones in~\cite{Fine75,BRV2001,BSDL-P-LOGCOM,Nguyen-TFS2019}.
Observe that every finite fuzzy interpretation is modally saturated w.r.t.\ \mLPp for any $\Phi$. 
%
\modifiedA{
	On the relationship between image-finiteness, modal saturatedness and being witnessed, note that:
	\begin{itemize}
		\item if $U \notin \Phi$, then every image-finite fuzzy interpretation w.r.t.\ $\Phi$ is witnessed and modally saturated w.r.t.~\mLPp, 
		\item there exist interpretations that are witnessed and modally saturated w.r.t.~\mLPp but not image-finite w.r.t.\ $\Phi$ (one can take a finite interpretation with an individual $x$ and its successor $y$ and clone $y$ together with its relationship to $x$ infinitely many times), 
		\item there exist interpretations that are witnessed but not modally saturated w.r.t.~\mLPp (see Examples~\ref{example: KLLWA} and~\ref{example: KLLWA 2} \modifiedB{given below}), 
		\item there exist interpretations that are modally saturated but not witnessed w.r.t.~\mLPp (the fuzzy interpretation $\mI'$ given in Example~\ref{example: HFKSB} is modally saturated but not witnessed w.r.t.~\mLPp for the case $\Phi = \emptyset$).
	\end{itemize} 
} 

\begin{example}\label{example: KLLWA}
\markModificationA
We give here an interpretation that is crisp (and hence also witnessed w.r.t.~\mLP) but not modally saturated w.r.t.~\mLPp. Let $\CN = \{A_i \mid i \in \NN\}$, $\RN = \{r\}$ and $\IN = \emptyset$. Let $\mI$ be the interpretation specified as follows:
\begin{itemize}
\item $\Delta^\mI = \{u, v_i \mid i \in \NN\}$, 
\item $r^\mI(u,v_i) = 1$ for $i \in \NN$, and $r^\mI(x,y) = 0$ for the other pairs $\tuple{x,y}$, 
\item $A_i^\mI(v_i) = 0$, and $A_i^\mI(x) = 1$ if $x \neq v_i$, for $i \in \NN$.
\end{itemize}
The first condition of Definition~\ref{def: modally saturated} does not hold for $p = 1$, $x = u$, $R = r$ and $\Gamma = \{A_i \mid i \in \NN\}$. Therefore, $\mI$ is not modally saturated w.r.t.~\mLPp. 
\myend
\end{example}

\begin{example}\label{example: KLLWA 2}
\markModificationA
We give here another interpretation that is not modally saturated w.r.t.~\mLPp, using $\CN = \emptyset$, $\RN = \{r\}$, $\IN = \emptyset$ and $\Phi \supseteq \{N_n \mid n \in \NN \setminus \{0\}\}$. Let $\mI$ be the interpretation specified as follows:
	\begin{itemize}
		\item $\Delta^\mI = \{u, v_i, w_{i,j} \mid i,j \in \NN, j < i\}$, 
		\item $r^\mI(x,y) = 1$ for $\tuple{x,y} \in \{\tuple{u,v_i}, \tuple{v_i,w_{i,j}} \mid i,j \in \NN, j < i\}$, and $r^\mI(x,y) = 0$ for the other pairs $\tuple{x,y}$.  
	\end{itemize}
	The first condition of Definition~\ref{def: modally saturated} does not hold for $p = 1$, $x = u$, $R = r$ and $\Gamma = \{\geq\!n\,r \mid$ $n \in \NN \setminus \{0\}\}$. Therefore, $\mI$ is not modally saturated w.r.t.~\mLPp. 
	\myend
\end{example}

\newcommand{\TextTheoremfGHM}{
	Let $\mI$ and $\mI'$ be fuzzy interpretations that are witnessed and modally saturated w.r.t.~\mLPp. 
	Let $Z : \Delta^\mI \times \Delta^\mIp \to [0,1]$ be specified by 
	\[ Z(x,x') = \inf\{C^\mI(x) \fequiv C^\mIp(x') \mid \textrm{$C$ is a concept of \mLPp}\}. \] 
	Then, $Z$ is the greatest fuzzy $\Phi$-bisimulation between $\mI$ and~$\mI'$.	
}

\begin{theorem} \label{theorem: fG H-M}
\TextTheoremfGHM
\end{theorem}

\modifiedA{See the Appendix for the proof of this theorem.}

Given fuzzy interpretations $\mI$, $\mIp$ and $x \in \Delta^\mI$, $x' \in \Delta^\mIp$, we write $x \equivP x'$ to denote that $C^\mI(x) = C^\mIp(x')$ for every concept $C$ of~\mLP. Similarly, we write $x \equivPp x'$ to denote that $C^\mI(x) = C^\mIp(x')$ for every concept $C$ of~\mLPp. 

\begin{corollary} \label{cor: fG H-M 1}
	Let $\mI$, $\mI'$ be fuzzy interpretations and let $x \in \Delta^\mI$, $x' \in \Delta^\mIp$.
	\begin{enumerate}
		\item\label{cor: fG H-M 1-1} If $\mI$ and $\mI'$ are witnessed and modally saturated w.r.t.~\mLPp, then
		\[ x \simP x'\ \ \textrm{iff}\ \ x \equivPp x'. \]
		\item\label{cor: fG H-M 1-2} If $\mI$ and $\mI'$ are image-finite w.r.t.~$\Phi$ and $U \notin \Phi$, then
		\[ x \simP x'\ \ \textrm{iff}\ \ x \equivPp x'. \] 
		\item\label{cor: fG H-M 1-3} If $\mI$ and $\mI'$ are witnessed w.r.t.~\mLP and modally saturated w.r.t.~\mLPp, then
		\[ x \equivP x'\ \ \textrm{iff}\ \ x \simP x'\ \ \textrm{iff}\ \ x \equivPp x'. \]
	\end{enumerate}	
\end{corollary}

The assertion~\ref{cor: fG H-M 1-1} (resp.~\ref{cor: fG H-M 1-3}) directly follows from Theorem~\ref{theorem: fG H-M} and Lemma~\ref{lemma: GDHAW2} (resp.~\ref{lemma: GDHAW}). 
The assertion~\ref{cor: fG H-M 1-2} directly follows from the assertion~\ref{cor: fG H-M 1-1}. 
The following corollary directly follows from Theorem~\ref{theorem: fG H-M} and Lemma~\ref{lemma: GDHAW}. 

\begin{corollary} \label{cor: fG H-M 4}
Suppose $\IN \neq \emptyset$ and let $\mI$ and $\mI'$ be fuzzy interpretations that are witnessed w.r.t.~\mLP and modally saturated w.r.t.~\mLPp. Then, $\mI$ and $\mIp$ are $\Phi$-bisimilar iff $a^\mI \equivPp a^\mIp$ for all $a \in \IN$.
\end{corollary}


\section{Crisp Bisimulations for Fuzzy DLs with Involutive Negation under the G\"odel Semantics}
\label{section: crisp bis}
	
In this section, we consider fuzzy DLs extended with involutive negation or the Baaz projection operator under the G\"odel semantics, which are specified below. We provide notions and results on crisp bisimulation and strong bisimilarity for such logics.

We denote involutive negation by $\ineg$ (as $\lnot$ and $\sim$ are used to denote the G\"odel negation and bisimilarity, respectively). 
Let \mLPn be the fuzzy DL that extends \mLP with involutive negation. That is, in the inductive definition, if $C$ is a concept of \mLPn, then \mbox{$\ineg\!C$} is also a concept of \mLPn. The meaning of \mbox{$\ineg\!C$} in a fuzzy interpretation $\mI$ is specified as follows: 
\[ (\ineg\!C)^\mI(x) = 1 - C^\mI(x) \textrm{ for $x \in \Delta^\mI$. } \]

We will use the projection operator $\triangle: [0,1] \to \{0,1\}$ defined by Baaz~\cite{Baaz96} as follows:
\[ \triangle x = \textrm{(if $x = 1$ then 1 else 0)} \]
It is called the Baaz Delta in~\cite{Fan15}. 
We treat $\triangle$ also as a concept constructor with the meaning specified by $(\triangle C)^\mI(x) = \triangle(C^\mI(x))$. It is easily seen that \mbox{$(\triangle C)^\mI = (\lnot\!\ineg\!C)^\mI$}. Thus, we will treat $\triangle$ as an abbreviation for \mbox{$\lnot\!\ineg$}. 
By \DLP we denote the largest sublanguage of \mLPn that uses $\ineg$ only in the form \mbox{$\lnot\!\ineg\!C$} (i.e., $\triangle C$).
By \DLPp we denote the largest sublanguage of \mLPn that:
\begin{itemize}
	\item disallows the role constructors $R \circ S$, $R \mor S$, $R^*$, $C?$ and the concept constructors $C \mor D$, $\V R.C$, $<\!n\,R.C$, $<\!n\,R$, 
	\item \modifiedA{uses $\to$ only in the form $C \to p$ or $p \to C$, where $p \in [0,1]$,} 
	\item uses $\lnot$ and $\ineg$ only in the form \mbox{$\lnot\!\ineg\!C$} (i.e., $\triangle C$).
\end{itemize} 

\subsection{Crisp Bisimulations and Invariance Results}
	
Let $\Phi \subseteq \{I$, $O$, $U$, $\Self$, $Q_n$, $N_n \mid$ $n \in \NN \setminus \{0\}\}$ be a set of features and $\mI$, $\mI'$ fuzzy interpretations. A function \mbox{$Z : \Delta^\mI \times \Delta^\mIp \to \{0,1\}$} is called a {\em crisp $\Phi$-bisimulation} (under the G\"odel semantics) between $\mI$ and $\mI'$ if it satisfies the conditions of being a fuzzy $\Phi$-bisimulation (as in Definition~\ref{def: DHGAK}). Notice that such a function can be treated as the (crisp) relation $\{\tuple{x,x'} \in \Delta^\mI \times \Delta^{\mIp} \mid Z(x,x') = 1\}$. 
	
Using the fact that $Z$ is crisp, Conditions~\eqref{eq: FB 2}--\eqFBlast can be rewritten appropriately. For example, \eqref{eq: FB 2} is equivalent to:
\[ \textrm{if $Z(x,x') = 1$, then $A^\mI(x) = A^\mIp(x')$.} \]
	
Note that a crisp $\Phi$-bisimulation is a special fuzzy $\Phi$-bisimulation. 

\begin{example}\label{example: HDKSL2-2}
	Let $\RN = \{r\}$, $\CN = \{A\}$ and $\IN = \{a\}$. Consider the fuzzy interpretations $\mI$ and $\mIp$ illustrated below and specified similarly as in Example~\ref{example: HDKSL}, with $a^\mI = u$ and $a^\mIp = u'$:
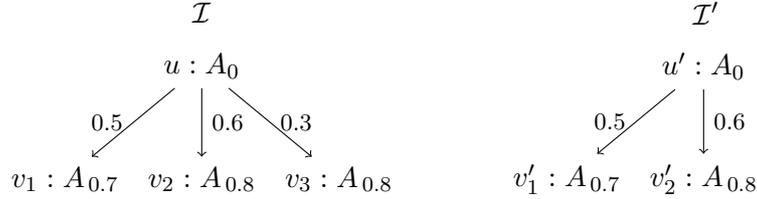
	\begin{center}		
		\begin{tikzpicture}
		\node (I) {$\mI$};
		\node (u) [node distance=0.7cm, below of=I] {$u:A_0$};
		\node (ub) [node distance=1.5cm, below of=u] {$v_2:A_{\,0.8}$};
		\node (v) [node distance=1.8cm, left of=ub] {$v_1:A_{\,0.7}$};
		\node (w) [node distance=1.8cm, right of=ub] {$v_3:A_{\,0.8}$};
		\draw[->] (u) to node [left]{\footnotesize{0.5}} (v);	
		\draw[->] (u) to node [right]{\footnotesize{0.6}} (ub);
		\draw[->] (u) to node [right]{\footnotesize{0.3}} (w);
		\node (vp) [node distance=3cm, right of=w] {$v'_1:A_{\,0.7}$};
		\node (upb) [node distance=1.8cm, right of=vp] {$v'_2:A_{\,0.8}$};
		\node (up) [node distance=1.5cm, above of=upb] {$u':A_0$};
		\node (Ip) [node distance=0.7cm, above of=up] {$\mIp$};
		\draw[->] (up) to node [left]{\footnotesize{0.5}} (vp);
		\draw[->] (up) to node [right]{\footnotesize{0.6}} (upb);
		\end{tikzpicture}
	\end{center}
	
	Let $Z : \Delta^\mI \times \Delta^\mIp \to \{0,1\}$ be the function specified by: $Z(x,x') = 1$ iff  $\tuple{x,x'} \in \{\tuple{u,u'}, \tuple{v_1,v'_1}, \tuple{v_2,v'_2}, \tuple{v_3,v'_2} \}$. 
	It can be checked that, for any $\Phi \subseteq \{O,U,\Self,Q_n,N_n \mid n \in \NN \setminus \{0,3\}\}$, $Z$ is a crisp $\Phi$-bisimulation between $\mI$ and~$\mIp$, and these fuzzy interpretations are strongly $\Phi$-bisimilar. If $I \in \Phi$, then $v_3 \not\simPn v'_2$, and hence $u \not\simPn u'$. If $\{Q_3, N_3\} \cap \Phi \neq \emptyset$, then clearly $u \not\simPn u'$. Therefore, if $\{I, Q_3, N_3\} \cap \Phi \neq \emptyset$, then $\mI$ and $\mIp$ are not strongly $\Phi$-bisimilar. 
	\myend
\end{example}

The following proposition is a counterpart of Proposition~\ref{prop: HFHSJ}. Its proof is also straightforward. 

\begin{proposition}\label{prop: HFHSJ 2}
	Let $\mI$, $\mIp$ and $\mI''$ be fuzzy interpretations.
	\begin{enumerate}
		\item\label{item: HFHSJ2 1} The function $Z : \Delta^\mI \times \Delta^\mI \to \{0,1\}$ specified by \[ Z(x,x') = (\textrm{if $x = x'$ then 1 else 0}) \] is a crisp $\Phi$-bisimulation between $\mI$ and itself.
		\item\label{item: HFHSJ2 2}  If $Z$ is a crisp $\Phi$-bisimulation between $\mI$ and $\mIp$, then $Z^-$ is a crisp $\Phi$-bisimulation between $\mIp$ and~$\mI$.
		\item\label{item: HFHSJ2 3}  If $Z_1$ is a crisp $\Phi$-bisimulation between $\mI$ and $\mIp$, and $Z_2$ is a crisp $\Phi$-bisimulation between $\mIp$ and $\mI''$, then $Z_1 \circ Z_2$ is a crisp $\Phi$-bisimulation between $\mI$ and $\mI''$.
		\item\label{item: HFHSJ2 4}  If $\mZ$ is a (finite or infinite) set of crisp $\Phi$-bisimulations between $\mI$ and $\mIp$, then $\sup\mZ$ is also a crisp $\Phi$-bisimulation between $\mI$ and $\mIp$.
	\end{enumerate}   
\end{proposition}

The notion of being {\em witnessed w.r.t.\ \mLPn} (resp.\ \DLPp) is defined in the usual way. 
The following lemma differs from Lemma~\ref{lemma: GDHAW} in that it is formulated for crisp $\Phi$-bisimulations and concepts of \mLPn.
	
\begin{lemma} \label{lemma: cGDHAW}
Let $\mI$ and $\mI'$ be fuzzy interpretations that are witnessed w.r.t.\ \mLPn and $Z$ a crisp $\Phi$-bisimulation between $\mI$ and $ \mI'$. Then, the following properties hold for every concept $C$ of \mLPn, every role $R$ of \mLPn, every $x \in \Delta^\mI$ and every $x' \in \Delta^\mIp$:
\begin{eqnarray}
&& Z(x,x') \leq (C^\mI(x) \fequiv C^\mIp(x')) \label{eq: cGDHAW 1} \\[0.5ex]
&& \V y \in \Delta^\mI\ \E y' \in \Delta^\mIp\; Z(x,x') \fand R^\mI(x,y) \leq Z(y,y') \fand R^\mIp(x',y') \label{eq: cGDHAW 2} \\[0.5ex]
&& \V y' \in \Delta^\mIp\ \E y \in \Delta^\mI\; Z(x,x') \fand R^\mIp(x',y') \leq Z(y,y') \fand R^\mI(x,y). \label{eq: cGDHAW 3}
\end{eqnarray}
\end{lemma}
	
\begin{proof}
We apply the proof of Lemma~\ref{lemma: GDHAW} with slight modifications. Apart from that the occurrences of \mLP are replaced by \mLPn, the only other change is to consider also the case \mbox{$C =\ \ineg\!D$} when proving~\eqref{eq: cGDHAW 1}. Consider this case. If $Z(x,x') = 0$, then \eqref{eq: cGDHAW 1} clearly holds. So, suppose that $Z(x,x') \neq 0$, i.e., $Z(x,x') = 1$. By the inductive assumption of~\eqref{eq: cGDHAW 1}, \mbox{$Z(x,x') \leq (D^\mI(x) \fequiv D^\mIp(x'))$}. Thus, $D^\mI(x) = D^\mIp(x')$. Consequently, $C^\mI(x) = 1 - D^\mI(x) = 1 - D^\mIp(x') = C^\mIp(x')$, and hence \mbox{$Z(x,x') \leq (C^\mI(x) \fequiv C^\mIp(x'))$}. 
\myend
\end{proof}
	
The following lemma is a counterpart of Lemma~\ref{lemma: GDHAW2}. Its proof can be obtained from the proof of Lemma~\ref{lemma: cGDHAW} (and Lemma~\ref{lemma: GDHAW}) by simplification.
	
\begin{lemma} \label{lemma: cGDHAW2}
Let $\mI$ and $\mI'$ be fuzzy interpretations that are witnessed w.r.t.\ \DLPp and $Z$ a crisp $\Phi$-bisimulation between $\mI$ and $ \mI'$. Then, for every concept $C$ of \DLPp, every $x \in \Delta^\mI$ and every $x' \in \Delta^\mIp$, 
\[ Z(x,x') \leq (C^\mI(x) \fequiv C^\mIp(x')). \]
\end{lemma}
	
Let $\mI$ and $\mIp$ be fuzzy interpretations. For $x \in \Delta^\mI$ and $x' \in \Delta^\mIp$, we write $x \simPn x'$ to denote that there exists a crisp $\Phi$-bisimulation $Z$ between $\mI$ and $\mI'$ such that \mbox{$Z(x,x') = 1$}. If $x \simPn x'$, then we say that $x$ and $x'$ are {\em strongly $\Phi$-bisimilar}. 
If $\IN \neq \emptyset$ and there exists a crisp $\Phi$-bisimulation $Z$ between $\mI$ and $\mI'$ such that $Z(a^\mI,a^\mIp)=1$ for all $a \in \IN$, then we say that $\mI$ and $\mI'$ are {\em strongly $\Phi$-bisimilar} and write $\mI \simPn \mIp$. 
Note that both $\Phi$-bisimilarity and strong $\Phi$-bisimilarity are crisp relations. The difference is that the former is defined by using fuzzy $\Phi$-bisimulation, while the latter is defined by using crisp $\Phi$-bisimulation.  
	
A concept $C$ of \mLPn is said to be {\em invariant under strong $\Phi$-bisimilarity} (between witnessed interpretations) if, for any fuzzy interpretations $\mI$ and $\mI'$ that are witnessed w.r.t.\ \mLPn and any $x \in \Delta^\mI$ and $x' \in \Delta^\mIp$, if \mbox{$x \simPn x'$}, then $C^\mI(x) = C^\mIp(x')$. 
The following theorem is a counterpart of Theorem~\ref{theorem: UFNSJ}. It follows from the assertion~\eqref{eq: cGDHAW 1} of Lemma~\ref{lemma: cGDHAW}. 
	
\begin{theorem}\label{theorem: UFNSJ2}
All concepts of \mLPn are invariant under strong $\Phi$-bisimilarity. 
\end{theorem}
	
A fuzzy TBox $\mT$ is said to be {\em invariant under strong $\Phi$-bisimilarity} (between witnessed interpretations) if, for every fuzzy interpretations $\mI$ and $\mI'$ that are witnessed w.r.t.\ \mLPn and strongly $\Phi$-bisimilar to each other, $\mI \models \mT$ iff $\mI' \models \mT$. The notion of invariance of fuzzy ABoxes under strong $\Phi$-bisimilarity (between witnessed interpretations) is defined analogously. 
The following theorem is a counterpart of Theorem~\ref{theorem: UDKMS} \modifiedA{and can be proved similarly.} 

\begin{theorem}\label{theorem: UDKMS2}
If $U \in \Phi$, then all fuzzy TBoxes in \mLPn are invariant under strong $\Phi$-bisimilarity. 
\end{theorem}

\comment{
The proof of this theorem is obtained from the proof of Theorem~\ref{theorem: UDKMS} by replacing:
\begin{itemize}
\item ``witnessed w.r.t.\ \mLP'' with ``witnessed w.r.t.\ \mLPn'', 
\item ``fuzzy $\Phi$-bisimulation'' with ``crisp $\Phi$-bisimulation'', 
\item ``Theorem~\ref{theorem: UFNSJ}'' with ``Theorem~\ref{theorem: UFNSJ2}''.
\end{itemize} 
}
	
\begin{definition}
A fuzzy interpretation $\mI$ is {\em connected} w.r.t.\ $\Phi$ if, for every $x \in \Delta^\mI$, there exist $a \in \IN$, $x_0,\ldots,x_n \in \Delta^\mI$ and basic roles $R_1,\ldots,R_n$ w.r.t.~$\Phi$ such that $x_0 = a^\mI$, $x_n = x$ and $R_i^\mI(x_{i-1},x_i) > 0$ for all $1 \leq i \leq n$. 
\myend
\end{definition}

Our notion of connectedness is an adaptation of the one in~\cite{BSDL-P-LOGCOM} and the notion of being unreachable-objects-free~\cite{BSDL-INS}. The following theorem concerns invariance of fuzzy TBoxes without requiring $U \in \Phi$.
	
\begin{theorem}\label{theorem: UFSSK}
Let $\mT$ be a fuzzy TBox in \mLPn and $\mI$, $\mIp$ fuzzy interpretations that are witnessed w.r.t.\ \mLPn and strongly $\Phi$-bisimilar to each other. If $\mI$ and $\mIp$ are connected w.r.t.\ $\Phi$, then $\mI \models \mT$ iff $\mI' \models \mT$.
\end{theorem}

\begin{proof}
Suppose $\mI$ and $\mIp$ are connected w.r.t.\ $\Phi$. We prove that, if $\mI \models \mT$, then $\mI' \models \mT$. The converse is similar and omitted. Suppose $\mI \models \mT$. Let $(C \sqsubseteq D) \rhd p$ be any fuzzy GCI of $\mT$ and let $x' \in \Delta^\mIp$. We need to show that $(C \to D)^\mIp(x') \rhd p$. 
Since $\mIp$ is connected w.r.t.\ $\Phi$, there exists $a \in \IN$, $x'_0,\ldots,x'_n \in \Delta^\mIp$ and basic roles $R_1,\ldots,R_n$ w.r.t.~$\Phi$ such that $x'_0 = a^\mIp$, $x'_n = x'$ and $R_i^\mI(x'_{i-1},x'_i) > 0$ for all $1 \leq i \leq n$. Let $x_0 = a^\mI$ and let $Z$ be a crisp $\Phi$-bisimulation between $\mI$ and $\mI'$ such that $Z(b^\mI,b^\mIp) = 1$ for all $b \in \IN$. 
For each $i$ from $1$ to $n$, since $Z(x_{i-1},x'_{i-1}) = 1$ and $R^\mIp(x'_{i-1},x'_i) > 0$, by Condition~\eqref{eq: FB 4}, there exists $x_i \in \Delta^\mI$ such that $Z(x_i,x'_i) = 1$. Let $x = x_n$. Thus, $Z(x,x') = 1$. Since $\mI \models \mT$, we have $(C \to D)^\mI(x) \rhd p$. Since $Z(x,x') = 1$, by Theorem~\ref{theorem: UFNSJ2}, $(C \to D)^\mIp(x') = (C \to D)^\mI(x) \rhd p$, which completes the proof.
\myend
\end{proof}

The following theorem is a counterpart of Theorem~\ref{theorem: IFDMS} \modifiedA{and can be proved similarly.}  

\begin{theorem}\label{theorem: IFDMS2}
Let $\mA$ be a fuzzy ABox in \mLPn. If $O \in \Phi$ or $\mA$ consists of only fuzzy assertions of the form $C(a) \bowtie p$, then $\mA$ is invariant under strong $\Phi$-bisimilarity. 
\end{theorem}

\comment{
The proof of this theorem is obtained from the proof of Theorem~\ref{theorem: IFDMS} by replacing:
\begin{itemize}
\item ``witnessed w.r.t.\ \mLP'' with ``witnessed w.r.t.\ \mLPn'', 
\item ``fuzzy $\Phi$-bisimulation'' with ``crisp $\Phi$-bisimulation'',
\item all occurrences of $\simP$ with $\simPn$,
\item ``Theorem~\ref{theorem: UFNSJ}'' with ``Theorem~\ref{theorem: UFNSJ2}'',
\item ``Lemma~\ref{lemma: GDHAW}'' with ``Lemma~\ref{lemma: cGDHAW}''.
\end{itemize} 
}

\subsection{The Hennessy-Milner Property}

The notion of being {\em modally saturated w.r.t.~\DLPp} defined below is less restrictive than the notion of being ``modally saturated w.r.t.~\mLPp'' specified in Definition~\ref{def: modally saturated}. 

\begin{definition}\label{def: modally saturated 2}
	A fuzzy interpretation $\mI$ is said to be {\em modally saturated} w.r.t.~\DLPp (and the G\"odel semantics) if the following conditions hold:
	\begin{itemize}
		\item for every $p \in (0,1]$, every $x \in \Delta^\mI$, every basic role $R$ w.r.t.~$\Phi$ and every infinite set $\Gamma$ of concepts in~\DLPp, if for every finite subset $\Lambda$ of $\Gamma$ there exists $y \in \Delta^\mI$ such that $R^\mI(x,y) \fand C^\mI(y) \geq p$ for all $C \in \Lambda$, then there exists $y \in \Delta^\mI$ such that $R^\mI(x,y) \geq p$ and $C^\mI(y) > 0$ for all $C \in \Gamma$; 
		
		\item if $Q_n \in \Phi$, then for every $p \in (0,1]$, every $x \in \Delta^\mI$, every basic role $R$ w.r.t.~$\Phi$ and every infinite set $\Gamma$ of concepts in~\DLPp, if for every finite subset $\Lambda$ of $\Gamma$ there exist $n$ pairwise distinct $y_1,\ldots,y_n \in \Delta^\mI$ such that $R^\mI(x,y_i) \fand C^\mI(y_i) \geq p$ for all $1 \leq i \leq n$ and $C \in \Lambda$, then there exist $n$ pairwise distinct $y_1,\ldots,y_n \in \Delta^\mI$ such that $R^\mI(x,y_i) \geq p$ and $C^\mI(y_i) > 0$ for all $1 \leq i \leq n$ and $C \in \Gamma$; 
		
		\item if $U \in \Phi$, then for every infinite set $\Gamma$ of concepts in~\DLPp, if for every finite subset $\Lambda$ of $\Gamma$ there exists $y \in \Delta^\mI$ such that $C^\mI(y) = 1$ for all $C \in \Lambda$, then there exists $y \in \Delta^\mI$ such that $C^\mI(y) > 0$ for all $C \in \Gamma$.
		\myend
	\end{itemize}
\end{definition}

Observe that a condition like ``$R^\mI(x,y) \geq p$ and $C^\mI(y) > 0$'' is weaker than $R^\mI(x,y) \fand C^\mI(y) \geq p$. Also notice that the condition $C^\mI(y) = 1$ for the case when $U \in \Phi$ is stronger than $C^\mI(y) > 0$. These loosenings make the class of modally saturated interpretations larger. 
Like the case of \mLPp, we also have the following claims, which can be easily proved:
\begin{itemize}
\item every finite fuzzy interpretation is modally saturated w.r.t.\ \DLPp for any $\Phi$,
\item if $U \notin \Phi$, then every image-finite fuzzy interpretation w.r.t.\ $\Phi$ is modally saturated w.r.t.~\DLPp. 
\end{itemize}

\newcommand{\TextTheoremfGHMt}{
	Let $\mI$ and $\mI'$ be fuzzy interpretations that are witnessed and modally saturated w.r.t.~\DLPp. Let $Z : \Delta^\mI \times \Delta^\mIp \to \{0,1\}$ be specified by: $Z(x,x') = 1$ if $C^\mI(x) = C^\mIp(x)$ for all concepts $C$ of \DLPp, and \mbox{$Z(x,x') = 0$} otherwise. 
	Then, $Z$ is the greatest crisp $\Phi$-bisimulation between $\mI$ and~$\mI'$.
}

\begin{theorem} \label{theorem: fG H-M 2}
\TextTheoremfGHMt
\end{theorem}

This theorem is a counterpart of Theorem~\ref{theorem: fG H-M}. It can be proved analogously. \modifiedA{To} make the text self-contained and ease the checking for the reader, \modifiedA{we present its proof in the Appendix.} 

Given fuzzy interpretations $\mI$, $\mIp$ and $x \in \Delta^\mI$, $x' \in \Delta^\mIp$, we write $x \equivPn x'$ to denote that $C^\mI(x) = C^\mIp(x')$ for every concept $C$ of~\mLPn. Similarly, we write $x \equivPdp x'$ to denote that $C^\mI(x) = C^\mIp(x')$ for every concept $C$ of~\DLPp. 

\begin{corollary} \label{cor: fG H-M 1 c}
	Let $\mI$, $\mI'$ be fuzzy interpretations and let $x \in \Delta^\mI$, $x' \in \Delta^\mIp$.
	\begin{enumerate}
		\item\label{cor: fG H-M 1-1 c} If $\mI$ and $\mI'$ are witnessed and modally saturated w.r.t.~\DLPp, then
		\[ x \simPn x'\ \ \textrm{iff}\ \ x \equivPdp x'. \]
		\item\label{cor: fG H-M 1-2 c} If $\mI$ and $\mI'$ are image-finite w.r.t.~$\Phi$ and $U \notin \Phi$, then
		\[ x \simPn x'\ \ \textrm{iff}\ \ x \equivPdp x'. \] 
		\item\label{cor: fG H-M 1-3 c} If $\mI$ and $\mI'$ are witnessed w.r.t.~\mLPn and modally saturated w.r.t.~\DLPp, then
		\[ x \equivPn x'\ \ \textrm{iff}\ \ x \simPn x'\ \ \textrm{iff}\ \ x \equivPdp x'. \]
	\end{enumerate}	
\end{corollary}

The assertion~\ref{cor: fG H-M 1-1 c} (resp.~\ref{cor: fG H-M 1-3 c}) directly follows from Theorem~\ref{theorem: fG H-M 2} and Lemma~\ref{lemma: cGDHAW2} (resp.~\ref{lemma: cGDHAW}). 
The assertion~\ref{cor: fG H-M 1-2 c} directly follows from the assertion~\ref{cor: fG H-M 1-1 c}. 
The following corollary directly follows from Theorem~\ref{theorem: fG H-M 2} and Lemma~\ref{lemma: cGDHAW}. 

\begin{corollary} \label{cor: fG H-M 4 c}
Suppose $\IN \neq \emptyset$ and let $\mI$ and $\mI'$ be fuzzy interpretations that are witnessed w.r.t.~\mLPn and modally saturated w.r.t.~\DLPp. Then, $\mI$ and $\mIp$ are strongly $\Phi$-bisimilar iff $a^\mI \equivPdp a^\mIp$ for all $a \in \IN$.
\end{corollary}


\section{Separating the Expressive Powers of Fuzzy DLs}
\label{section: exp-powers}

We say that a concept $C$ {\em cannot be expressed} in a fuzzy DL $L$ if there does not exist any concept $D$ of $L$ equivalent to $C$. A fuzzy TBox $\mT$ (resp.\ fuzzy ABox $\mA$) {\em cannot be expressed} in a fuzzy DL $L$ if there does not exist any fuzzy TBox $\mT'$ (resp.\ fuzzy ABox $\mA'$) in $L$ such that, for every fuzzy interpretation~$\mI$, $\mI \models \mT$ iff $\mI \models \mT'$ (resp.\ $\mI \models \mA$ iff $\mI \models \mA'$).
If $L$ is a sublogic of a fuzzy DL $L'$ and there is a concept (resp.\ a fuzzy TBox, a fuzzy ABox) in $L'$ that cannot be expressed in $L$, then we say that $L'$ is {\em strictly more expressive} than $L$ w.r.t.\ concepts (resp.\ TBoxes, ABoxes). 

We denote $\PhiF = \{I,O,U,\Self,Q_n,N_n \mid n \in \NN \setminus \{0\}\}$. 
In~\cite{BSDL-INS}, Divroodi and Nguyen used crisp bisimulations to separate the expressive powers of traditional DLs w.r.t.\ concepts, TBoxes and ABoxes. As traditional DLs are a special kind of fuzzy DLs and crisp bisimulations are a special kind of fuzzy bisimulations, the proofs of~\cite{BSDL-INS} can still be applied to separate the expressive powers of fuzzy DLs \mLP for $\Phi \subseteq \PhiF$. The work~\cite{BSDL-INS} considers the feature $Q$ (qualified number restrictions) that includes $Q_n$ for all $n \in \NN$, but it does neither consider features $Q_n$ separately, nor unqualified number restrictions $N_n$. To deal with the features $Q_n$ and $N_n$, one can use a similar technique, which relies on crisp bisimulations and crisp interpretations. The task is straightforward and omitted. 

In this section, as a special point that relies on fuzzy bisimulations, we prove that involutive negation and the Baaz projection operator cannot be expressed in fuzzy DLs by using the other constructors. 

\comment{
\modifiedA{As related work, it is observed in~\cite{DBLP:conf/kr/BorgwardtDP14} that {\em ``the residual negation is often used in fuzzy logics, but under G\"odel semantics it is much less expressive than the involutive negation''}. Both the works~\cite{DBLP:conf/kr/BorgwardtDP14,DBLP:journals/ai/BorgwardtDP15} concern decidability of fuzzy DLs, but they also contain certain assertions related to expressivity. For example, in~\cite{DBLP:journals/ai/BorgwardtDP15} Borgwardt et al.\ wrote: {\em ``Intuitively, a fuzzy DL is undecidable whenever it can express upper bounds for the membership degrees of concepts, e.g.\ through the involutive negation or the implication constructor. On the other hand, our decidability results exploit the fact that some fuzzy DLs cannot express such upper bounds except for 0''.}}
} 

\begin{proposition}\label{prop: JDAJA}
The concepts \mbox{$\ineg\!A$} and $\triangle A$ cannot be expressed in \mLP with $\Phi = \PhiF$.
\end{proposition}

\begin{proof}
	Let $\mI$ be the fuzzy interpretation such that $\Delta^\mI = \{v\}$, $A^\mI(v) = 0.5$, $B^\mI(v) = 0$ for all $B \in \CN\setminus \{A\}$, and $r^\mI(v,v) = 0$ for all $r \in \RN$. Let $\mIp$ be the fuzzy interpretation defined similarly, except that $A^\mIp(v) = 1$. 
	Let $Z : \Delta^\mI \times \Delta^\mIp \to [0,1]$ be the function specified by $Z(v,v) = 0.5$. 
	Clearly, $Z$ is the greatest fuzzy $\Phi$-bisimulation between $\mI$ and~$\mIp$. 
	
	We have $(\ineg\!A)^\mI(v) = 0.5$ and $(\ineg\!A)^\mIp(v) = 0$. 
	If $\ineg\!A$ can be expressed by a concept $C$ of \mLP, then $C^\mI(v) = 0.5$ and, by the assertion~\eqref{eq: GDHAW 1} of Lemma~\ref{lemma: GDHAW}, $C^\mIp(v) \geq 0.5$, which contradicts $(\ineg\!A)^\mIp(v) = 0$. 
	
	We have $(\triangle A)^\mI(v) = 0$ and $(\triangle A)^\mIp(v) = 1$. 
	If $\triangle A$ can be expressed by a concept $C$ of \mLP, then $C^\mI(v) = 0$ and, by the assertion~\eqref{eq: GDHAW 1} of Lemma~\ref{lemma: GDHAW}, $C^\mIp(v) = 0$, which contradicts $(\triangle A)^\mIp(v) = 1$.  
	\myend
\end{proof}

\begin{proposition}\label{prop: JDGAS 2}
The fuzzy ABoxes \mbox{$\{(\E r.\!\ineg\!A)(a) \geq 0.1\}$} and \mbox{$\{(\E r.\triangle A)(a) \geq 0.1\}$} cannot be expressed in \mLP when $\PhiF\setminus\Phi$ is $\{U\}$, $\{I\}$ or $\{O\}$.
The fuzzy TBoxes \mbox{$\{(B \sqsubseteq \E r.\!\ineg\!A) \geq 0.1\}$} and \mbox{$\{(B \sqsubseteq \E r.\triangle A) \geq 0.1\}$} cannot be expressed in \mLP when $\PhiF\setminus\Phi$ is $\{I\}$ or $\{O\}$.
\end{proposition}

\begin{proof}
	Let $\mA$ be any of the mentioned fuzzy ABoxes and $\mT$ any of the mentioned fuzzy TBoxes.
	
	Consider the case when $\PhiF \setminus \Phi = \{U\}$. 
	Let $\mI$ be the fuzzy interpretation illustrated below on the left: 
	\begin{center}
		\begin{tikzpicture}
		\node (x0) {};
		\node (x) [node distance=0.0cm, right of=x0] {};
		\node (I) [node distance=0.0cm, below of=x] {$\mI$};
		\node (I1) [node distance=3.0cm, right of=I] {$\,\mI'$};
		\node (u) [node distance=0.5cm, below of=I] {$u: B_1$};
		\node (v) [node distance=1.5cm, below of=u] {$v: A_{\,0.9}$};
		\node (up) [node distance=0.5cm, below of=I1] {$\,u': B_1$};
		\node (vp) [node distance=1.5cm, below of=up] {$v': A_1$};
		\draw[->] (u) to node [left]{\footnotesize{0.9}} (v);
		\draw[->] (up) to node [left]{\footnotesize{0.9}} (vp);
		\end{tikzpicture}
	\end{center}	
	and specified as follows: 
	\begin{itemize}
		\item $\Delta^\mI = \{u,v\}$, $b^\mI = u$ for all $b \in \IN$,  
		\item $A^\mI(u) = 0$, $A^\mI(v) = 0.9$, $B^\mI(u) = 1$ and $B^\mI(v) = 0$, 
		\item $r^\mI(u,v) = 0.9$ and $r^\mI(x,y) = 0$ for the other pairs $\tuple{x,y}$, 
		\item $C^\mI(x) = 0$ and $s^\mI(x,y) = 0$ for all $C \in \CN \setminus \{A,B\}$, $s \in \RN \setminus \{r\}$ and $x,y \in \Delta^\mI$.
	\end{itemize}
	Let $\mIp$ be the fuzzy interpretation illustrated above on the right and specified analogously.
	Let $Z : \Delta^\mI \times \Delta^\mIp \to [0,1]$ be the function specified by $Z(u,u') = 1$, $Z(v,v') = 0.9$ and $Z(u,v') = Z(v,u') = 0$.
	It is easy to check that $Z$ is the greatest fuzzy $\Phi$-bisimulation between $\mI$ and $\mIp$. 
	Thus, $\mI$ and $\mIp$ are $\Phi$-bisimilar. 
	Observe that $\mI \models \mA$ iff $\mIp \not\models \mA$. 
	If $\mA$ could be expressed by a fuzzy ABox $\mA_2$ in \mLP, then 
	$\mI \models \mA_2$ iff $\mIp \not\models \mA_2$, but by Theorem~\ref{theorem: IFDMS}, 
	$\mI \models \mA_2$ iff $\mIp \models \mA_2$, leading to a contradiction. 
	
	Consider the case when $\PhiF \setminus \Phi = \{I\}$. Let $\mI_2$ be the fuzzy interpretation that differs from $\mI$ only in that $\Delta^{\mI_2} = \{u,v,w\}$, $A^{\mI_2}(w) = 1$ and $B^{\mI_2}(w) = 0$. Similarly, let $\mI'_2$ be the fuzzy interpretation that differs from $\mIp$ only in that $\Delta^{\mI'_2} = \{u',v',w'\}$, $A^{\mI'_2}(w') = 0.9$ and $B^{\mI'_2}(w') = 0$. 
	Let $Z_2 : \Delta^{\mI_2} \times \Delta^{\mI'_2} \to [0,1]$ be the function specified by $Z_2(u,u') = 1$, $Z_2(v,v') = Z_2(w,w') = 0.9$, $Z_2(v,w') = Z_2(w,v') = 1$ and $Z_2(x,x') = 0$ for the four remaining pairs $\tuple{x,x'} \in \Delta^{\mI_2} \times \Delta^{\mI'_2}$.
	It is easy to check that $Z_2$ is the greatest fuzzy $\Phi$-bisimulation between $\mI_2$ and~$\mI'_2$. 
	Thus, $\mI_2$ and $\mI'_2$ are $\Phi$-bisimilar. Using a similar argumentation as for the previous case, we can conclude that $\mA$ cannot be expressed in \mLP.
	Observe that $\mI_2 \models \mT$ iff $\mI'_2 \not\models \mT$. 
	If $\mT$ could be expressed by a fuzzy TBox $\mT_2$ in \mLP, then 
	$\mI_2 \models \mT_2$ iff $\mI'_2 \not\models \mT_2$, but by Theorem~\ref{theorem: UDKMS}, 
	$\mI_2 \models \mT_2$ iff $\mI'_2 \models \mT_2$, leading to a contradiction.

	Consider the case when $\PhiF \setminus \Phi = \{O\}$. Let $\mI_3$ and $\mI'_3$ be the fuzzy interpretations illustrated below: 
	\begin{center}
		\begin{tikzpicture}
		\node (x) {};
		\node (L) [node distance=0.0cm, below of=x] {};
		\node (Lp) [node distance=7.0cm, right of=L] {};
		\node (I) [node distance=1.0cm, right of=L] {$\mI_3$};
		\node (Ip) [node distance=1.0cm, right of=Lp] {$\,\mI'_3$};
		\node (u) [node distance=0.5cm, below of=L] {$u_1: B_1$};
		\node (v) [node distance=1.5cm, below of=u] {$v_1: A_{\,0.9}$};
		\node (u2) [node distance=2.0cm, right of=u] {$u_2: B_1$};
		\node (v2) [node distance=1.5cm, below of=u2] {$v_2: A_1$};
		\node (up) [node distance=0.5cm, below of=Lp] {$\,u'_1: B_1$};
		\node (vp) [node distance=1.5cm, below of=up] {$v'_1: A_1$};
		\node (u2p) [node distance=2.0cm, right of=up] {$u'_2: B_1$};
		\node (v2p) [node distance=1.5cm, below of=u2p] {$v'_2: A_{\,0.9}$};
		\draw[->] (u) to node [left]{\footnotesize{0.9}} (v);
		\draw[->] (u2) to node [left]{\footnotesize{0.9}} (v2);
		\draw[->] (up) to node [left]{\footnotesize{0.9}} (vp);
		\draw[->] (u2p) to node [left]{\footnotesize{0.9}} (v2p);
		\end{tikzpicture}
	\end{center}	
	They are specified similarly to the way for $\mI$ and $\mIp$, with $b^{\mI_3} = u_1$ and $b^{\mI'_3} = u'_1$ for all $b \in \IN$. 
	Let $Z_3 : \Delta^{\mI_3} \times \Delta^{\mI'_3} \to [0,1]$ be the function specified by:
	\begin{itemize}
		\item $Z_3(u_i,u'_j) = 1$ and $Z_3(u_i,v'_j) = Z_3(v_i,u'_j) = 0$ for all $\tuple{i,j} \in \{1,2\} \times \{1,2\}$, 
		\item $Z_3(v_1,v'_1) = Z_3(v_2,v'_2) = 0.9$ and 
		$Z_3(v_1,v'_2) = Z_3(v_2,v'_1) = 1$. 
	\end{itemize}
	It is easy to check that $Z_3$ is the greatest fuzzy $\Phi$-bisimulation between $\mI_3$ and~$\mI'_3$. 
	Thus, $\mI_3$ and $\mI'_3$ are $\Phi$-bisimilar. 
	Like the previous two cases, $\mI_3 \models \mA$ iff $\mI'_3 \not\models \mA$. 
	If $\mA$ could be expressed by a fuzzy ABox $\mA_3$ in \mLP, then 
	$\mI \models \mA_3$ iff $\mIp \not\models \mA_3$, but by Theorem~\ref{theorem: IFDMS}, 
	$\mI \models \mA_3$ iff $\mIp \models \mA_3$ (here note that, although $O \notin \Phi$, we have $b^{\mI_3} = u_1$ and $b^{\mI'_3} = u'_1$ for all $b \in \IN$, and assertions of the form $c \doteq d$, $c \not\doteq d$ or $R(c,d) \bowtie p$ are trivial), leading to a contradiction. 
	Using a similar argumentation as for the previous case, we can also conclude that $\mT$ cannot be expressed in \mLP.
	\myend
\end{proof}

\begin{remark}\label{remark: HDKSL}
	Reconsider the text of the above proof. 
	\modifiedA{Observe that $Z$ does not satisfy Conditions~\eqref{eq: FB 8} and~\eqref{eq: FB 9} (with $x = u$ and $x' = u'$), and thus $\mI$ and $\mIp$ cannot be used to deal with the case when $U \in \Phi$. One can think of $\mI_2$ and $\mIp_2$ as a repair of $\mI$ and $\mIp$ to cover the case when $U \in \Phi$, but this ``repair'' works only when $I \notin \Phi$. Similarly, $\mI_3$ and $\mIp_3$ are a further try to cover the case when $\{U,I\} \subseteq \Phi$, but the change spoils the satisfaction of Condition~\eqref{eq: FB 5}, which is related to the feature $O$ (nominals).}

	For the case when $\Phi = \PhiF \setminus \{U\}$, if $Z'$ is a crisp $\Phi$-bisimulation between $\mI$ and $\mIp$, then $Z' \leq Z$, hence $Z'(v,v') = 0$ and consequently, $Z'(u,u') = 0$ (i.e., $Z'(x,x') = 0$ for all $\tuple{x,x'} \in \Delta^\mI \times \Delta^\mIp$). 
	Thus, $\mI$ and $\mIp$ are $\Phi$-bisimilar but not strongly $\Phi$-bisimilar. 
	Similarly, for the case when $\Phi = \PhiF \setminus \{I\}$, $\mI_2$ and $\mI'_2$ are $\Phi$-bisimilar but not strongly $\Phi$-bisimilar. 
	For the case when $\Phi = \PhiF \setminus \{O\}$, $\mI_3$ and $\mI'_3$ are $\Phi$-bisimilar but not strongly $\Phi$-bisimilar. 
	\myend 
\end{remark}

\begin{proposition}\label{prop: JDGAS 3}
For $n \geq 2$ and $\Phi = \PhiF$, 
the following fuzzy ABoxes/TBoxes cannot be expressed in \mLP:
\[
\begin{array}{lcl}
\{(\geq\!n\,r.\!\ineg\!A)(a) \geq 0.1\} & & 
\{(\geq\!n\,r.\triangle A)(a) \geq 0.1\} \\[0.5ex]
\{(B \sqsubseteq\; \geq\!n\,r.\!\ineg\!A) \geq 0.1\} & & 
\{(B \sqsubseteq\; \geq\!n\,r.\triangle A) \geq 0.1\}.
\end{array}
\]
\end{proposition}

\begin{proof}
	Let $\mA$ be any of the mentioned fuzzy ABoxes and $\mT$ any of the mentioned fuzzy TBoxes. 
	We prove this proposition for the case when $n = 2$. The other cases are similar and omitted. 
	Let $\mI$ and $\mIp$ be the fuzzy interpretations illustrated below and specified similarly to the way in the proof of Proposition~\modifiedA{\ref{prop: JDGAS 2}}: 
	\begin{center}
		\begin{tikzpicture}
		\node (x0) {};
		\node (x) [node distance=0.0cm, right of=x0] {};
		\node (I) [node distance=0.0cm, below of=x] {$\mI$};
		\node (I1) [node distance=7.0cm, right of=I] {$\,\mI'$};
		\node (u) [node distance=0.5cm, below of=I] {$u: B_1$};
		\node (v1) [node distance=1.5cm, below of=u] {$v_1: A_{\,0.9}$};
		\node (v0) [node distance=2.0cm, left of=v1] {$v_0: A_{\,0.9}$};
		\node (v2) [node distance=2.0cm, right of=v1] {$v_2: A_1$};
		\node (up) [node distance=0.5cm, below of=I1] {$\,u': B_1$};
		\node (v1p) [node distance=1.5cm, below of=up] {$v'_1: A_1$};
		\node (v0p) [node distance=2.0cm, left of=v1p] {$v'_0: A_{\,0.9}$};
		\node (v2p) [node distance=2.0cm, right of=v1p] {$v'_2: A_1$};
		\draw[->] (u) to node [left]{\footnotesize{0.9}} (v0);
		\draw[->] (u) to node [left]{\footnotesize{0.9}} (v1);
		\draw[->] (u) to node [left]{\footnotesize{0.9}} (v2);
		\draw[->] (up) to node [left]{\footnotesize{0.9}} (v0p);
		\draw[->] (up) to node [left]{\footnotesize{0.9}} (v1p);
		\draw[->] (up) to node [left]{\footnotesize{0.9}} (v2p);
		\end{tikzpicture}
	\end{center}
	Let $Z : \Delta^\mI \times \Delta^\mIp \to [0,1]$ be the function specified by:
	\begin{itemize}
		\item $Z(u,u') = 1$, $Z(u,v'_i) = Z(v_i,u') = 0$ for all $0 \leq i \leq 2$, 
		\item $Z(v_i,v'_j) = 1$ for all $\tuple{i,j} \in \{0,1,2\} \times \{0,1,2\}$ with $A^\mI(v_i) = A^\mIp(v'_j)$, 
		\item $Z(v_i,v'_j) = 0.9$ for all $\tuple{i,j} \in \{0,1,2\} \times \{0,1,2\}$ with $A^\mI(v_i) \neq A^\mIp(v'_j)$.
	\end{itemize}
	It is easy to check that $Z$ is the greatest fuzzy $\Phi$-bisimulation between $\mI$ and $\mIp$. 
	Using a similar argumentation as for Proposition~\ref{prop: JDGAS 2}, we can conclude that $\mA$ and $\mT$  cannot be expressed in \mLP.
	\myend
\end{proof}

\begin{theorem}
	\mLPn and \DLP are strictly more expressive than \mLP 
	\begin{itemize}
		\item w.r.t.\ concepts for any $\Phi \subseteq \PhiF$, 
		\item w.r.t.\ ABoxes for any $\Phi \subseteq \PhiF$ such that $\{I,O,U\} \setminus \Phi \neq \emptyset$ or $\Phi$ contains some $Q_n$ with $n \geq 2$, 
		\item w.r.t.\ TBoxes for any $\Phi \subseteq \PhiF$ such that $\{I,O\} \setminus \Phi \neq \emptyset$ or $\Phi$ contains some $Q_n$ with $n \geq 2$. 
	\end{itemize} 
\end{theorem}

This theorem immediately follows from Propositions~\ref{prop: JDAJA}--\ref{prop: JDGAS 3}.
It remains open whether the second and third assertions of this theorem can be significantly strengthened, e.g., to any $\Phi \subseteq \PhiF$ without restrictions.


\section{Minimizing Finite Fuzzy Interpretations}
\label{section: minimization}

In this section, as an application of strong $\Phi$-bisimilarity, we study the problem of minimizing a finite fuzzy interpretation while preserving certain properties. 

Given a fuzzy interpretation $\mI$, by $\simPIn$ we denote the binary relation on $\Delta^\mI$ such that, for $x,x' \in \Delta^\mI$, $x \simPIn x'$ iff $x \simPn x'$. We call it the {\em strong $\Phi$-bisimilarity relation of~$\mI$}. By the assertions~\ref{item: HFHSJ2 1}--\ref{item: HFHSJ2 3} of Proposition~\ref{prop: HFHSJ 2}, $\simPIn$ is an equivalence relation on $\Delta^\mI$. By the assertion~\ref{item: HFHSJ2 4} of Proposition~\ref{prop: HFHSJ 2}, the function $Z : \Delta^\mI \times \Delta^\mI \to \{0,1\}$ specified by \mbox{$Z(x,x') =$ (if $x \simPIn x'$ then 1 else 0)} is the greatest crisp $\Phi$-bisimulation between $\mI$ and itself, which we also call the {\em greatest crisp $\Phi$-auto-bisimulation of~$\mI$}. 

\begin{definition}\label{def: HDJOS}
	Given a fuzzy interpretation $\mI$ and $\Phi \subseteq \{I,O,U\}$, the {\em quotient fuzzy interpretation} $\mIsimPn$ of $\mI$ w.r.t.\ the equivalence relation $\simPIn$ is specified as follows:\footnote{Formally, the quotient fuzzy interpretation of $\mI$ w.r.t.\ the equivalence relation $\simPIn$ should be denoted by~$\mI/_{\simPIn}$. We use $\mIsimPn$ instead to simplify the notation.}	 
	\begin{itemize}
		\item $\Delta^{\mIsimPn} = \{[x]_{\simPIn} \mid x \in \Delta^\mI \}$, where $[x]_{\simPIn}$ is the equivalence class of $x$ w.r.t.\ $\simPIn$, 
		\item $a^{\mIsimPn} = [a^\mI]_{\simPIn}$ for $a \in \IN$, 
		\item $A^{\mIsimPn}([x]_{\simPIn}) = A^\mI(x)$ for $A \in \CN$ and $x \in \Delta^\mI$, 
		\item $r^{\mIsimPn}([x]_{\simPIn},[y]_{\simPIn}) = \sup\{r^\mI(x,y') \mid y' \in [y]_{\simPIn}\}$ for $r \in \RN$ and $x,y \in \Delta^\mI$.
		\myend
	\end{itemize}
\end{definition}
To justify that Definition~\ref{def: HDJOS} is well specified, we need to show that:
\begin{enumerate}
	\item For every $A \in \CN$, $x \in \Delta^\mI$ and $x' \in [x]_{\simPIn}$, $A^\mI(x) = A^\mI(x')$.
	\item For every $r \in \RN$, $x,y \in \Delta^\mI$ and $x' \in [x]_{\simPIn}$, 
	\[ 
	\sup\{r^\mI(x,y') \mid y' \in [y]_{\simPIn}\} = 
	\sup\{r^\mI(x',y') \mid y' \in [y]_{\simPIn}\}.
	\]
\end{enumerate}
Let $Z$ be $\simPIn$. Then, the first assertion follows from Condition~\eqref{eq: FB 2} and the assumption that $Z(x,x') = 1$.
The second one follows from Conditions~\eqref{eq: FB 3}, \eqref{eq: FB 4} and the assumption $Z(x,x') = 1$. 

\begin{example}\label{example: HDKSL2-3}
	Let $\RN = \{r\}$, $\CN = \{A\}$ and $\IN = \{a\}$. Consider the fuzzy interpretation $\mI$ illustrated below and specified similarly as in Example~\ref{example: HDKSL}, with $a^\mI = u$:
	\begin{center}		
		\begin{tikzpicture}
		\node (I) {};
		\node (u) [node distance=0.0cm, below of=I] {$u:A_0$};
		\node (ub) [node distance=1.5cm, below of=u] {$v_2:A_{\,0.8}$};
		\node (v) [node distance=1.8cm, left of=ub] {$v_1:A_{\,0.7}$};
		\node (w) [node distance=1.8cm, right of=ub] {$v_3:A_{\,0.8}$};
		\draw[->] (u) to node [left]{\footnotesize{0.5}} (v);
		\draw[->] (u) to node [right]{\footnotesize{0.6}} (ub);
		\draw[->] (u) to node [right]{\footnotesize{0.3}} (w);
		\node (vp) [node distance=3cm, right of=w] {$v'_1:A_{\,0.7}$};
		\node (upb) [node distance=1.8cm, right of=vp] {$v'_2:A_{\,0.8}$};
		\node (up) [node distance=1.5cm, above of=upb] {$u':A_0$};
		\draw[->] (up) to node [left]{\footnotesize{0.5}} (vp);
		\draw[->] (up) to node [right]{\footnotesize{0.6}} (upb);
		\end{tikzpicture}
	\end{center}
	
	\begin{itemize}
		\item Case $\Phi \subseteq \{U\}$: 
		We have
		\[
			\simPIn\ \, =\, \{\tuple{x,x} \mid x \in \Delta^\mI\} \cup 
			\{\tuple{u,u'}, \tuple{u',u}, \tuple{v_1,v'_1}, \tuple{v'_1,v_1}\} \cup \{\tuple{x,x'} \mid x, x' \in \{v_2,v_3,v'_2\}\}
		\]
		and $\mIsimPn$ has the following form, with 
		$a^{\mIsimPn} = \{u,u'\}$:
		\begin{center}		
			\begin{tikzpicture}
			\node (I) {};
			\node (u) [node distance=0.0cm, below of=I] {$\{u,u'\}:A_0$};
			\node (ub) [node distance=1.5cm, below of=u] {};
			\node (v1) [node distance=1.8cm, left of=ub] {$\{v_1,v'_1\}:A_{\,0.7}$};
			\node (v2) [node distance=1.8cm, right of=ub] {$\{v_2,v_3,v'_2\}:A_{\,0.8}$};
			\draw[->] (u) to node [left]{\footnotesize{0.5}} (v1);
			\draw[->] (u) to node [right]{\footnotesize{0.6}} (v2);
			\end{tikzpicture}
		\end{center}
		
		\item Case $\{O\} \subseteq \Phi \subseteq \{O,U\}$: 
		We have
		\[ 
		\simPIn \;\, =\, \{\tuple{x,x} \mid x \in \Delta^\mI\} \cup \{\tuple{v_1,v'_1}, \tuple{v'_1,v_1}\} \cup 
		\{\tuple{x,x'} \mid x, x' \in \{v_2,v_3,v'_2\}\}
		\]
		and $\mIsimPn$ has the following form, with 
		$a^{\mIsimPn} = \{u\}$:
		\begin{center}		
			\begin{tikzpicture}
			\node (I) {};
			\node (u) [node distance=0.0cm, below of=I] {$\{u\}:A_0$};
			\node (up) [node distance=5.0cm, right of=u] {$\{u'\}:A_0$};
			\node (v1) [node distance=1.5cm, below of=u] {$\{v_1,v'_1\}:A_{\,0.7}$};
			\node (v2) [node distance=1.5cm, below of=up] {$\{v_2,v_3,v'_2\}:A_{\,0.8}$};
			\draw (u)  edge[->, left] node{\footnotesize{0.5}} (v1)
			(u)  edge[->, below=2pt, pos=.15] node{\footnotesize{0.6}} (v2)
			(up) edge[->, below=2pt, pos=.15] node{\footnotesize{0.5}} (v1)
			(up) edge[->, right] node{\footnotesize{0.6}} (v2);
			\end{tikzpicture}
		\end{center}
		
		\item Case $I \in \Phi\,$: We have $ \simPIn\ = \{\tuple{x,x} \mid x \in \Delta^\mI\}$ and $\mIsimPn$ has the same form as~$\mI$. 
		\myend
	\end{itemize}
\end{example}

\begin{lemma}\label{lemma: HDAMA}
	Let $\Phi \subseteq \{I,O,U\}$, $\mI$ be a fuzzy interpretation that is image-finite w.r.t.~$\Phi$, and \mbox{$Z: \Delta^\mI \times \Delta^{\mIsimPn} \to \{0,1\}$} be specified by \mbox{$Z(x,[x'']_{\simPIn}) =$ (if $x \in [x'']_{\simPIn}$ then 1 else 0)}. Then, $Z$ is a crisp $\Phi$-bisimulation between $\mI$ and $\mIsimPn$. 
	It is also a crisp $(\Phi \cup \{U\})$-bisimulation between $\mI$ and $\mIsimPn$.
\end{lemma}

\begin{proof}
	We need to prove Conditions~\eqref{eq: FB 2}--\eqref{eq: FB 9} (regardless of whether $U \in \Phi$) for $\mIp = \mIsimPn$. 
	Without loss of generality, assume that $Z(x,x') = 1$, which means $x' = [x]_{\simPIn}$. 
	
	\begin{itemize}
		\item Condition~\eqref{eq: FB 2} directly follows from Definition~\ref{def: HDJOS}. 
		
		\item Consider Condition~\eqref{eq: FB 3} and take $y' = [y]_{\simPIn}$. By Definition~\ref{def: HDJOS}, $R^\mI(x,y) \leq R^{\mIsimPn}([x]_{\simPIn},[y]_{\simPIn})$. We also have $Z(y,y') = 1$. 
		
		\item Consider Condition~\eqref{eq: FB 4}. Since $\mI$ is image-finite, by Definition~\ref{def: HDJOS}, $R^{\mIsimPn}([x]_{\simPIn},y') = \max\{R^\mI(x,y) \mid y \in y'\}$. Hence, there exists $y \in y' \subseteq \Delta^\mI$ such that $R^{\mIsimPn}([x]_{\simPIn},y') = R^\mI(x,y)$. We also have $Z(y,y') = 1$.  
		
		\item Consider Condition~\eqref{eq: FB 5} for the case $O \in \Phi$. 
		If $x = a^\mI$, then by Definition~\ref{def: HDJOS}, $x' = [x]_{\simPIn} = [a^\mI]_{\simPIn} = a^\mIp$.  
		Conversely, if $x' = a^\mIp$, then $x \simPIn a^\mI$ and, by Condition~\eqref{eq: FB 5} with $Z$, $x'$ and $\mIp$ replaced by $\simPIn$, $a^\mI$ and $\mI$, respectively, we can derive that $x = a^\mI$. 
		
		\item Condition~\eqref{eq: FB 8} holds because we can take $y' = [y]_{\simPIn}$.
		\item Condition~\eqref{eq: FB 9} holds because we can take any $y \in y'$.
		\myend
	\end{itemize}
\end{proof}

\begin{corollary}\label{cor: JFWKA}
	Suppose that $\IN \neq \emptyset$, $\Phi \subseteq \{I,O,U\}$ and $\mI$ is a fuzzy interpretation that is image-finite w.r.t.~$\Phi$. Then, $\mI$ and $\mIsimPn$ are both strongly $\Phi$-bisimilar and strongly $(\Phi \cup \{U\})$-bisimilar.
\end{corollary}

This corollary immediately follows from Lemma~\ref{lemma: HDAMA}. 

\begin{corollary}\label{cor: JFWKA2}
	Let $\Phi \subseteq \{I,O,U\}$, $\mI$ be a finite fuzzy interpretation, $\mT$ a fuzzy TBox and $\mA$ a fuzzy ABox in \mLPn. Then:
	\begin{enumerate}
		\item $\mI \models \mT$ iff $\mIsimPn \models \mT$, 
		\item if $O \in \Phi$ or $\mA$ consists of only fuzzy assertions of the form $C(a) \bowtie p$, then $\mI \models \mA$ iff $\mIsimPn \models \mA$.  
	\end{enumerate}
\end{corollary}

\begin{proof}
	Without loss of generality, we assume that $\IN \neq \emptyset$. 
	By Corollary~\ref{cor: JFWKA}, $\mI$ and $\mIsimPn$ are both strongly $\Phi$-bisimilar and strongly $(\Phi \cup \{U\})$-bisimilar.
	Also recall that every finite fuzzy interpretation is witnessed w.r.t.\ \mLPn and \mLPUn. 
	
	By Theorem~\ref{theorem: UDKMS2}, all fuzzy TBoxes in \mLPUn are invariant under strong $(\Phi \cup \{U\})$-bisimilarity. In particular, $\mT$ is invariant under strong $(\Phi \cup \{U\})$-bisimilarity. Hence, $\mI \models \mT$ iff $\mIsimPn \models \mT$. 
	
	Consider the second assertion and suppose that $O \in \Phi$ or $\mA$ consists of only fuzzy assertions of the form $C(a) \bowtie p$. By Theorem~\ref{theorem: IFDMS2}, $\mA$ is invariant under strong $\Phi$-bisimilarity. Hence, $\mI \models \mA$ iff $\mIsimPn \models \mA$.
	\myend
\end{proof}

In the following theorem, the term ``\modifiedB{minimal}'' is understood w.r.t.\ the size of the domain of the considered fuzzy interpretation. 

\begin{theorem}\label{theorem: HSJAO}
	Let $\Phi \subseteq \{I,O,U\}$ and let $\mI$ be a finite fuzzy interpretation. Then:
	\begin{enumerate}
		\item $\mIsimPn$ is a \modifiedB{minimal} fuzzy interpretation that validates the same set of fuzzy GCIs in \mLPn as $\mI$, 
		\item if $\IN \neq \emptyset$ and either $U \in \Phi$ or $\mI$ is connected w.r.t.\ $\Phi$, then:
		\begin{enumerate}
			\item\label{item: GHDJW} $\mIsimPn$ is a \modifiedB{minimal} fuzzy interpretation strongly $\Phi$-bisimilar to $\mI$, 
			\item\label{item: GHDJW2} $\mIsimPn$ is a \modifiedB{minimal} fuzzy interpretation that validates the same set of fuzzy assertions of the form $C(a) \bowtie p$ in \mLPn as $\mI$.
		\end{enumerate}
	\end{enumerate}
\end{theorem}

\begin{proof}
	Without loss of generality, we assume that $\IN \neq \emptyset$. 
	By Corollaries~\ref{cor: JFWKA}	and~\ref{cor: JFWKA2}, $\mIsimPn$ validates the same set of fuzzy GCIs in \mLPn as $\mI$, is strongly $\Phi$-bisimilar to $\mI$, and validates the same set of fuzzy assertions of the form $C(a) \bowtie p$ in \mLPn as $\mI$. It remains to justify its minimality. 
	
	Since $\mI$ is finite, $\mIsimPn$ is also finite. 
	Let $\Delta^{\mIsimPn} = \{v_1,\ldots,v_n\}$, where $v_1,\ldots,v_n$ are pairwise distinct and $v_i = [x_i]_{\simPIn}$ with $x_i \in \Delta^\mI$, for $1 \leq i \leq n$. By Lemma~\ref{lemma: HDAMA}, $x_i \simPn v_i$ for all $1 \leq i \leq n$. Let $i$ and $j$ be arbitrary indexes such that $1 \leq i,j\leq n$ and $i \neq j$. We have that $x_i \not\simPn x_j$, and by Proposition~\ref{prop: HFHSJ 2}, it follows that $v_i \not\simPn v_j$. Consequently, by Corollary~\ref{cor: fG H-M 1 c}, $v_i \not\equivPdp v_j$. Therefore, there exists a concept $C_{i,j}$ of \DLPp such that $C_{i,j}^{\mIsimPn}(v_i) \neq C_{i,j}^{\mIsimPn}(v_j)$. Let $D_{i,j} = \triangle(C_{i,j} \to C_{i,j}^{\mIsimPn}(v_i))$ if $C_{i,j}^{\mIsimPn}(v_i) < C_{i,j}^{\mIsimPn}(v_j)$, and $D_{i,j} = \triangle(C_{i,j}^{\mIsimPn}(v_i) \to C_{i,j})$ otherwise (i.e., when $C_{i,j}^{\mIsimPn}(v_i) > C_{i,j}^{\mIsimPn}(v_j)$). We have that $D_{i,j}^{\mIsimPn}(v_i) = 1$ and $D_{i,j}^{\mIsimPn}(v_j) = 0$. Let $D_i = D_{i,1} \mand\ldots\mand D_{i,i-1} \mand D_{i,i+1} \mand\ldots\mand D_{i,n}$. We have that $D_i^{\mIsimPn}(v_i) = 1$ and $D_i^{\mIsimPn}(v_j) = 0$. Let $E = D_1 \mor\ldots\mor D_n$ and $E_i = D_1 \mor\ldots\mor D_{i-1} \mor D_{i+1} \mor\ldots\mor D_n$. 
	
	We have that $\mIsimPn$ validates the fuzzy GCI $(\top \sqsubseteq E) \geq 1$ but does not validate $(\top \sqsubseteq E_i) \geq 1$ for any $1 \leq i \leq n$. Any other fuzzy interpretation with such properties must
	have at least $n$ elements in the domain. That is, $\mIsimPn$ is a \modifiedB{minimal} fuzzy interpretation that validates the same set of fuzzy GCIs in \mLPn as $\mI$. 
	
	Consider the second assertion of the theorem and suppose that either $U \in \Phi$ or $\mI$ is connected w.r.t.\ $\Phi$. Let $Z: \Delta^\mI \times \Delta^{\mIsimPn} \to \{0,1\}$ be specified by \mbox{$Z(x,[x'']_{\simPIn}) =$ (if $x \in [x'']_{\simPIn}$ then 1 else 0)}. By Lemma~\ref{lemma: HDAMA}, $Z$ is a crisp $\Phi$-bisimulation between $\mI$ and $\mIsimPn$. 
	
	\begin{itemize}
		\item Consider the assertion~\ref{item: GHDJW} of the theorem and let $\mIp$ be a fuzzy interpretation strongly $\Phi$-bisimilar to~$\mI$. There exists a crisp $\Phi$-bisimulation $Z'$ between $\mIp$ and $\mI$ such that $Z'(a^\mIp, a^\mI) = 1$ for all $a \in \IN$. Let $Z'' = Z' \circ Z$. By Proposition~\ref{prop: HFHSJ 2}, $Z''$ is a crisp $\Phi$-bisimulation between $\mIp$ and $\mIsimPn$. Furthermore, 
		\begin{equation}
		Z''(a^\mIp, a^{\mIsimPn}) = 1 \textrm{ for all $a \in \IN$.} \label{eq: HFOWA}
		\end{equation}
		
		\begin{itemize}
			\item Consider the case when $\mI$ is connected w.r.t.\ $\Phi$. Thus, $\mIsimPn$ is also connected w.r.t.\ $\Phi$. By Condition~\eqref{eq: FB 4} (with $\mI$ and $\mIp$ in Condition~\eqref{eq: FB 4} replaced by $\mIp$ and $\mIsimPn$, respectively), it follows from \eqref{eq: HFOWA} that, for every $1 \leq i \leq n$, there exists $u_i \in \Delta^\mIp$ such that $Z''(u_i,v_i) = 1$.
			\item Consider the case when $U \in \Phi$. Since $\IN \neq \emptyset$, by \eqref{eq: HFOWA} and Condition~\eqref{eq: FB 9} (with $\mI$ and $\mIp$ in Condition~\eqref{eq: FB 9} replaced by $\mIp$ and $\mIsimPn$, respectively), for every $1 \leq i \leq n$, there exists $u_i \in \Delta^\mIp$ such that $Z''(u_i,v_i) = 1$.
		\end{itemize}
		
		Thus, $u_i \simPn v_i$ for all $1 \leq i \leq n$. Since $v_i \not\simPn v_j$ for any $i \neq j$, it follows that $u_i \not\simPn u_j$ for any $i \neq j$. Therefore the cardinality of $\Delta^\mIp$ is greater than or equal to $n$. 
		
		\item Consider the assertion~\ref{item: GHDJW2} of the theorem and let $\mIp$ be a fuzzy interpretation that validates the same set of fuzzy assertions of the form $C(a) \bowtie p$ in \mLPn as $\mI$. Thus, $a^\mIp \equivPdp a^\mI$ for all $a \in \IN$. Without loss of generality, assume that $\mIp$ is finite. By Corollary~\ref{cor: fG H-M 4 c}, $\mIp$ and $\mI$ are strongly $\Phi$-bisimilar. 
		By the assertion~\ref{item: GHDJW} proved above, it follows that the cardinality of $\Delta^\mIp$ is greater than or equal to $n$. 
		\myend
	\end{itemize}
\end{proof}

Given a fuzzy interpretation $\mI$, we say that an individual $x \in \Delta^\mI$ is {\em $\Phi$-reachable} (from a named individual) if there exist $a \in \IN$, $x_0,\ldots,x_n \in \Delta^\mI$ and basic roles $R_1,\ldots,R_n$ w.r.t.~$\Phi$ such that $x_0 = a^\mI$, $x_n = x$ and $R_i^\mI(x_{i-1},x_i) > 0$ for all $1 \leq i \leq n$. 

\begin{corollary}\label{cor: HFJHW}
Let $\Phi \subseteq \{I,O\}$ and suppose $\IN \neq \emptyset$. Let $\mI$ be a finite fuzzy interpretation and $\mIp$ the fuzzy interpretation obtained from $\mI$ by deleting from the domain all $\Phi$-unreachable individuals and restricting the interpretation function accordingly. Then:
\begin{enumerate}
\item $\mIpsimPn$ is a \modifiedB{minimal} fuzzy interpretation strongly $\Phi$-bisimilar to $\mI$, 
\item $\mIpsimPn$ is a \modifiedB{minimal} fuzzy interpretation that validates the same set of fuzzy assertions of the form $C(a) \bowtie p$ in \mLPn as $\mI$.
\end{enumerate}
\end{corollary}

This corollary follows from Theorem~\ref{theorem: HSJAO}, Proposition~\ref{prop: HFHSJ 2}, Theorem~\ref{theorem: IFDMS2} and the observations that $\mIp$ is connected w.r.t.~$\Phi$ and strongly $\Phi$-bisimilar to $\mI$ by using the following crisp $\Phi$-bisimulation \[ Z = \lambda \tuple{x,x'} \in \Delta^\mI \times \Delta^\mIp.(\textrm{if $x=x'$ then 1 else 0}).\]

\begin{remark}
	In this section, we have studied minimizing finite fuzzy interpretations for the case when $\Phi \subseteq \{I,O,U\}$. In~\cite{BSDL-INS,thesis-ARD}, Nguyen and Divroodi studied the problem of minimizing (traditional) interpretations also for the cases when the considered interpretation is infinite or $\Phi \cap \{Q,\Self\} \neq \emptyset$. For dealing with the case when $\Phi \cap \{Q,\Self\} \neq \emptyset$, they introduced the notion of QS-interpretation that allows ``multi-edges'' and keeps information about ``self-edges'' (where ``edge'' is understood as an instance of a role). Minimizing fuzzy interpretations can be extended for those cases using their approach. Here, we have restricted to finite (fuzzy) interpretations and the case when $\Phi \subseteq \{I,O,U\}$ to increase the readability. 
	\myend
\end{remark}


\section{Conclusions}
\label{sec: conc}

We have defined fuzzy bisimulation and bisimilarity for a large class of fuzzy DLs under the G\"odel semantics, as well as crisp bisimulation and strong bisimilarity for such logics extended with involutive negation or the Baaz projection operator. We have formulated and proved results on invariance of concepts under fuzzy/crisp bisimulation and conditional invariance of fuzzy TBoxes/ABoxes under bisimilarity and strong bisimilarity. We have also formulated and proved results on the Hennessy-Milner property of the introduced bisimulations. 

As mentioned in the Introduction, we use elementary conditions 
instead of the ones based on relational composition for defining bisimulations. They are suitable for dealing with number restrictions. Thus, our notion of fuzzy bisimulation is different in nature from the one introduced by Fan~\cite{Fan15} for G\"odel monomodal logics. Furthermore, in comparison with~\cite{Fan15}, not only are the logics studied by us expressive, with various role and concept constructors, we also study invariance of fuzzy TBoxes/ABoxes and our theorems on the Hennessy-Milner property are formulated and proved  for witnessed and modally saturated interpretations, which are more general than image-finite interpretations. 

In addition, we have provided new results on using fuzzy bisimulations to separate the expressive powers of fuzzy DLs and using strong bisimilarity to minimize fuzzy interpretations while preserving validity of fuzzy axioms/assertions.

As far as we know, this is the first time bisimulation and bisimilarity are defined and studied for fuzzy DLs under the G\"odel semantics. Our notions and results may have potential applications to concept learning in fuzzy DLs. 
\modifiedA{As another open problem, one may try to study bisimulation and bisimilarity in fuzzy DLs under the {\L}ukasiewicz or Product semantics.}



\bibliography{BSfDL}
\bibliographystyle{plain}

\appendix

\section{\modifiedA{Proofs}}
\label{appendix A}

\modifiedA{In this appendix, we present the proofs of Lemma~\ref{lemma: GDHAW}, Theorem~\ref{theorem: fG H-M} and Theorem~\ref{theorem: fG H-M 2}. To increase the readability, we recall the lemma and theorems before providing their proofs.} 

\bigskip

\noindent
\textbf{Lemma~\ref{lemma: GDHAW}.}
{\em \TextLemmaGDHAWp}

\vspace{-1em}

\begin{proof}
	We prove this lemma by induction on the structures of $C$ and~$R$.
	First, consider the assertion~\eqref{eq: GDHAW 2}. Let $x,y \in \Delta^\mI$ and $x' \in \Delta^\mIp$. It is sufficient to show that there exists $y' \in \Delta^\mIp$ such that
	\begin{equation}\label{eq: HDHAK}
	Z(x,x') \fand R^\mI(x,y) \leq Z(y,y') \fand R^\mIp(x',y').
	\end{equation}
	The base case occurs when $R$ is a basic role w.r.t.~$\Phi$ and follows from~\eqref{eq: FB 3}. The induction steps are given below.
	\begin{itemize}
		\item Case $R = R_1 \circ R_2$: Since $\mI$ is witnessed w.r.t.\ \mLP, there exists $z \in \Delta^\mI$ such that \mbox{$R^\mI(x,y) = R_1^\mI(x,z) \fand R_2^\mI(z,y)$}. By the inductive assumption of~\eqref{eq: GDHAW 2}, there exist $z'$ and $y'$ such that:
		\begin{eqnarray*}
			Z(x,x') \fand R_1^\mI(x,z) & \leq & Z(z,z') \fand R_1^\mIp(x',z') \\
			Z(z,z') \fand R_2^\mI(z,y) & \leq & Z(y,y') \fand R_2^\mIp(z',y').
		\end{eqnarray*}
		Thus, 
		\begin{eqnarray*}
			Z(x,x') \fand R^\mI(x,y) & = & Z(x,x') \fand R_1^\mI(x,z) \fand R_2^\mI(z,y) \\
			& \leq & Z(z,z') \fand R_1^\mIp(x',z') \fand R_2^\mI(z,y) \\
			& \leq & Z(z,z') \fand R_2^\mI(z,y) \fand R_1^\mIp(x',z') \\
			& \leq & Z(y,y') \fand R_2^\mIp(z',y') \fand R_1^\mIp(x',z') \\
			& \leq & Z(y,y') \fand R^\mIp(x',y').
		\end{eqnarray*}
		
		\item Case $R = R_1 \sqcup R_2$: Without loss of generality, suppose $R^\mI(x,y) = R_1^\mI(x,y) \geq R_2^\mI(x,y)$. By the inductive assumption of~\eqref{eq: GDHAW 2}, there exists $y' \in \Delta^\mIp$ such that 
		\[
		Z(x,x') \fand R_1^\mI(x,y) \leq Z(y,y') \fand R_1^\mIp(x',z').
		\] 
		Thus, 
		\begin{eqnarray*}
			Z(x,x') \fand R^\mI(x,y) & = & Z(x,x') \fand R_1^\mI(x,y) \\
			& \leq & Z(y,y') \fand R_1^\mIp(x',y') \\
			& \leq & Z(y,y') \fand R^\mIp(x',y').
		\end{eqnarray*}
		
		\item Case $R = S^*$: Since $\mI$ is witnessed w.r.t.\ \mLP, there exist $x_0, \ldots, x_k \in \Delta^\mI$ such that $x_0 = x$, $x_k = y$ and 
		\[ R^\mI(x,y) = S^\mI(x_0,x_1) \fand\cdots\fand S^\mI(x_{k-1},x_k). \]
		Let $x'_0 = x'$. By the inductive assumption of~\eqref{eq: GDHAW 2}, there exists $x'_1,\ldots,x'_k \in \Delta^\mIp$ such that  
		\[ Z(x_i,x'_i) \fand S^\mI(x_i,x_{i+1}) \leq Z(x_{i+1},x'_{i+1}) \fand S^\mIp(x'_i,x'_{i+1}) \]
		for all $0 \leq i < k$. Thus,
		\begin{eqnarray*}
			Z(x_0,x'_0) \fand R^\mI(x_0,x_k) & = & Z(x_0,x'_0) \fand S^\mI(x_0,x_1) \fand\cdots\fand S^\mI(x_{k-1},x_k) \\ 
			& \leq & Z(x_k,x'_k) \fand S^\mIp(x'_0,x'_1) \fand\cdots\fand S^\mIp(x'_{k-1},x'_k) \\
			& \leq & Z(x_k,x'_k) \fand R^\mIp(x'_0,x'_k).
		\end{eqnarray*}
		Taking $y' = x'_k$, we obtain~\eqref{eq: HDHAK}.
		
		\item Case $R = (D?)$: If $x \neq y$, then $R^\mI(x,y) = 0$ and \eqref{eq: HDHAK} clearly holds. Suppose $x = y$ and take $y' = x'$. By the inductive assumption of~\eqref{eq: GDHAW 1}, \mbox{$Z(x,x') \leq (D^\mI(x) \fequiv D^\mIp(x'))$}. Hence,  
		\[ Z(x,x') \fand D^\mI(x) \leq Z(x,x') \fand D^\mIp(x'), \] 
		which implies \eqref{eq: HDHAK}.
		
		\item Case $U \in \Phi$ and $R = U$: With $R = U$, \eqref{eq: HDHAK} is equivalent to $Z(x,x') \leq Z(y,y')$. The existence of such a $y'$ is guaranteed by~\eqref{eq: FB 8}. 
	\end{itemize}
	
	The assertion~\eqref{eq: GDHAW 3} can be proved analogously as for~\eqref{eq: GDHAW 2}.
	
	Consider the assertion~\eqref{eq: GDHAW 1}. 
	It clearly holds when $C^\mI(x) = C^\mIp(x')$. Suppose $C^\mI(x) > C^\mIp(x')$ (the case $C^\mI(x) < C^\mIp(x')$ is similar and omitted). Thus, $C^\mIp(x') = (C^\mI(x) \fequiv C^\mIp(x'))$. 
	The cases when $C$ is of the form $p$ or $A$ are trivial. The cases when $C$ is of the form $\lnot D$ or $D \mor E$ are omitted (see Remark~\ref{remark: OFHSJ}). 
	\begin{itemize}
		\item Case $C = D \mand E$: We have $C^\mI(x) = D^\mI(x) \fand E^\mI(x)$ and \mbox{$C^\mIp(x') = D^\mIp(x') \fand E^\mIp(x')$}. By the inductive assumption of~\eqref{eq: GDHAW 1}, 
		\begin{eqnarray}
		Z(x,x') & \leq & (D^\mI(x) \fequiv D^\mIp(x')) \label{eq: JDGSH 1} \\
		Z(x,x') & \leq & (E^\mI(x) \fequiv E^\mIp(x')). \label{eq: JDGSH 2}
		\end{eqnarray}
		
		\modifiedA{Without loss of generality, assume that $D^\mIp(x') \leq E^\mIp(x')$. Since $C^\mI(x) > C^\mIp(x')$, it follows that $D^\mI(x) \fand E^\mI(x) > D^\mIp(x')$. Hence, by~\eqref{eq: JDGSH 1}, $Z(x,x') \leq D^\mIp(x')$. Therefore,}
		\[ Z(x,x') \leq D^\mIp(x') = D^\mIp(x') \fand E^\mIp(x') = C^\mIp(x') = (C^\mI(x) \fequiv C^\mIp(x')). \] 
		
		\item Case $C = (D \to E)$: By the inductive assumption of~\eqref{eq: GDHAW 1}, we also have~\eqref{eq: JDGSH 1} and~\eqref{eq: JDGSH 2}. Since $C^\mIp(x') < C^\mI(x) \leq 1$, we have $C^\mIp(x') = E^\mIp(x') < D^\mIp(x')$. 
		If $E^\mI(x) \neq E^\mIp(x')$, then by~\eqref{eq: JDGSH 2}, $Z(x,x') \leq E^\mIp(x') = C^\mIp(x')$, and hence $Z(x,x') \leq (C^\mI(x) \fequiv C^\mIp(x'))$. Suppose $E^\mI(x) = E^\mIp(x')$. Since $C^\mI(x) > C^\mIp(x') = E^\mIp(x') = E^\mI(x)$, we must have that $D^\mI(x) \leq E^\mI(x)$. Since  
		$D^\mI(x) \leq E^\mI(x) = E^\mIp(x') < D^\mIp(x')$, by~\eqref{eq: JDGSH 1}, $Z(x,x') \leq D^\mI(x)$. We have that 
		\[ Z(x,x') \leq D^\mI(x) \leq E^\mI(x) = E^\mIp(x') = C^\mIp(x') = (C^\mI(x) \fequiv C^\mIp(x')). \]
		
		\item Case $C = \E R.D$: Since $\mI$ is witnessed w.r.t.\ \mLP, there exists $y \in \Delta^\mI$ such that 
		\begin{equation}
		C^\mI(x) = R^\mI(x,y) \fand D^\mI(y). \label{eq: HJKAQ 1}
		\end{equation}
		By the inductive assumption of~\eqref{eq: GDHAW 2}, there exists $y' \in \Delta^\mIp$ such that 
		\begin{equation}
		Z(x,x') \fand R^\mI(x,y) \leq Z(y,y') \fand R^\mIp(x',y'). \label{eq: HJKAQ 2}
		\end{equation}
		By definition, 
		\begin{equation}
		R^\mIp(x',y') \fand D^\mIp(y') \leq C^\mIp(x'). \label{eq: HJKAQ 2a}
		\end{equation}
		
		By the inductive assumption of~\eqref{eq: GDHAW 1}, 
		\begin{equation}
		Z(y,y') \leq (D^\mI(y) \fequiv D^\mIp(y')). \label{eq: HJKAQ 3}
		\end{equation}

		We have the following:
		\[
		\begin{array}{lclcl}
		Z(x,x') & \leq & R^\mI(x,y) \fto Z(y,y') \fand R^\mIp(x',y') && \textrm{(by \eqref{eq: HJKAQ 2} and \eqref{fop: HSDJW 3})} \\
		        & \leq & R^\mI(x,y) \fto (D^\mI(y) \fequiv D^\mIp(y')) \fand R^\mIp(x',y') && \textrm{(by \eqref{eq: HJKAQ 3} and \eqref{fop: HSDJW 2})} \\
		        & \leq & R^\mI(x,y) \fto (D^\mI(y) \fequiv R^\mIp(x',y') \fand D^\mIp(y')) && \textrm{(by \eqref{fop: HSDJW 4} and \eqref{fop: HSDJW 2})} \\
		        & \leq & R^\mI(x,y) \fand D^\mI(y) \fto R^\mIp(x',y') \fand D^\mIp(y') && \textrm{(by \eqref{fop: HSDJW 5})} \\
		        & \leq & C^\mI(x) \fto C^\mIp(x') && \textrm{(by \eqref{eq: HJKAQ 1}, \eqref{eq: HJKAQ 2a} and \eqref{fop: HSDJW 2})} \\
		        & \leq & C^\mI(x) \fequiv C^\mIp(x') && \textrm{(since $C^\mI(x) > C^\mIp(x')$).}
		\end{array}
		\]
		
		\item Case $C = \V R.D$: Since $\mIp$ is witnessed w.r.t.\ \mLP, there exists $y' \in \Delta^\mIp$ such that 
		\begin{equation}
		C^\mIp(x') = (R^\mIp(x',y') \fto D^\mIp(y')). \label{eq: KDNSJ 1}
		\end{equation}
		By the inductive assumption of~\eqref{eq: GDHAW 3}, there exists $y \in \Delta^\mI$ such that 
		\begin{equation}
		Z(x,x') \fand R^\mIp(x',y') \leq Z(y,y') \fand R^\mI(x,y). \label{eq: KDNSJ 2}
		\end{equation}
		By definition, 
		\begin{equation}
		C^\mI(x) \leq (R^\mI(x,y) \fto D^\mI(y)). \label{eq: KDNSJ 2a}
		\end{equation}
		By the inductive assumption of~\eqref{eq: GDHAW 1}, 
		\begin{equation}
		Z(y,y') \leq (D^\mI(y) \fequiv D^\mIp(y')). \label{eq: KDNSJ 3}
		\end{equation}
		
		We have the following:
		\[
		\begin{array}{lclcl}
		Z(x,x') & \leq & R^\mIp(x',y') \fto Z(y,y') \fand R^\mI(x,y) && \textrm{(by \eqref{eq: KDNSJ 2} and \eqref{fop: HSDJW 3})} \\
		        & \leq & R^\mIp(x',y') \fto (D^\mI(y) \fequiv D^\mIp(y')) \fand R^\mI(x,y) && \textrm{(by \eqref{eq: KDNSJ 3} and \eqref{fop: HSDJW 2})} \\
		        & \leq & R^\mIp(x',y') \fto ((R^\mI(x,y) \fto D^\mI(y)) \fto D^\mIp(y'))  && \textrm{(by \eqref{fop: HSDJW 4b} and \eqref{fop: HSDJW 2})} \\
		        & \leq & (R^\mI(x,y) \fto D^\mI(y)) \fto (R^\mIp(x',y') \fto D^\mIp(y'))  && \textrm{(by \eqref{fop: HSDJW 5b})} \\
		        & \leq & C^\mI(x) \fto C^\mIp(x')  && \textrm{(by \eqref{eq: KDNSJ 2a}, \eqref{eq: KDNSJ 1} and \eqref{fop: HSDJW 2})} \\
		        & \leq & C^\mI(x) \fequiv C^\mIp(x') && \textrm{(since $C^\mI(x) > C^\mIp(x')$).}
		\end{array}
		\]

		\item Case $O \in \Phi$ and $C = \{a\}$: By~\eqref{eq: FB 5}, $Z(x,x') \leq (C^\mI(x) \fequiv C^\mIp(x'))$. 
		\item Case $\Self \in \Phi$ and $C = \E r.\Self$: By~\eqref{eq: FB 10}, $Z(x,x') \leq (C^\mI(x) \fequiv C^\mIp(x'))$. 
		
		\item Case $Q_n \in \Phi$ and $C = (\geq\!n\,R.D)$: The case $Z(x,x') = 0$ is trivial. So, assume that $Z(x,x') > 0$. If $C^\mIp(x') > 0$, then: 
		\begin{itemize}
			\item there exist pairwise distinct $y'_1,\ldots,y'_n \in \Delta^\mIp$ such that $R^\mIp(x',y'_j) > 0$ and $D^\mIp(y'_j) > 0$ for all $1 \leq j \leq n$;
			\item by~\eqref{eq: FB 7}, there exist pairwise distinct $y_1,\ldots,y_n \in \Delta^\mI$ such that, for every $1 \leq i \leq n$, $R^\mI(x,y_i) > 0$ and $Z(y_i,y'_{j_i}) > 0$ for some $1 \leq j_i \leq n$;
			\item for every $1 \leq i \leq n$, since $Z(y_i,y'_{j_i}) > 0$ and $D^\mIp(y'_{j_i}) > 0$, by the inductive assumption of~\eqref{eq: GDHAW 1}, $D^\mI(y_i) > 0$; 
			\item therefore, $C^\mI(x) > 0$.
		\end{itemize}
		Hence, if $C^\mI(x)=0$, then $C^\mIp(x')=0$ and $Z(x,x') \leq (C^\mI(x) \fequiv C^\mIp(x'))$. So, assume that $C^\mI(x) > 0$. 
		Since $\mI$ is witnessed w.r.t.\ \mLP, there exist pairwise distinct $y_1$, \ldots, $y_n \in \Delta^\mI$ such that 
		\begin{equation}
		C^\mI(x) = \textstyle\bigotimes\{R^\mI(x,y_i) \fand D^\mI(y_i) \mid 1 \leq i \leq n\}. \label{eq: HFMLA 1}
		\end{equation}
		By~\eqref{eq: FB 6}, there exist pairwise distinct elements $y'_1,\ldots,y'_n$ of $\Delta^\mIp$ such that, for every $1 \leq i \leq n$, there exists $1 \leq j_i \leq n$ such that 
		\[ Z(x,x') \fand R^\mI(x,y_1)  \fand\cdots\fand R^\mI(x,y_n) \leq Z(y_{j_i},y'_i) \fand R^\mIp(x',y'_i). \]
		Therefore, 
		\begin{equation}
		Z(x,x') \fand R^\mI(x,y_1)  \fand\cdots\fand R^\mI(x,y_n) \leq \textstyle\bigotimes\{Z(y_{j_i},y'_i) \fand R^\mIp(x',y'_i) \mid 1 \leq i \leq n\}.  \label{eq: HFMLA 2}
		\end{equation}
		By definition, 
		\begin{equation}
		\textstyle\bigotimes\{R^\mIp(x',y'_i) \fand D^\mIp(y'_i) \mid 1 \leq i \leq n\} \leq C^\mIp(x').  \label{eq: HFMLA 3}
		\end{equation}
		By the inductive assumption of~\eqref{eq: GDHAW 1}, for every $1 \leq i \leq n$, 
		\begin{equation}
		Z(y_{j_i},y'_i) \leq (D^\mI(y_{j_i}) \fequiv D^\mIp(y'_i)).  \label{eq: HFMLA 4}
		\end{equation}
		Let $p = R^\mI(x,y_1) \fand\cdots\fand R^\mI(x,y_n)$ and $p' = R^\mIp(x',y'_1) \fand\cdots\fand R^\mIp(x',y'_n)$.
		We have the following:
		\[
		\begin{array}{lcll}
		& & Z(x,x') \\
		& \!\!\leq\!\! & p \fto \textstyle\bigotimes\{Z(y_{j_i},y'_i) \fand R^\mIp(x',y'_i) \mid 1 \leq i \leq n\} & \textrm{(by \eqref{eq: HFMLA 2} and \eqref{fop: HSDJW 3})} \\
		& \!\!\leq\!\! & p \fto p' \fand \textstyle\bigotimes\{D^\mI(y_{j_i}) \fequiv D^\mIp(y'_i) \mid 1 \leq i \leq n\} & \textrm{(by \eqref{eq: HFMLA 4} and \eqref{fop: HSDJW 2})} \\
		& \!\!\leq\!\! & p \fto p' \fand (D^\mI(y_{j_1}) \fand\cdots\fand D^\mI(y_{j_n}) \fequiv
		                          D^\mIp(y'_1) \fand\cdots\fand D^\mIp(y'_n)) & 
			\textrm{(by \eqref{fop: HSDJW 6}, \eqref{fop: HSDJW 1} and \eqref{fop: HSDJW 2})} \\
		& \!\!\leq\!\! & p \fto (D^\mI(y_{j_1}) \fand\cdots\fand D^\mI(y_{j_n}) \fequiv
						 p' \fand D^\mIp(y'_1) \fand\cdots\fand D^\mIp(y'_n)) & \textrm{(by \eqref{fop: HSDJW 4} and \eqref{fop: HSDJW 2})} \\
		& \!\!\leq\!\! & p \fand D^\mI(y_{j_1}) \fand\cdots\fand D^\mI(y_{j_n}) \fto
			     p' \fand D^\mIp(y'_1) \fand\cdots\fand D^\mIp(y'_n) & \textrm{(by \eqref{fop: HSDJW 5})} \\
		& \!\!\leq\!\! & p \fand D^\mI(y_1) \fand\cdots\fand D^\mI(y_n) \fto
			     p' \fand D^\mIp(y'_1) \fand\cdots\fand D^\mIp(y'_n) & \textrm{(by \eqref{fop: HSDJW 2})} \\
		& \!\!\leq\!\! & C^\mI(x) \fto C^\mIp(x')  & \textrm{(by \eqref{eq: HFMLA 1}, \eqref{eq: HFMLA 3} and \eqref{fop: HSDJW 2})} \\
		& \!\!\leq\!\! & C^\mI(x) \fequiv C^\mIp(x') & \textrm{(since $C^\mI(x) > C^\mIp(x')$).}
		\end{array}
		\]
		\item Case $N_n \in \Phi$ and $C = (\geq\!n\,R)$: The proof for this case is obtained from the proof of the previous case by: 
		replacing $D$, $Z(y_i,y'_{j_i})$ and $Z(y_{j_i},y'_i)$ with $1$, 
		replacing \eqref{eq: FB 6} with \eqref{eq: FB 6n}, 
		replacing \eqref{eq: FB 7} with \eqref{eq: FB 7n}, 
		and then simplifying the text appropriately.
		
		\item Case $Q_n \in \Phi$ and $C = (<\!n\,R.D)$: For a contradiction, suppose $Z(x,x') > 0$. Since $C^\mI(x) > C^\mIp(x')$, by Remark~\ref{remark: JFLWB}, $C^\mI(x) = 1$ and $C^\mIp(x') = 0$. Since $C^\mIp(x') = 0$, there exist pairwise distinct elements $y'_1,\ldots,y'_n \in \Delta^\mIp$ such that $R^\mIp(x',y'_i) \fand D^\mIp(y'_i) > 0$ for all $1 \leq i \leq n$. 
		By~\eqref{eq: FB 7}, there exist pairwise distinct elements $y_1,\ldots,y_n \in \Delta^\mI$ such that, for every $1 \leq i \leq n$, there exists $1 \leq j_i \leq n$ such that 
		\[ Z(x,x') \fand R^\mIp(x',y'_1) \fand\cdots\fand R^\mIp(x',y'_n) \leq Z(y_i,y'_{j_i}) \fand R^\mI(x,y_i). \]
		By the inductive assumption of~\eqref{eq: GDHAW 1}, for every $1 \leq i \leq n$, 
		\[ Z(y_i,y'_{j_i}) \leq (D^\mI(y_i) \fequiv D^\mIp(y'_{j_i})). \]
		Observe that, for every $1 \leq i \leq n$, $R^\mIp(x',y'_i)$, $D^\mIp(y'_i)$, $Z(y_i,y'_{j_i})$, $R^\mI(x,y_i)$ and $D^\mI(y_i)$ are all greater than 0. Hence, $(<\!n\,R.D)^\mI = 0$, which contradicts $C^\mI(x) = 1$. Therefore, 
		\[ Z(x,x') = 0 \leq (C^\mI(x) \fequiv C^\mIp(x')). \] 
		
		\item Case $N_n \in \Phi$ and $C = (<\!n\,R)$: The proof for this case is obtained from the proof of the previous case by: 
		replacing $D$ and $Z(y_i,y'_{j_i})$ with $1$, 
		replacing \eqref{eq: FB 7} with \eqref{eq: FB 7n}, 
		and then simplifying the text appropriately.
		\myend
	\end{itemize}
\end{proof}


\medskip

\noindent
\textbf{Theorem~\ref{theorem: fG H-M}.}
{\em \TextTheoremfGHM}

\begin{proof}
	By Lemma~\ref{lemma: GDHAW2}, it is sufficient to \modifiedA{prove} that $Z$ is a fuzzy $\Phi$-bisimulation between $\mI$ and~$\mI'$. Let $x \in \Delta^\mI$, $x' \in \Delta^\mIp$, $A \in \CN$, $a \in \IN$, $r \in \RN$ and let $R$ be a basic role w.r.t.~$\Phi$. We prove Conditions~\eqref{eq: FB 2}--\eqFBlast (under the corresponding assumptions about~$\Phi$). 
	\begin{itemize}
		\item Condition~\eqref{eq: FB 2} directly follows from the definition of~$Z$. 
		
		\item Consider Condition~\eqref{eq: FB 3} and let $y \in \Delta^\mI$. 
		Let $p = Z(x,x') \fand R^\mI(x,y)$ and $Y' = \{y' \in \Delta^\mIp \mid R^\mIp(x',y') \geq p\}$. 
		Without loss of generality, assume that $p > 0$. 
		Observe that $Y' \neq \emptyset$, because otherwise we would not have that $Z(x,x') \leq ((\E R.\top)^\mI \fequiv (\E R.\top)^\mIp)$ (we use here the assumption that $\mIp$ is witnessed w.r.t.~\mLPp). 
		We prove that there exists $y' \in Y'$ such that $Z(y,y') \geq p$. For a contradiction, suppose that, for every $y' \in Y'$, \mbox{$Z(y,y') < p$}, which means there exists a concept $C_{y'}$ of \mLPp such that 
		\mbox{$(C_{y'}^\mI(y) \fequiv C_{y'}^\mIp(y')) < p$}.
		For every $y' \in Y'$, let \modifiedA{$D_{y'} = (C_{y'} \to C_{y'}^\mI(y)) \mand (C_{y'}^\mI(y) \to C_{y'})$}. Let $\Gamma = \{D_{y'} \mid y' \in Y'\}$. 
		Observe that, for every $y' \in Y'$, $D_{y'}^\mI(y) = 1$ and $D_{y'}^\mIp(y') < p$. 
		Thus, for every $y' \in \Delta^\mIp$, there exists $D \in \Gamma$ such that $R^\mIp(x',y') \fand D^\mIp(y') < p$. Since $\mIp$ is modally saturated w.r.t.\ \mLPp, there exists a finite subset $\Lambda$ of $\Gamma$ such that, for every $y' \in \Delta^\mIp$, there exists $D \in \Lambda$ such that $R^\mIp(x',y') \fand D^\mIp(y') < p$. Let $C = \E R.\bigsqcap\Lambda$. We have $C^\mI(x) \geq p$ (since $(\bigsqcap\Lambda)^\mI(y) = 1$) and $C^\mIp(x') < p$ (since $\mIp$ is witnessed w.r.t.\ \mLPp). This contradicts \mbox{$p \leq Z(x,x') \leq (C^\mI(x) \fequiv C^\mIp(x'))$}.  
		
		\item Condition~\eqref{eq: FB 4} can be proved analogously as for Condition~\eqref{eq: FB 3}.
		
		\item Consider Condition~\eqref{eq: FB 5} when $O \in \Phi$. Since $Z(x,x') \leq (C^\mI(x) \fequiv C^\mIp(x'))$ for $C = \{a\}$, Condition~\eqref{eq: FB 5} clearly holds. 
		
		\item Consider Condition~\eqref{eq: FB 8} when $U \in \Phi$. Let $y \in \Delta^\mI$ and $p = Z(x,x')$. For a contradiction, suppose that, for every $y' \in \Delta^\mIp$, $Z(y,y') < p$, which means that there exists a concept $C_{y'}$ of \mLPp such that $(C_{y'}^\mI(y) \fequiv C_{y'}^\mIp(y')) < p$. 
		For every $y' \in \Delta^\mIp$, let $D_{y'} = (C_{y'} \to C_{y'}^\mI(y))$ if $C_{y'}^\mI(y) < \min\{p,C_{y'}^\mIp(y')\}$, and $D_{y'} = (C_{y'}^\mI(y) \to C_{y'})$ otherwise (i.e., when $C_{y'}^\mIp(y') < \min\{p,C_{y'}^\mI(y)\}$). Let $\Gamma = \{D_{y'} \mid y' \in \Delta^\mIp\}$. 
		Observe that, for every $D_{y'} \in \Gamma$, $D_{y'}^\mI(y) = 1$ and $D_{y'}^\mIp(y') < p$. 
		Thus, for every $y' \in \Delta^\mIp$, there exists $D \in \Gamma$ such that $D^\mIp(y') < p$. 
		Since $\mIp$ is modally saturated w.r.t.\ \mLPp, there exists a finite subset $\Lambda$ of $\Gamma$ such that, for every $y' \in \Delta^\mIp$, there exists $D \in \Lambda$ such that $D^\mIp(y') < p$. Let $C = \E U.\bigsqcap\Lambda$. We have $C^\mI(x) = 1 \geq p$ (since $(\bigsqcap\Lambda)^\mI(y) = 1$) and $C^\mIp(x') < p$ (since $\mIp$ is witnessed w.r.t.\ \mLPp). This contradicts \mbox{$p = Z(x,x') \leq (C^\mI(x) \fequiv C^\mIp(x'))$}.  
		
		\item Condition~\eqref{eq: FB 9} for the case when $U \in \Phi$ can be proved analogously as for Condition~\eqref{eq: FB 8}. 
		
		\item Consider Condition~\eqref{eq: FB 10} when $\Self \in \Phi$. Since $Z(x,x') \leq (C^\mI(x) \fequiv C^\mIp(x'))$ for $C = \E r.\Self$, Condition~\eqref{eq: FB 10} clearly holds. 

		\item Consider Condition~\eqref{eq: FB 6} when $Q_n \in \Phi$. 
		Suppose \mbox{$Z(x,x') > 0$} and let $y_1$, \ldots, $y_n$ be pairwise distinct elements of $\Delta^\mI$ such that $R^\mI(x,y_i) > 0$ for all $1 \leq i \leq n$.  
		Let $p = Z(x,x') \fand R^\mI(x,y_1) \fand\cdots\fand R^\mI(x,y_n)$. We have that $p > 0$. 
		Let $Y' = \{y'\in \Delta^\mIp \mid$ there exists $1 \leq i \leq n$ such that \mbox{$p \leq Z(y_i,y') \fand R^\mIp(x',y')\}$}. We need to prove that $\# Y' \geq n$. 
		For every \mbox{$y' \in \Delta^\mIp \setminus Y'$} and every $1 \leq i \leq n$, we have $Z(y_i,y') \fand R^\mIp(x',y') < p$, hence there exists a concept $C_{y',i}$ of \mLPp such that 
		\[ (C_{y',i}^\mI(y_i) \fequiv C_{y',i}^\mIp(y')) \fand R^\mIp(x',y') < p. \] 
		For $y' \in \Delta^\mIp \setminus Y'$ and $1 \leq i \leq n$, let $D_{y',i}$ be:
		\begin{itemize}
			\item 1 if $R^\mIp(x',y') < p$, 
			\item \modifiedA{$(C_{y',i} \to C_{y',i}^\mI(y_i)) \mand (C_{y',i}^\mI(y_i) \to C_{y',i})$ otherwise}.
		\end{itemize}
		With such $y'$ and $i$, we have that $D_{y',i}^\mI(y_i) = 1$ and $D_{y',i}^\mIp(y') \fand R^\mIp(x',y') < p$. 
		Let \mbox{$C_{y'} = D_{y',1} \mor\ldots\mor D_{y',n}$} for $y' \in \Delta^\mIp \setminus Y'$. 
		By Remark~\ref{remark: OFHSJ}, $C_{y'}$ is equivalent to a concept of \mLPp. 
		We have that $Z(x,x') \geq p$ and $C_{y'}^\mI(y_i) = 1$ for all $y' \in \Delta^\mIp \setminus Y'$ and $1 \leq i \leq n$. 
		Furthermore, $R^\mIp(x',y') \fand C_{y'}^\mIp(y') < p$ for all $y' \in \Delta^\mIp \setminus Y'$. 
		Let $\Gamma = \{C_{y'} \mid y' \in \Delta^\mIp \setminus Y'\}$. 
		Consider any finite subset $\Lambda$ of $\Gamma$. Since $R^\mI(x,y_i) \geq p$ and $(\bigsqcap\Lambda)^\mI(y_i) = 1$ for all $1 \leq i \leq n$, we have \mbox{$(\geq\!n\,R.\bigsqcap\Lambda)^\mI(x) \geq p$}. Since \mbox{$p \leq Z(x,x')$}, it follows that \mbox{$(\geq\!n\,R.\bigsqcap\Lambda)^\mIp(x') \geq p$}. 
		Since $\mIp$ is witnessed w.r.t.\ \mLPp, there are pairwise distinct $y'_1,\ldots,y'_n \in \Delta^\mIp$ such that $R^\mIp(x',y'_i) \fand C^\mIp(y'_i) \geq p$ for all $1 \leq i \leq n$ and $C \in \Lambda$. Since $\mIp$ is modally saturated w.r.t.\ \mLPp, it follows that there are pairwise distinct $y'_1,\ldots,y'_n \in \Delta^\mIp$ such that $R^\mIp(x',y'_i) \fand C^\mIp(y'_i) \geq p$ for all $1 \leq i \leq n$ and $C \in \Gamma$. Recall that $R^\mIp(x',y') \fand C_{y'}^\mIp(y') < p$ and $C_{y'} \in \Gamma$ for all $y' \in \Delta^\mIp \setminus Y'$. Hence, $\# Y' \geq n$.  
		
		\item Condition~\eqref{eq: FB 7} for the case when $Q_n \in \Phi$ can be proved analogously as for Condition~\eqref{eq: FB 6}.
		
		\item Consider Condition~\eqref{eq: FB 6n} when $N_n \in \Phi$. 
		Suppose \mbox{$Z(x,x') > 0$} and let $y_1$, \ldots, $y_n$ be pairwise distinct elements of $\Delta^\mI$ such that $R^\mI(x,y_i) > 0$ for all $1 \leq i \leq n$.  
		Let $p = Z(x,x') \fand R^\mI(x,y_1) \fand\cdots\fand R^\mI(x,y_n)$. 
		We have that $0 < p \leq Z(x,x')$. 
		Let $Y' = \{y'\in \Delta^\mIp \mid p \leq R^\mIp(x',y')\}$. We need to prove that $\# Y' \geq n$. Let $C = (\geq\!n\,R)$. We have 
		\[ p \leq Z(x,x') \leq (C^\mI(x) \fequiv C^\mIp(x')). \]
		Since $C^\mI(x) \geq p$, it follows that $C^\mIp(x') \geq p$. Hence, $\# Y' \geq n$.
		
		\item Condition~\eqref{eq: FB 7n} for the case when $N_n \in \Phi$ can be proved analogously as for Condition~\eqref{eq: FB 6n}.
		\myend
	\end{itemize}
\end{proof}


\medskip

\noindent
\textbf{Theorem~\ref{theorem: fG H-M 2}.}
{\em \TextTheoremfGHMt}

\begin{proof}
	By Lemma~\ref{lemma: cGDHAW2}, it is sufficient to proved that $Z$ is a crisp $\Phi$-bisimulation between $\mI$ and~$\mI'$. Let $x \in \Delta^\mI$, $x' \in \Delta^\mIp$, $A \in \CN$, $a \in \IN$, $r \in \RN$ and let $R$ be a basic role w.r.t.~$\Phi$. We prove Conditions~\eqref{eq: FB 2}--\eqFBlast (under the corresponding assumptions about~$\Phi$). 
	\begin{itemize}
		\item Consider Condition~\eqref{eq: FB 3} and let $y \in \Delta^\mI$. Without loss of generality assume that $Z(x,x') = 1$ and $R^\mI(x,y) = p > 0$. Let $Y' = \{y' \in \Delta^\mIp \mid R^\mIp(x',y') \geq p\}$. Observe that $Y' \neq \emptyset$, because otherwise we would have that $(\E R.\top)^\mI \geq p > (\E R.\top)^\mIp$ (since $\mIp$ is witnessed w.r.t.~\DLPp), which contradicts $Z(x,x') = 1$. We prove that there exists $y' \in Y'$ such that $Z(y,y') = 1$. For a contradiction, suppose that, for every $y' \in Y'$, $Z(y,y') = 0$, which means there exists a concept $C_{y'}$ of \DLPp such that 
		\( C_{y'}^\mI(y) \neq C_{y'}^\mIp(y'). \)
		For every $y' \in Y'$, let \modifiedA{$D_{y'} = \triangle((C_{y'} \to C_{y'}^\mI(y)) \mand (C_{y'}^\mI(y) \to C_{y'}))$}. Let $\Gamma = \{D_{y'} \mid y' \in Y'\}$. Observe that, for every $y' \in Y'$, $D_{y'}^\mI(y) = 1$ and $D_{y'}^\mIp(y') = 0$. 
		Thus, for every $y' \in \Delta^\mIp$ with $R^\mIp(x',y') \geq p$, there exists $D \in \Gamma$ such that $D^\mIp(y') = 0$. Since $\mIp$ is modally saturated w.r.t.\ \DLPp, there exists a finite subset $\Lambda$ of $\Gamma$ such that, for every $y' \in \Delta^\mIp$, there exists $D \in \Lambda$ such that $R^\mIp(x',y') \fand D^\mIp(y') < p$. Let $C = \E R.\bigsqcap\Lambda$. We have $C^\mI(x) \geq p > 0$ (since $(\bigsqcap\Lambda)^\mI(y) = 1$) and $C^\mIp(x') < p$ (since $\mIp$ is witnessed w.r.t.\ \DLPp). This contradicts $Z(x,x') = 1$.  
		
		\item Consider Condition~\eqref{eq: FB 8} when $U \in \Phi$. Without loss of generality, assume that $Z(x,x') = 1$. Let $y \in \Delta^\mI$. For a contradiction, suppose that, for every $y' \in \Delta^\mIp$, $Z(y,y') = 0$, which means that there exists a concept $C_{y'}$ of \DLPp such that $C_{y'}^\mI(y) \neq C_{y'}^\mIp(y')$. 
		For every $y' \in \Delta^\mIp$, let \modifiedA{$D_{y'} = \triangle((C_{y'} \to C_{y'}^\mI(y)) \mand  (C_{y'}^\mI(y) \to C_{y'}))$}. Let $\Gamma = \{D_{y'} \mid y' \in \Delta^\mIp\}$. 
		Observe that, for every $D_{y'} \in \Gamma$, $D_{y'}^\mI(y) = 1$ and $D_{y'}^\mIp(y') = 0$. 
		Thus, for every $y' \in \Delta^\mIp$, there exists $D \in \Gamma$ such that $D^\mIp(y') = 0$. 
		Since $\mIp$ is modally saturated w.r.t.\ \DLPp, there exists a finite subset $\Lambda$ of $\Gamma$ such that, for every $y' \in \Delta^\mIp$, there exists $D \in \Lambda$ such that $D^\mIp(y') < 1$. Let $C = \E U.\bigsqcap\Lambda$. We have $C^\mI(x) = 1$ (since $(\bigsqcap\Lambda)^\mI(y) = 1$) and $C^\mIp(x') < 1$ (since $\mIp$ is witnessed w.r.t.\ \DLPp). This contradicts \mbox{$Z(x,x')= 1$}.  
		
		\item Consider Condition~\eqref{eq: FB 6} when $Q_n \in \Phi$. 
		Suppose \mbox{$Z(x,x') > 0$}, i.e.\ $Z(x,x') = 1$, and let $y_1$, \ldots, $y_n$ be pairwise distinct elements of $\Delta^\mI$ such that $R^\mI(x,y_i) > 0$ for all $1 \leq i \leq n$. Let $p = R^\mI(x,y_1) \fand\cdots\fand R^\mI(x,y_n)$. We have that $p > 0$. 
		Let $Y' = \{y'\in \Delta^\mIp \mid$ $R^\mIp(x',y') \geq p$ and there exists $1 \leq i \leq n$ such that $Z(y_i,y') = 1\}$. We need to prove that $\# Y' \geq n$. 
		For every \mbox{$y' \in \Delta^\mIp \setminus Y'$}, either $R^\mIp(x',y') < p$ or $Z(y_i,y') = 0$ for all $1 \leq i \leq n$. Therefore, for every \mbox{$y' \in \Delta^\mIp \setminus Y'$} with $R^\mIp(x',y') \geq p$ and for every $1 \leq i \leq n$, there exists a concept $C_{y',i}$ of \DLPp such that $C_{y',i}^\mI(y_i) \neq C_{y',i}^\mIp(y')$. 
		For $y' \in \Delta^\mIp \setminus Y'$ and $1 \leq i \leq n$, let $D_{y',i}$ be:
		\begin{itemize}
			\item 1 if $R^\mIp(x',y') < p$, 
			\item \modifiedA{$\triangle((C_{y',i} \to C_{y',i}^\mI(y_i)) \mand (C_{y',i}^\mI(y_i) \to C_{y',i}))$ otherwise}.
		\end{itemize}
		With such $y'$ and $i$, we have that $D_{y',i}^\mI(y_i) = 1$ and, if $R^\mIp(x',y') \geq p$, then $D_{y',i}^\mIp(y') = 0$. 
		Let \mbox{$C_{y'} = D_{y',1} \mor\ldots\mor D_{y',n}$} for $y' \in \Delta^\mIp \setminus Y'$. 
		By Remark~\ref{remark: OFHSJ}, $C_{y'}$ is equivalent to a concept of \DLPp. 
		We have that $C_{y'}^\mI(y_i) = 1$ for all $y' \in \Delta^\mIp \setminus Y'$ and $1 \leq i \leq n$. 
		Furthermore, for all $y' \in \Delta^\mIp \setminus Y'$, if $R^\mIp(x',y') \geq p$, then $C_{y'}^\mIp(y') = 0$. Let $\Gamma = \{C_{y'} \mid y' \in \Delta^\mIp \setminus Y'\}$. 
		Consider any finite subset $\Lambda$ of $\Gamma$. Since $R^\mI(x,y_i) \geq p$ and $(\bigsqcap\Lambda)^\mI(y_i) = 1$ for all $1 \leq i \leq n$, we have \mbox{$(\geq\!n\,R.\bigsqcap\Lambda)^\mI(x) \geq p$}. Since \mbox{$Z(x,x') = 1$}, it follows that \mbox{$(\geq\!n\,R.\bigsqcap\Lambda)^\mIp(x') \geq p$}. 
		Since $\mIp$ is witnessed w.r.t.\ \DLPp, there are pairwise distinct $y'_1,\ldots,y'_n \in \Delta^\mIp$ such that $R^\mIp(x',y'_i) \fand C^\mIp(y'_i) \geq p$ for all $1 \leq i \leq n$ and $C \in \Lambda$. Since $\mIp$ is modally saturated w.r.t.\ \DLPp, it follows that there are pairwise distinct $y'_1,\ldots,y'_n \in \Delta^\mIp$ such that $R^\mIp(x',y'_i) \geq p$ and $C^\mIp(y'_i) > 0$ for all $1 \leq i \leq n$ and $C \in \Gamma$. Recall that, for all $y' \in \Delta^\mIp \setminus Y'$, if $R^\mIp(x',y') \geq p$, then $C_{y'}^\mIp(y') = 0$. Hence, $\# Y' \geq n$.  
	\end{itemize}	
	The proofs concerning the other conditions among \eqref{eq: FB 2}--\eqFBlast \modifiedA{are similar to the ones of the above considered conditions or the corresponding ones of} Theorem~\ref{theorem: fG H-M}. 
	\myend
\end{proof}

\end{document}